\documentclass[11pt]{article}
\usepackage{fullpage}

\usepackage{amsfonts,amsmath,amssymb,boxedminipage,color,url,nccmath}

\usepackage[normalem]{ulem}
 
\usepackage{enumerate,paralist}
\usepackage[hypcap=false]{caption}
\usepackage[T1]{fontenc}
\usepackage{lmodern} \normalfont 
\DeclareFontShape{T1}{lmr}{bx}{sc} { <-> ssub * cmr/bx/sc }{}
\DeclareFontShape{T1}{lmr}{m}{scit}{<->ssub*lmr/m/scsl}{}%

\usepackage[bottom]{footmisc} 

\usepackage[pageanchor=false,pdfstartview=FitH,colorlinks,linkcolor=blue,filecolor=blue,citecolor=blue,urlcolor=blue]{hyperref} 

\usepackage{amsthm} 

\usepackage{aliascnt}
\usepackage[numbers]{natbib} 
\usepackage{cleveref}
\usepackage{xspace}
\usepackage{xstring}

\usepackage[nottoc]{tocbibind} 

\usepackage{tikz}

\newcommand{\remove}[1]{}
\newcommand{\Draft}[1]{\ifdefined\IsDraft\texttt{ #1} \fi}

\newcommand{\TLLNCS}[2]{\ifdefined\IsLLNCS#1\else #2 \fi}

\ifdefined\IsDraft
    \newcommand{\authnote}[2]{\textbf{[{\color{red} #1's Note:} {\color{blue} #2}]}}
\else
    \newcommand{\authnote}[2]{}
\fi

\ifdefined\IsTrackChanges
\newcommand{\stkout}[1]{\ifmmode\text{\sout{\ensuremath{#1}}}\else\sout{#1}\fi}
\newcommand{\changed}[3]{\textbf{Changed:} [{\color{red} \stkout{#2}}] {\color{blue} #3}}
\newcommand{\deleted}[2]{{\textbf{Deleted:}~{\color{red} \stkout{#2} }}}
\newcommand{\added}[2]{{\textbf{Added:}~{\color{blue} #2}}}
\else
\newcommand{\changed}[3]{#3}
\newcommand{\deleted}[2]{}
\newcommand{\added}[2]{#2}
\fi

\newcommand{\sdotfill}{\textcolor[rgb]{0.8,0.8,0.8}{\dotfill}} 

\newenvironment{protocol}{\begin{proto}}{\vspace{-\topsep}\sdotfill\end{proto}}
\newenvironment{algorithm}{\begin{algo}}{\vspace{-\topsep}\sdotfill\end{algo}}

\newcommand{\Ensuremath}[1]{\ensuremath{#1}\xspace}
\newcommand{\MathAlg}[1]{\mathsf{#1}}

\newcommand{\MathAlgX}[1]{\Ensuremath{\MathAlg{#1}}}

\newcommand \mycaption {\small }     
\newcommand \mylabel {}

\ifdefined\IsLLNCS
    \newenvironment{nfbox}[3]{
    \renewcommand \mycaption {#1}
    \renewcommand \mylabel {#2}
    \begin{center}\small
    \begin{tabular}{|ll|}
    \hline
    \hspace{.3ex}
    \begin{minipage}{.97\linewidth}
         \vspace{0.5ex}
         #3}
         {\smallskip
         \captionof{figure}{\mycaption}
         \label{\mylabel}
     \end{minipage}
     &\hspace{.3ex} \\
     \hline
     \end{tabular}
     \end{center}    
    }
\else
    \newenvironment{nfbox}[3]{
    \renewcommand \mycaption {#1}
    \renewcommand \mylabel {#2}
    \begin{center}\small
    \begin{tabular}{|ll|}
    \hline
    \hspace{.3ex}
    \begin{minipage}{.94\linewidth}
         \vspace{0.5ex}
         #3}
         {\smallskip
         \captionsetup{type=figure}
     \end{minipage}
     &\hspace{-0ex} \\
     \hline
     \end{tabular}
     \captionof{figure}{\mycaption}
     \label{\mylabel}
     \end{center}    
    }
\fi

\ifdefined\IsLLNCS

\else

\fi

\newcommand{\resp}{resp.,\xspace}
\newcommand{\ie}  {i.e.,\xspace}
\newcommand{\eg}  {e.g.,\xspace}

\newcommand{\wrt} {with respect to\xspace}
\newcommand{\wlg} {without loss of generality\xspace}

\newcommand{\abs}[1]{\left\lvert #1 \right\rvert}
\newcommand{\ceil}[1]{\left\lceil #1 \right\rceil}

\newcommand{\set}[1]{\ens{#1}}
\newcommand{\sset}[1]{\{#1\}}

\newcommand{\floor}[1]{\left \lfloor#1 \right \rfloor}

\newcommand{\ith}           {$i$'th\xspace}
\newcommand{\jth}           {$j$'th\xspace}

\newcommand{\lth}           {$l$'th\xspace}

\newcommand{\half}{\tfrac{1}{2}}

\newcommand{\N}{{\mathbb{N}}}

\newcommand{\zo}{\{0,1\}}

\newcommand{\zs}{{\zo^\ast}}

\newcommand{\xor}{\oplus}

\newcommand{\eps}{\varepsilon}
\newcommand{\ci} {\equiv_c}
\newcommand{\deltaci} {\ci^\delta}
\newcommand{\statclose} {\equiv_s}
\newcommand{\deltaclose} {\statclose^\delta}
\newcommand{\deltaequiv} {\equiv^\delta}

\newcommand{\la}{\gets}

\newcommand{\poly}{\mathsf{poly}}

\newcommand{\logstar}{\operatorname{log^\ast}}

\newcommand{\uglyExp}{{\log(e)\cdot\left(\frac{2}{e}+\frac{1}{\varphi(\secParam)}\right)}}

\newcommand{\Com}{\MathAlg{Com}}

\newcommand{\Recon}{\MathAlgX{Recon}}
\newcommand{\Share}{\MathAlgX{Share}}

\newcommand{\EncGen}{\MathAlgX{Gen}}
\newcommand{\Enc}{\MathAlgX{Enc}}
\newcommand{\Dec}{\MathAlgX{Dec}}

\newcommand{\negl}{\operatorname{neg}}

\newcommand{\dk}{\mathit{dk}}
\newcommand{\ek}{\mathit{ek}}
\newcommand{\vek}{\vect{\ek}}

\newcommand{\decomval}{\rho}
\newcommand{\comval}{c}
\newcommand{\vcomval}{\vect{\comval}}

\newcommand{\encval}{e}
\newcommand{\vencval}{\vect{\encval}}
\newcommand{\rndaugct}{r}
\newcommand{\decomaugct}{\rho}
\newcommand{\comaugct}{\sigma}
\newcommand{\vcomaugct}{\vect{\comaugct}}
\newcommand{\sval}{s}
\newcommand{\vsval}{\vs}
\newcommand{\svalhat}{\hat{\sval}}
\newcommand{\vsvalhat}{\hat{\vect{\sval}}}
\newcommand{\rval}{r}

\renewcommand{\cref}{\Cref}

\TLLNCS{
\newaliascnt{claiml}{theorem}
\newtheorem{claiml}[claiml]{Claim}
\aliascntresetthe{claiml}

\renewenvironment{claim}{\begin{claiml}}{\end{claiml}}
}
{
\newtheorem{theorem}{Theorem}[section]

\newaliascnt{lemma}{theorem}
\newtheorem{lemma}[lemma]{Lemma}
\aliascntresetthe{lemma}

\newaliascnt{claim}{theorem}
\newtheorem{claim}[claim]{Claim}
\aliascntresetthe{claim}

\newaliascnt{corollary}{theorem}
\newtheorem{corollary}[corollary]{Corollary}
\aliascntresetthe{corollary}

\newaliascnt{proposition}{theorem}
\newtheorem{proposition}[proposition]{Proposition}
\aliascntresetthe{proposition}

\newaliascnt{conjecture}{theorem}

\aliascntresetthe{conjecture}

\newaliascnt{definition}{theorem}
\newtheorem{definition}[definition]{Definition}
\aliascntresetthe{definition}

\newaliascnt{remark}{theorem}
\newtheorem{remark}[remark]{Remark}
\aliascntresetthe{remark}

\newaliascnt{example}{theorem}

\aliascntresetthe{example}
}

\crefname{lemma}{Lemma}{Lemmas}
\crefname{figure}{Figure}{Figures}
\crefname{claim}{Claim}{Claims}
\crefname{corollary}{Corollary}{Corollaries}
\crefname{proposition}{Proposition}{Propositions}
\crefname{conjecture}{Conjecture}{Conjectures}
\crefname{definition}{Definition}{Definitions}
\crefname{remark}{Remark}{Remarks}
\crefname{exmaple}{Example}{Examples}

\newaliascnt{construction}{theorem}
\newtheorem{construction}[construction]{Construction}
\aliascntresetthe{construction}
\crefname{construction}{Construction}{Constructions}

\newaliascnt{fact}{theorem}

\aliascntresetthe{fact}
\crefname{fact}{Fact}{Facts}

\newaliascnt{notation}{theorem}

\aliascntresetthe{notation}
\crefname{notation}{Notation}{Notation}

\crefname{equation}{Equation}{Equations}

\newaliascnt{proto}{theorem}

\newtheorem{proto}[proto]{Protocol}

\aliascntresetthe{proto}
\crefname{proto}{protocol}{protocols}

\newaliascnt{algo}{theorem}
\newtheorem{algo}[algo]{Algorithm}
\aliascntresetthe{algo}
\crefname{algo}{algorithm}{algorithms}

\newaliascnt{expr}{theorem}
\newtheorem{expr}[expr]{Experiment}
\aliascntresetthe{expr}
\crefname{experiment}{experiment}{experiments}

\def\FullBox{$\Box$}
\def\qed{\ifmmode\qquad\FullBox\else{\unskip\nobreak\hfil
\penalty50\hskip1em\null\nobreak\hfil\FullBox
\parfillskip=0pt\finalhyphendemerits=0\endgraf}\fi}

\def\qedsketch{\ifmmode\Box\else{\unskip\nobreak\hfil
\penalty50\hskip1em\null\nobreak\hfil$\Box$
\parfillskip=0pt\finalhyphendemerits=0\endgraf}\fi}

\newenvironment{proofsketch}{\begin{trivlist} \item {\it Proof sketch.}} {\qed\end{trivlist}}

\newcommand{\eex}[2]{\Ex_{#1}\left[#2\right]}
\newcommand{\ex}[1]{\Ex\left[#1\right]}
\newcommand{\Ex}{{\mathrm E}}
\renewcommand{\Pr}{{\mathrm {Pr}}}
\newcommand{\pr}[1]{\Pr\left[#1\right]}
\newcommand{\ppr}[2]{\Pr_{#1}\left[#2\right]}

\newcommand{\Ac}{\MathAlgX{A}}

\newcommand{\Bc}{\mathsf{B}}
\newcommand{\Cc}{\mathsf{C}}

\newcommand{\ens}[1]{\left\{#1\right\}}
\newcommand{\size}[1]{\left|#1\right|}
\newcommand{\ssize}[1]{|#1|}

\newcommand{\com}{\operatorname{\mathsf{Com}}}

\newcommand{\naive}{{na{\"i}ve}\xspace}

\newcommand{\Uni}{{\mathord{\mathcal{U}}}}

\newcommand{\prob}[1]{\mathsf{\textsc{#1}}}

\newcommand{\SD}{\prob{SD}}

\newcommand{\cM}{{\cal{M}}}

\newcommand{\I}{\mathcal{I}}
\newcommand{\J}{\mathcal{J}}

\def\state{{\sf state}}
\newcommand{\ppt}{{\sc ppt}\xspace}

\newcommand{\cT}{\mathcal{T}}

\newcommand{\ebv}{\ex{B_\cT \mid V = v}}

\newcommand{\AdvpiI}{\Adv^{\pi}_\cT}
\newcommand{\Advpsi}{\Adv^{\psi}}
\newcommand{\Advpsiprime}{\Adv^{\psi}}
\newcommand{\AdvpiIprime}{{\Adv_{\cT}^{\pi}}}
\newcommand{\Jcoll}{\J_1,\ldots,\J_\numcalls}
\newcommand{\comStar}{\committee^\ast}

\newcommand{\cs}{{\cal{S}}}

\newcommand{\cj}{{\cal{J}}}

\newcommand{\ct}{{\cal{T}}}

\newcommand{\is}{i^\ast}

\newcommand{\js}{j^\ast}
\newcommand{\ls}{l^\ast}

\newcommand{\Tableofcontents}{
\ifdefined\IsLLNCS \else
    \thispagestyle{empty}
    \pagenumbering{gobble}
    \clearpage
    \setcounter{tocdepth}{2}
    \tableofcontents
    \thispagestyle{empty}
    \clearpage
    \pagenumbering{arabic}
\fi
}

\newcommand{\vect}[1]{{ \boldsymbol{#1}}}

\newcommand{\vc}{\vect{c}}
\newcommand{\vf}{\vect{f}}
\newcommand{\vm}{\vect{m}}
\newcommand{\vr}{\vect{r}}
\newcommand{\vs}{\vect{s}}

\newcommand{\vx}{\vect{x}}

\newcommand{\vS}{\vect{S}}

\newcommand{\rinput}{r_\textsf{input}}
\newcommand{\rmask}{r_\textsf{mask}}
\newcommand{\vrmask}{\vr_\textsf{mask}}
\newcommand{\rprot}{r_\textsf{prot}}
\newcommand{\vrprot}{\vr_\textsf{prot}}

\newcommand{\party}[1]{%
    \IfEqCase{#1}{%
        {1}{\Ac}
        {2}{\Bc}
        {3}{\Cc}
    }[\PackageError{\party}{Undefined option to party: #1}{}]%
}%

\newcommand{\Adv}{\Ac} 
\newcommand{\ptp}{{point-to-point}\xspace}
\newcommand{\secParam}{\kappa}

\newcommand{\Party}{\MathAlgX{P}}
\newcommand{\TParty}{\MathAlgX{\tilde P}}

\newcommand{\Sim}{\MathAlgX{S}}
\newcommand{\aux}{z}

\newcommand{\SMbox}[1]{\mbox{\scriptsize {\sc #1}}}
\newcommand{\REAL}{\SMbox{REAL}}
\newcommand{\IDEAL}{\SMbox{IDEAL}}
\newcommand{\HYBRID}{\SMbox{HYBRID}}
\newcommand{\HYB}{\SMbox{HYB}}

\newcommand{\type}{\MathAlgX{type}}

\newcommand{\abort}{\MathAlgX{abort}}
\newcommand{\full}{\MathAlgX{full}}
\newcommand{\fair}{\MathAlgX{fair}}

\newcommand{\idabort}{\MathAlgX{id \mhyphen abort}}

\newcommand{\idfair}{\MathAlgX{id \mhyphen fair}}

\newcommand{\continue}{\MathAlgX{continue}}

\newcommand{\bigbrack}{{\vphantom{2^{2^2}}}}
\mathchardef\mhyphen="2D

\newcommand{\committee}{{\cal{C}}}
\newcommand{\vCS}{\vect{\cal{C}}}

\newcommand{\IS}{{\mathcal{I}}}
\newcommand{\JS}{{\mathcal{J}}}
\newcommand{\MS}{{\cal{M}}}
\newcommand{\PS}{{\cal{P}}}
\newcommand{\ID}{{\cal{J}}}
\newcommand{\OUT}{{\cal{V}}}

\newcommand{\outvalue}{\MathAlgX{out}}

\newcommand{\Comp}{\MathAlgX{Compiler}}
\newcommand{\CompNHMNI}{\MathAlgX{Compiler}_{\mbox{\tiny $\MathAlgX{no\mhyphen in}$}}}
\newcommand{\CompHMNI}{\MathAlgX{Compiler}_{\mbox{\tiny $\MathAlgX{hm,no\mhyphen in}$}}}
\newcommand{\CompNHM}{\MathAlgX{Compiler}_{\mbox{\tiny $\MathAlgX{}$}}}
\newcommand{\CompHM}{\MathAlgX{Compiler}_{\mbox{\tiny $\MathAlgX{hm}$}}}

\newcommand{\CompilerHMNI}[2]{\CompHMNI^{{#1}\text{\tiny$\rightarrow$}{#2}}}
\newcommand{\CompilerNHMNI}[2]{\CompNHMNI^{{#1}\text{\tiny$\rightarrow$}{#2}}}
\newcommand{\CompilerHM}[2]{\CompHM^{{#1}\text{\tiny$\rightarrow$}{#2}}}
\newcommand{\CompilerNHM}[2]{\CompNHM^{{#1}\text{\tiny$\rightarrow$}{#2}}}

\newcommand{\fout}[1]{\MathAlgX{SS_{out}}({#1})}
\newcommand{\foutss}[3]{\MathAlgX{SS_{out}^{({#2},{#3})}}({#1})}
\newcommand{\finout}[4]{\MathAlgX{SS}_{\MathAlgX{in \mhyphen out}}^{{#2} \text{\tiny$\rightarrow$} ({#3},{#4})}({#1})}
\newcommand{\fin}[3]{\MathAlgX{SS}_{\MathAlgX{in}}^{{#2} \text{\tiny$\rightarrow$} {#3}}({#1})}
\newcommand{\faugct}{f_{\MathAlgX{aug\mhyphen ct}}}
\newcommand{\zkmany}{\MathAlgX{ZK}^{\textsc{1:M}}}
\newcommand{\Renc}{R_{\textsf{enc}}}
\newcommand{\fcomor}{f_{\MathAlgX{com \mhyphen or}}}
\newcommand{\felect}{f_{\MathAlgX{elect}}}
\newcommand{\fcf}[1]{f^{{#1}}_{\MathAlgX{cf}}}

\newcommand{\append}{\MathAlgX{extend}}
\newcommand{\numcomm}{\ell}
\newcommand{\numcalls}{s}
\newcommand{\consistent}{\gamma}

\newcommand{\TDP}{{TDP}\xspace}
\newcommand{\CRH}{{CRH}\xspace}

\newcommand{\err}{\textsf{err}}

\def\rnote{\Rnote}
\def\enote{\Enote}

\newcommand{\Inote}[1]{\authnote{Iftach}{#1}}
\newcommand{\Enote}[1]{\authnote{Eran}{#1}}
\newcommand{\Rnote}[1]{\authnote{Ran}{#1}}

\newcommand{\radded}[1]{\added{Ran}{#1}}

\newcommand{\rchanged}[2]{\changed{Ran}{#1}{#2}}

\newcommand{\rdeleted}[1]{\deleted{Ran}{#1}}

\title{From Fairness to Full Security in Multiparty Computation\thanks{A preliminary version of this work appeared at \emph{SCN 2018}~\cite{CHOR18b}.}
\Draft{\\{\small \sc Working Draft: Please Do Not Distribute}}
}

\author{Ran Cohen\thanks{Northeastern University. E-mail: \texttt{rancohen@ccs.neu.edu}. Research supported by Alfred P.\ Sloan Foundation Award 996698, NEU Cybersecurity and Privacy Institute, and NSF TWC-1664445. Most of this work was done while the author was a post-doctoral researcher at Tel Aviv University, supported by ERC starting grant 638121.}
\and Iftach Haitner\thanks{School of Computer Science, Tel Aviv University. E-mail: \texttt{iftachh@cs.tau.ac.il}. Member of the Israeli Center of Research Excellence in Algorithms (ICORE) and the Check Point Institute for Information Security. Research supported by ERC starting grant 638121.}
\and Eran Omri\thanks{Department of Computer Science, Ariel University. Ariel Cyber Innovation Center (ACIC). E-mail: \texttt{omrier@ariel.ac.il}. Research supported by ISF grants 544/13 and 152/17.}
\and Lior Rotem\thanks{School of Computer Science and Engineering, Hebrew University of Jerusalem. E-mail: \texttt{lior.rotem@cs.huji.ac.il}. Supported by the European Union's Horizon 2020 Framework Program (H2020) via an ERC Grant (Grant No.\ 714253) and by the Israel Science Foundation (Grant No.\ 483/13).}
}

\begin{document}

\sloppy
\maketitle

\begin{abstract}
In the setting of secure multiparty computation (MPC), a set of mutually distrusting parties wish to jointly compute a function, while guaranteeing the privacy of their inputs and the correctness of the output.
An MPC protocol is called \emph{fully secure} if no adversary can prevent the honest parties from obtaining their outputs. A protocol is called \emph{fair} if an adversary can prematurely abort the computation, however, only before learning any new information.

We present highly efficient transformations from fair computations to fully secure computations, assuming the fraction of honest parties is constant (e.g., $1\%$ of the parties are honest).
Compared to previous transformations that require linear invocations (in the number of parties) of the fair computation, our transformations require super-logarithmic, and sometimes even super-constant, such invocations.
The main idea is to delegate the computation to chosen random committees that invoke the fair computation. Apart from the benefit of uplifting security, the reduction in the number of parties is also useful, since only committee members are required to work, whereas the remaining parties simply ``listen'' to the computation over a broadcast channel.

One application of these transformations is a new $\delta$-bias coin-flipping protocol, whose round complexity has a super-logarithmic dependency on the number of parties, improving over the protocol of Beimel, Omri, and Orlov (Crypto 2010) that has a linear dependency. A second application is a new fully secure protocol for computing the Boolean OR function, with a super-constant round complexity, improving over the protocol of Gordon and Katz (TCC 2009) whose round complexity is linear in the number of parties.

Finally, we show that our positive results are in a sense optimal, by proving that for some functionalities, a super-constant number of (sequential) invocations of the fair computation is necessary for computing the functionality in a fully secure manner.
\end{abstract}

\vfill
\noindent\textbf{Keywords: multiparty computation; fairness; guaranteed output delivery; identifiable abort, security reductions.}

\Tableofcontents

\section{Introduction}\label{sec:intro}

In the setting of secure multiparty computation (MPC), a set of mutually distrusting parties wish to jointly compute a function of their inputs, while guaranteeing the privacy of their local inputs and the correctness of the output. The security definition of such a computation has numerous variants. A major difference between the variants, which is the focus of this work, is the ability of an adversary to prevent the honest parties from completing the computation by corrupting a subset of the parties. According to the \emph{full-security} variant, an adversary cannot prevent the honest parties from receiving their output.\footnote{This property is also referred to as \emph{guaranteed output delivery}.} A more relaxed security definition called \emph{fairness}, allows an adversary to prematurely abort the computation, but only \emph{before} it has learned any information from the computation. Finally, \emph{security with abort} allows an adversary to prevent the honest parties from receiving the output, even \emph{after} it has learned the output, but never to learn anything more.\footnote{\radded{Throughout the paper, unless explicitly stated otherwise, by security with abort we mean \emph{unanimous} abort where all honest parties reach agreement on whether to abort or not. We note that since we consider a broadcast model, the weaker notion of \emph{non-unanimous} abort~\cite{GL05,CGZ20} in which some honest parties may abort while other receive their output, can be uplifted to unanimous abort in a single broadcast round.}}

A common paradigm for constructing a protocol that provides a high security guarantee (\eg full security) for a given functionality $f$, is to start with constructing a protocol for $f$ of a low security guarantee (\eg security with abort), and then to ``uplift'' the security of the protocol via different generic transformations (\eg the GMW compiler from semi-honest security to malicious security). Hence, finding such security-uplifting transformations is an important research question in the study of MPC. In this work, we study such security-uplifting transformations from security with abort and fairness to full security.

It is known that when the \emph{majority} of the parties are honest, security with abort can be uplifted to fairness.
Given an $n$-party functionality $f$, let $\fout{f}$ denote the functionality that outputs secret shares of $y=f(x_1,\ldots,x_n)$ using an $\ceil{n/2}$-out-of-$n$ error-correcting secret-sharing scheme (ECSS).\footnote{A $(t+1)$-out-of-$n$ secret-sharing scheme is error correcting, if the reconstruction algorithm outputs the correct secret even when up to $t$ shares are arbitrarily modified. ECSS schemes are also known as \emph{robust secret sharing}.} Assume that $\fout{f}$ can be computed securely with abort. In case the adversary aborts the computation of $\fout{f}$, it does not learn any new information, since it can only obtain less than $n/2$ shares. Whereas in case the adversary does not abort, it cannot prevent the honest parties from reconstructing the \emph{correct} output, thus completing the computation.
Similarly, assume $\fout{f}$ can be securely computed with \emph{identifiable abort},\footnote{Same as security with abort, except that upon a premature abort, all honest parties identify a corrupted party.} then the security of computing $\fout{f}$ can be uplifted to a fully secure computation of $f$ via the following player-elimination technique: All parties iteratively compute $\fout{f}$ with identifiable abort, such that in each iteration either all honest parties obtain the output, or the adversary aborts the computation at the cost of revealing the identity of a corrupted party. After at most $t+1$ iterations, it is guaranteed that the computation will successfully complete. Security with identifiable abort can be reduced to security with abort, assuming one-way functions, via a generic reduction~\cite{GoldreichMW87}. More efficient generic reductions in terms of round complexity appear in~\cite{Pass04,IOZ14}, using stronger hardness assumptions.

In case no honest majority is assumed, it is impossible to \emph{generically} transform security with (identifiable) abort to full security, and even not to fairness; every functionality can be computed with abort~\cite{GoldreichMW87} (assuming oblivious transfer exists), but some functionalities cannot be fairly computed~\cite{Cleve86,Mak14}. In contrast, fairness can be uplifted to full security also in the no-honest-majority case \cite{CL17} (assuming one-way functions exist),\footnote{Unless stated otherwise, we assume that parties can communicate over a broadcast channel. If a broadcast channel is not available, identifiable abort cannot be achieved generically~\cite{CL17}, and indeed, some functionalities can be fairly computed, but not with full security~\cite{CL17,CHOR18a}.} by first uplifting the security to fairness with identifiable abort, and then invoking (up to) $t+1$ fair computations of $f$ with identifiable abort.

In the setting of large-scale computation, the linear dependency on number of corruptions forms a bottleneck, and might blow-up the round complexity of the fully secure protocol. In this work, we explore how, and to what extent, this linear dependency can be reduced.

\subsection{Our Results}\label{sec:intro:ourResult}
Our main positive result is highly efficient reductions from full security to fair computation, assuming that the fraction of honest parties is constant (\eg $1\%$ of the parties are honest).
We show how to compute in a fully secure manner an $n$-party functionality $f^n$, by fairly computing a related $n'$-party functionality $f^{n'}$ for $\omega(1)$ sequential times, where $n'=\omega(\log(\secParam))$ (\eg $n' = \logstar(\secParam) \cdot \log(\secParam)$) and $\secParam$ is the security parameter. For some functionalities, we only need to be able to compute the functionality $f^{n'}$ in a security-with-abort manner (no fairness is needed). Throughout, we assume the static-corruption model, where the corrupted parties are determined \emph{before} the protocol begins.

Apart from the obvious benefit of being security-uplifting (from fairness to full security), the reduction in the number of parties is also useful, \ie only $n'=\omega(\log(\secParam))$ parties are required to work in the protocol, whereas the remaining parties simply ``listen'' to the computation over a broadcast channel. The efficiency of secure protocols is typically proportional to the number of parties (in some cases, \eg \cite{BeimelOO15,BHLT17}, the dependency is exponential). Furthermore, for implementations that are only $\delta$-close to being fair (\ie the real-world computation is $\delta$-distinguishable from the ideal-world computation, denoted $\delta$-fair below), the error parameter $\delta$ is typically a function of the number of parties. Hence, even given a fully secure implementation (or $\delta$-close to being fully secure, denoted $\delta$-fully-secure below) of a functionality, applying the above reductions can improve both the security error and the efficiency (see the applications part below for concrete examples).
The reductions presented in this paper are depicted in \cref{fig:results}, alongside previously known reductions.

\begin{figure}[!htb]
\begin{center}
\ifdefined\IsResultWithAbort
\includegraphics[scale=1]{./figures/summary_with_abort.pdf}
\else
\includegraphics[scale=1]{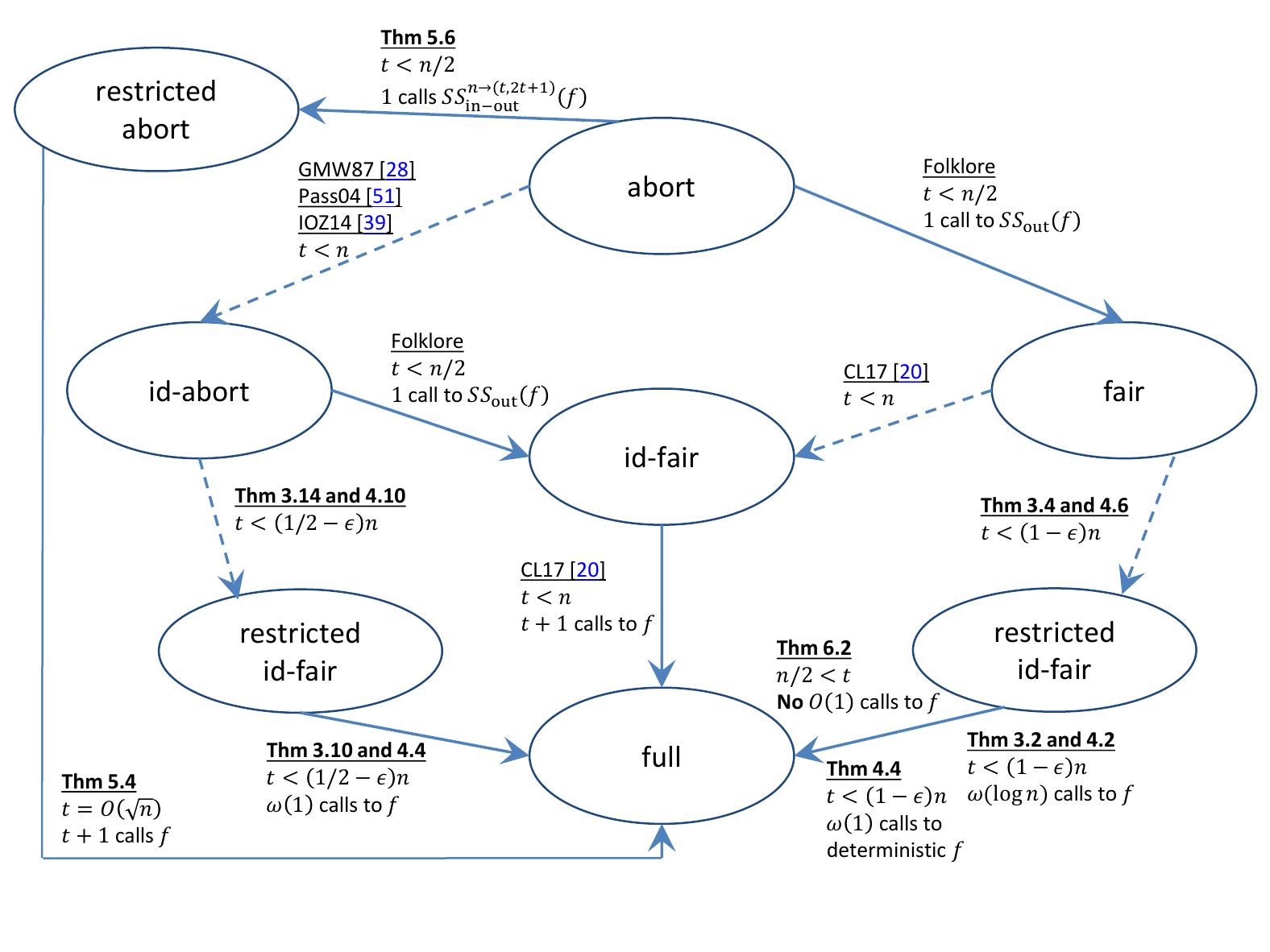}
\fi
\end{center}
\vspace{-6ex}
\caption{Reductions between security notions. Solid arrows refer to black-box reductions \wrt the functionality (\ie a hybrid model) whereas dashed arrows refer to non-black-box reductions (\ie a protocol compiler). \emph{Restricted id-fair} refers to fairness where the set of parties who can abort the computation is restricted to a designated subset.}
\label{fig:results}
\end{figure}

To keep the following introductory discussion simple, we focus below on no-input, public-output functionalities (a single output is given to all parties). A less detailed description of the reductions to security with abort and the reductions for the case of with-input functionalities can be found in \cref{sec:additionalRess}. 
We start by describing the reduction from a fully secure computation of a no-input functionality (\eg coin flipping) to a fair computation of this functionality, and an application of this reduction to fair coin flipping. We then describe a lower bound on the number of rounds in which such a reduction (from fully secure) invokes the fair functionality.

\rchanged{As mentioned above (and elaborated on in Section~\ref{sec:intro:technique}), our}{Our} protocols make use of Feige's lightest-bin protocol for committee election~\cite{Feige99} \radded{(see \cref{sec:intro:technique})}. For integers $n'<n$ and for $0<\beta<\beta'<1$, Feige's protocol is used by $n$ parties, $\beta$ fraction of which are corrupted, to elect a committee of size $n'$, whose fraction of corrupted parties is at most $\beta'$. We denote by $\err(n,n',\beta,\beta')=\frac{n}{n'} \cdot e^{-\frac{(\beta'-\beta)^2 n'}{2(1-\beta)}}$ the error probability of Feige's protocol. Note that for $n'=\omega(\log(\secParam))$ it holds that $\err(n,n',\beta,\beta')$ is negligible (in $\secParam$), \radded{and for $n'=\Omega(\log(\secParam))$ it holds that $\err(n,n',\beta,\beta')$ is inverse-polynomial (the latter is applicable to the $\delta$-bias coin-flipping application in which $\delta$ is inverse-polynomial).}

Our results in the no-honest-majority setting hold under the assumption that \emph{enhanced trapdoor permutations (TDP)} and \emph{collision-resistant hash functions (CRH)} exist. Given a no-input functionality $f$, let $f^n$ denote its $n$-party variant: the output contains $n$ copies of the common output.

\begin{theorem}[fairness to full security, no-input case, informal]\label{thm:intro:mainInlessfair_to_full}
Let $f$ be a no-input functionality, let $n'<n\in \N$, let $0<\beta<\beta'<1$, let $t=\beta n$, let $t'=\beta' n'$, and let $\err=\err(n,n',\beta,\beta')$.
If $f^{n'}$ can be \rchanged{$\delta$}{$\delta'$}-fairly computed by an \rchanged{$r$}{$r'$}-round protocol~$\pi'$ tolerating $t'$ corruptions, then the following hold.
\begin{enumerate}
    \item\label{item::intro:mainInlessfair_to_full_nohm}
    Assuming \TDP and \CRH, $f^n$ can be computed with \textsf{\rchanged{$(t'\cdot\delta+\err)$}{$(t'\cdot\delta'+\err)$}-full-security}, tolerating $t$ corruptions by an \rchanged{$O(t'\cdot r)$}{$O(t'\cdot r')$}-round protocol. Furthermore, if $\pi'$ is \rchanged{$\delta$}{$\delta'$}-fully-secure, then the resulting protocol is \rchanged{$(\delta+\err)$}{$(\delta'+\err)$}-fully-secure and has \rchanged{$O(t'+ r)$}{$O(t'+ r')$} rounds.
    \item
    For $\beta' < 1/2$ and $n'=\varphi(\secParam) \cdot \log(\secParam)$ for $\varphi=\Omega(1)$, if $\pi'$ can be computed $\ell$-times in parallel, for $\ell=\secParam^{c}$ (for some universal constant $c$), then $f^n$ can be computed with \textsf{\rchanged{$(\varphi(\secParam)^2 \cdot \ell\cdot\delta + \err)$}{$(\varphi(\secParam)^2 \cdot \ell\cdot\delta' + \err)$}-full-security}, unconditionally, tolerating $t$ corruptions by an \rchanged{$O(\varphi(\secParam)^2 \cdot r)$}{$O(\varphi(\secParam)^2 \cdot r')$}-round protocol. Furthermore, the computation is black-box in the protocol $\pi'$.\footnote{Following~\cite{IKOS07}, by a black-box access to a protocol we mean a black-box usage of a semi-honest MPC protocol computing its next-message function.}
\end{enumerate}
\end{theorem}

The idea underlying the above reduction is quite simple. To achieve a fully secure computation of an $n$-party functionality $f^n$, we first choose a small committee of size $n'$, using an information-theoretically secure committee-election protocol. The computation is then delegated to this small committee, which in turn, securely computes the functionality with fairness and identifiable abort. Since, the computation of the small committee might abort, we might need to repeat this process several times, while eliminating the aborting parties. See \cref{sec:intro:technique} for more details.

\paragraph{Application to coin flipping.}
As an application of the above type of reduction, we show how to improve on the round complexity of $\delta$-bias coin-flipping protocols. The $n$-party, no-input, public-output, coin-flipping functionality $\fcf{n}$ outputs to all parties a uniformly distributed bit $b\in\zo$. A $\delta$-bias, $t$-secure, $n$-party coin-flipping protocol is a real-world, polynomial-time, $n$-party protocol that emulates the ideal functionality $\fcf{n}$ up to a $\delta$ distinguishing distance, even in the face of up to $t$ corruptions.

\citet{Cleve86} has given a lower bound that relates the bias in any $r$-round coin-flipping protocol to $1/r$. \citet{ABCGM85} constructed an $r$-round, $t$-secure, $O(t/\sqrt{r})$-bias coin-flipping protocol for an arbitrary number of parties $n$ and $t<n$. This was improved by \citet{BeimelOO15}, who gave an $r$-round, $t$-secure, $O(1/\sqrt{r-t})$-bias coin-flipping protocol for the case that $t=\beta n$ for some constant $0<\beta <1$.
Recently, \citet{BHMO18} showed that for a ``large'' number of parties, $n=r^\eps$ (for a constant $\eps>0$), any $r$-round protocol can be efficiently biased by $\tilde{\Omega}(1/\sqrt{r})$.
We remark that $r$-round coin-flipping protocols of bias $o(t/\sqrt{r})$ are known when the number of parties is ``small,''
$n<\log\log{r}$, or when the difference between corrupted and honest parties is constant~\cite{MoranNS09,BeimelOO15,HaitnerT14,AO16,BHLT17}. None of these protocols, however, deals with a large number of parties when a $\beta> 0.51$ fraction of them are malicious.
For this case, it is not known how to obtain a bias that is independent of the number of corruptions.
Using \cref{thm:intro:mainInlessfair_to_full}, we are able to improve upon \cite{BeimelOO15} by replacing the linear dependency on $t$ with a super-logarithmic dependency on the security parameter $\secParam$.
\rnote{Eran -- can you please double check this paragraph}

\begin{corollary}[informal]\label{cor:intro_CF}
Assume that \TDP and \CRH exist.
Let $n'<n$ be integers, and let $0<\beta<\beta'<1$ be constants.
If there exists an $n'$-party, \rchanged{$\delta$}{$\delta'$}-bias, \rchanged{$r$}{$r'$}-round coin-flipping protocol $\pi'$ tolerating $t'=\beta'n'$ corrupted parties, then there exists an $n$-party, \rchanged{$(\delta+\err(n,n',\beta,\beta'))$}{$(\delta'+\err(n,n',\beta,\beta'))$}-bias, \rchanged{$O(t'+r)$}{$O(t'+r')$}-round coin-flipping protocol, tolerating $t=\beta n$ corrupted parties.
\end{corollary}

Concretely, by using the protocol of \citet{BeimelOO15}, we obtain that for every $\varphi(\secParam)=\omega(1)$, every sufficiently large $n$ (greater than $\varphi(\secParam)\cdot\log(\secParam)$), every $0<\beta<1$, and every efficiently computable $r:\N\mapsto\N$ there exists an $n$-party, $r(\secParam)$-round, $O(1/\sqrt{r(\secParam)-\varphi(\secParam)\cdot\log(\secParam)})$-bias, $\beta n$-secure coin-flipping protocol.


\paragraph{Lower bound on the number of sequential fair calls.}
We prove that some functionalities, and in particular coin flipping, achieving full-security requires a super-constant number of \textit{functionality rounds}, \ie rounds in which a fair ideal functionality is invoked, even if a constant fraction of parties are honest. Namely, the (super-)logarithmic multiplicative overhead in the round complexity, \rchanged{appearing in}{induced by} \cref{thm:intro:mainInlessfair_to_full} (Item~\ref{item::intro:mainInlessfair_to_full_nohm}) \radded{for achieving negligible (or inverse-polynomial) error}, cannot be reduced to constant.

The lower bound is proven in a hybrid model in which an ideal computation with fairness and identifiable abort of the functionality is carried out by a trusted party.
For a no-input functionality $f^n$, the model allows different subsets of parties (committees) to invoke the trusted party \emph{in parallel} (in the same functionality round), such that only committee members can abort the call to the trusted party that is made by the committee.
We assume that the outputs of such parallel invocations, which consist of bit-values and/or identities of the aborting parties, are given at the \emph{same time} to all $n$ parties, unless an invocation is made by an all-corrupted committee, which can first see the output of the other parallel invocations \emph{before} deciding upon its action.

The above model is more optimistic than the one we can actually prove to exist, assuming a fair protocol for computing the functionality at hand (hence, proving lower bounds is harder in this model). Actually, the no-honest-majority part of \cref{thm:intro:mainInlessfair_to_full} (Item~\ref{item::intro:mainInlessfair_to_full_nohm}) can be pushed further in this model to match the lower bound given below.
See \cref{sec:imposability:model} for further discussion regarding this model.

\begin{theorem}[necessity of super-constant sequential fair calls, informal]\label{thm:intro_lowerbound}
The following holds in the hybrid model in which any subset of the parties can invoke the trusted party that \textsf{fairly} computes the coin-flipping functionality. Let $\pi$ be a coin-flipping protocol in this model that calls the trusted party in a constant number of rounds (\ie in each round, the trusted party can be invoked many times in parallel by different subsets). Then, for any $1/2<\beta<1$, there exists an efficient fail-stop adversary controlling $\beta n$ parties that noticeably biases the output of the protocol.
\end{theorem}

Note that in this model, fully secure coin-flipping protocols do exist (\eg as we show in \cref{thm:intro:mainInlessfair_to_full}, by invoking the trusted party in a super-logarithmic number of rounds).

\subsection{Our Techniques}\label{sec:intro:technique}

We start with describing the techniques underlying our positive results, focusing on the no-input case for the sake of clarity of the presentation. Later below, we discuss the ideas underlying the lower bound on round complexity.

\paragraph{Upper bound.}
Let $f^n$ be some $n$-party (no-input, public-output) functionality, and let $\pi$ be an $n$-party, $r$-round protocol that computes $f^n$ with fairness, tolerating $t<n$ corruptions. It was shown by~\cite{CL17} that $\pi$ can be compiled into a protocol that computes $f^n$ with fairness and identifiable abort. The original compilation uses the technique of~\cite{GoldreichMW87} and is inefficient in terms of round complexity. However, using the constant-round, bounded-concurrent, zero-knowledge techniques of \citet{Pass04} \radded{(that require TDP and CRH)}, the resulting protocol has $O(r)$ rounds. Having this compilation in mind, we henceforth consider the goal of uplifting fairness with identifiable abort to full security. Let $\pi$ be a protocol that computes $f^n$ with fairness and identifiable abort tolerating $t=\beta\cdot n$ corruptions. A \naive way for achieving full security is using the above mentioned player-elimination technique to obtain a fully secure computation of $f^n$. This, however, comes at a cost in terms of round complexity. Specifically, the resulting protocol will run in $O(t\cdot r)$ rounds.

In the following, we explain how the security-uplifting transformation can be kept efficient in terms of round complexity. Our transformation builds on the player-elimination technique and works given the following three components: (i) a method to select a small subset (committee) $\committee$ of $n'$ parties that contains at most $t'=\beta' \cdot\size{\committee}$ corrupted parties (for arbitrary small $\beta'>\beta$), (ii) an $n'$-party, $r'$-round protocol $\pi'$ that computes $f^{n'}$ with fairness and identifiable abort, and (iii) a monitoring procedure for all $n$ parties to verify the correctness of an execution of $\pi'$ run by the committee members. In such a case, we could get a simple security-uplifting reduction with a low round complexity (assuming $r'\le r$). Specifically, in order to compute $f^n$ with full security, we would select a committee $\committee$, let the parties in $\committee$ execute $\pi'$ with full security using the player-elimination technique, while the remaining parties monitor the execution and receive the final output from the committee members. Since player elimination will only be applied to committee members, it may be applied at most $t'$ times. Hence, the resulting protocol will run in $O(t'\cdot r')$ rounds. Below, we explain how to select a committee $\committee$, and how the execution of the protocol $\pi'$ can be monitored by non-committee parties. Whether an appropriate protocol $\pi'$ exists depends on the functionality at hand.

Our key tool for electing the committee is Feige's lightest-bin protocol~\cite{Feige99}. This is a single-round protocol, secure against computationally unbounded adversaries, ensuring the following. If $n$ parties with up to $\beta\cdot n$ corruptions use the protocol to elect a committee $\committee$ of size $n'$, then for all $\beta'>\beta$, the fraction of corrupted parties in the committee is at most $\beta'$, with all but probability $\err(n,n',\beta,\beta')=\frac{n}{n'} \cdot e^{-\frac{(\beta'-\beta)^2 n'}{2(1-\beta)}}$. In particular, for $n'=\omega(\log(\secParam))$ Feige's protocol succeeds with all but negligible probability (in $\secParam$). The beauty of this protocol is in its simplicity, as parties are simply instructed to select a random bin (out of $n/n'$ possible ones), and the elected committee are the parties that chose the lightest bin.

We now turn to explain how the non-committee parties can monitor the work of the committee members. In the no-input setting that we have discussed so far, things are quite simple. Recall that all our protocols assume a broadcast channel, which allows the non-committee parties to see all communication among committee members.\footnote{{Private messages should be encrypted before being sent over the broadcast channel.}} Now, all that is needed is that when the protocol terminates, the non-committee parties can verify that they obtain the correct output from the computation. To this end, we start the protocol with committee members being publicly committed to a random string (used as their randomness in the execution). Then, as the protocol ends, a committee member notifies all parties of the output it received by proving in zero knowledge that it has followed the prescribed protocol using the randomness it committed to.

Proving security of the above reduction raises a subtle technical issue. Whenever a computation by the committee is invoked, it is required that all parties will obtain the output (either a genuine output or an identity of a corrupted committee member), however, only corrupted committee members are allowed to abort the computation. This property is not captured by the standard definition of fairness with identifiable abort, where every corrupted party can abort the computation. We therefore introduce a new ideal model with \textsf{fairness and restricted identifiable abort} that models this property. In this ideal model, the trusted party is parametrized by a subset $\committee\subseteq[n]$. The adversary, controlling parties in $\IS\subseteq[n]$, can abort the computation only if $\IS\cap\committee\neq\emptyset$, by revealing the identity of a corrupted party $\is\in\IS\cap\committee$. This means that if $\IS\cap\committee=\emptyset$ this ideal model provides full security, however, in case $\committee\subseteq \IS$, no security is provided, and the adversary gets to choose the output.\footnote{In the with-input setting \cref{sec:additionalRess}, the adversary also obtains the input values of all honest parties.}

The proof consists of two steps. Initially, full security is reduced to fairness with restricted identifiable abort. This is done by electing a super-logarithmic committee $\committee$ using Feige's protocol, and iteratively invoking the trusted party for computing $f^n$ with fairness with restricted identifiable abort, parametrized by $\committee$, until the honest parties obtain the output. Next, fairness with restricted identifiable abort is reduced to fairness. This is done by compiling (in a similar way to the GMW compiler) the protocol $\pi'$ for computing $f^{n'}$ with fairness into a protocol $\pi$ for computing $f^n$ with fairness with restricted identifiable abort.

\paragraph{Lower bound.}

Recall that our lower bound is given in the hybrid model in which a trusted party computes the coin-flipping functionality with fairness and {restricted} identifiable abort (as presented before, \cref{thm:intro_lowerbound}). In this model, in addition to standard communication rounds, a protocol also has \emph{functionality rounds} in which different committees (subsets) of the parties invoke the trusted party.

Consider an $n$-party coin-flipping protocol $\pi$ in this hybrid model with a constant number of functionality rounds. The heart of the proof is showing that if none of the committees is large, \ie has more than \rchanged{$\log(n)$}{$\log(\secParam)$} parties, then the protocol can be biased noticeably. The proof is completed by showing that since $\pi$ has only a constant number of functionality rounds, an adversary can force all calls made by large committees to abort, and thus attacking arbitrary protocols reduces to the no-large-committee case.

To prove the no-large-committees case, we transform the $n$-party coin-flipping protocol $\pi$ in the hybrid model, into a two-party coin-flipping protocol $\psi$ in the standard model. By \citet{Cleve86}, there exists an attack on protocol $\psi$. Hence, we complete the proof by showing how to transform the attack on $\psi$ (guaranteed by \citet{Cleve86}) into an attack on $\pi$. The aforementioned protocol transformation goes as follows: partition the $n$ parties of $\pi$ into two subsets, $\cs_0$ of size $\beta n$ and $\cs_1 = [n] \setminus \cs_0$. The two-party protocol $\psi = (\Party_0,\Party_1)$ emulates a random execution of $\pi$ by letting party $\Party_0$ emulate the parties in $\cs_0$ and party $\Party_1$ emulate the parties in $\cs_1$. The calls to the trusted party are emulated by $\Party_0$ as follows: let $\committee_1,\ldots,\committee_\ell$ be the (small) committees that invoke the trusted party, in parallel, in a functionality round. In protocol $\psi$, party $\Party_0$ sends $\ell$ uniformly distributed bits, each bit in a different round, and the parties interpret these bits as the output produced by the coin-flipping functionality. At the end of the protocol, each party outputs the output of the first party of $\pi$ in its control. If $\Party_0$ aborts while emulating a functionality round, \ie when it is supposed to send the output bit of a committee $\committee$, party $\Party_1$ continues as if the first party in $\committee$ (for simplicity, we assume this party is in $\cs_0$) aborts the call to the trusted party in $\pi$, and the rest of the parties in $\cs_0$ abort immediately after the call to the trusted party. If $\Party_0$ aborts in a round that emulates a communication round in $\pi$, party $\Party_1$ continues the emulation of $\pi$ as if all parties in $\cs_0$ abort. Party $\Party_0$ handles an abort by $\Party_1$ analogously.

By \citet{Cleve86}, there exists a round $\is$ such that one of the parties in $\psi$ can bias the protocol merely by deciding, depending on its view, whether to abort in round $\is$ or not.\footnote{The attacker of \cite{Cleve86} either aborts at round $\is$ or at round $\is+1$, but the transformation to the above attacker is simple (see \cref{sec:impo:Cleve}).} Assume, \wlg, that the attacking party is $\Party_0$, and the round $\is$ is a functionality round (other cases translate directly to attacks on $\pi$).
The core difference between the ability of an adversary corrupting party $\Party_0$ in $\psi$ from that of an adversary corrupting the parties in $\cs_0$ in $\pi$, is that the adversary in $\psi$ can decide whether to abort \emph{before} sending the $\is$'th message.
This raises a subtle issue, since the $\is$'th message corresponds to an output of the coin-flipping functionality in $\pi$, in response to a call made by some committee $\committee$. Yet, if the adversary in $\pi$ controls \emph{all} parties in $\committee$, he can abort \emph{after} seeing the output of the call to trusted party made by $\committee$ and the results of all other parallel calls, while still preventing other parties from getting the output of the call made by $\committee$.
We conclude the proof by showing that if the corrupted subset $\cs_0$ is chosen at random, then it contains all parties in the relevant committee with a noticeable probability, and thus the attack on $\pi$ goes through.

\subsection{Additional Results}\label{sec:additionalRess}
In the above discussion we only reviewed our reductions from full security to fairness for the no-input case. This was done for the sake of clarity, however, in this paper we also deal with arbitrary functionalities (with input)%
\ifdefined\IsResultWithAbort
 and with reducing full security to security with abort (when a vast majority of the parties are honest). We next state these additional results.
\else
.
\fi
We remark that the lower bound for the no-input case, described above, applies also to the with-input case.

\subsubsection{Full Security to Fairness -- Arbitrary Functionalities (with Inputs)}
The case of functionalities with inputs is somewhat more involved than that of no-input functionalities.
As in the no-input case, our fully secure computation of an $n$-party functionality $f^n$ is done by delegating the computation to a small committee that computes a related $n'$-party functionality with fairness.
However, when considering functionalities with inputs, parties outside the committee cannot reveal their inputs to committee members, but still need to make sure that the right input was used in the computation performed by committee members.
This can be done using secret-sharing schemes and commitments. Note that non-committee parties take a bigger role in the computation now. However, corrupted parties outside the committee should never be able to cause the protocol to prematurely terminate, as otherwise the number of rounds would depend on the number of corruptions among all parties and not only committee members. The above becomes even more challenging when wishing to have a few committees perform the computation in parallel. Here, it must also be verified that each party provides \emph{the same} input to all committees.


Considering the no-honest-majority case, we let each party $\Party_i$ secret share its input $x_i$ in an $n'$-out-of-$n'$ secret sharing, publicly commit to every share, and send each decommitment value, encrypted, to the corresponding committee member. We define $\fin{f^n}{n}{n'}$ to be the $n'$-party functionality, parametrized by a vector of commitments $(\comval_i^1,\ldots,\comval_i^{n'})$ for every $\Party_i$, where $\comval_i^j$ is a commitment to the \jth share of $x_i$. The functionality receives as input the decommitments of each $c_i^j$, reconstructs the decommitted values to obtain the $n$-tuple $(x_1,\ldots,x_n)$, computes $y=f^n(x_1,\ldots,x_n)$, and outputs $y$ in the clear (see \cref{fig:ssin}).

By having the parties publicly commit to shares of their inputs (using a perfectly binding commitment) and send the decommitment values to the committee members, corrupted committee members cannot change the values corresponding to honest parties (otherwise the decommit will fail and the cheating committee member will be identified). Preventing corrupted parties from sending invalid decommitments to honest committee members is external to the functionality and must be part of the protocol. In addition to TDP and CRH, we assume non-interactive perfectly binding commitment schemes exist.\footnote{Although non-interactive perfectly binding commitments can be constructed from one-way permutations, in our setting, one-way functions are sufficient. This follows since Naor's commitments~\cite{Naor91} can be made non-interactive in the common random string (CRS) model, and even given a weak CRS (a high min-entropy common string). A high min-entropy string can be constructed by $n$ parties, without assuming an honest majority, using the protocol from~\cite{GVZ06} that requires $\logstar(n)+O(1)$ rounds.}
We prove the following.

\begin{theorem}[fairness to full security, informal]\label{thm:intro:mainThmfair_to_full}
Let $f^n$ be an $n$-party functionality, let $n'=\varphi(\secParam)\cdot\log(\secParam)$ for $\varphi=\Omega(1)$, let $0<\beta<\beta'<1$, let $t=\beta n$, let $t'=\beta' n'$, and let $\err=\err(n,n',\beta,\beta')$.
The following hold assuming TDP, CRH, and non-interactive perfectly binding commitment schemes.
\begin{enumerate}
    \item
    If $\fin{f^n}{n}{n'}$ can be \rchanged{$\delta$}{$\delta'$}-fairly computed by an \rchanged{$r$}{$r'$}-round protocol, tolerating $t'$ corruptions, then $f^n$ can be computed with \rchanged{$(t'\cdot\delta + \err)$}{$(t'\cdot\delta' + \err)$}-full-security, tolerating $t$ corruptions, by an \rchanged{$O(t'\cdot r)$}{$O(t'\cdot r')$}-round protocol.
    \item
    If $\fin{f^n}{n}{n'}$ can be \rchanged{$\delta$}{$\delta'$}-fairly computed by an \rchanged{$r$}{$r'$}-round protocol, tolerating $n'-1$ corruptions, $\ell$-times in parallel, for $\ell=\secParam^{c}$ (for some universal constant $c$), then $f^n$ can be computed with \rchanged{$(\varphi(\secParam)^2\cdot\ell\cdot\delta + \err)$}{$(\varphi(\secParam)^2\cdot\ell\cdot\delta' + \err)$}-full-security, tolerating $t$ corruptions, by an \rchanged{$O(\varphi(\secParam)^2\cdot r)$}{$O(\varphi(\secParam)^2\cdot r')$}-round protocol.
\end{enumerate}
\end{theorem}

In the honest-majority setting, a similar result can be achieved with the transformation only requiring black-box access to the fair protocol, and the resulting security being unconditional. Furthermore, the transformation becomes much simpler with an honest majority and relies solely on ECSS scheme.
We denote by $\finout{f^n}{n}{t'}{n'}$, for $t'<n'/2$, the $n'$-party functionality that receives secret shares of an $n$-tuple $(x_1,\ldots,x_n)$, reconstructs the inputs, computes $y=f(x_1,\ldots,x_n)$ and outputs secret shares of $y$ (see \cref{fig:RecCompShare}). See \cref{sec:fairtofull_withinput_HM} for more details.

\paragraph{Reducing a logarithmic factor.}
When considering functionalities with inputs, it is possible to use generic techniques (see, for example, \cite[Sec.\ 2.5]{HL10}) and assume without loss of generality that the functionality is deterministic and has a public output (\ie all parties receive the same output). In this case, we show how to reduce an additional logarithmic factor from the number of fair computations performed by the committee, compared to the no-input case. The parties start by electing a random, (super-)logarithmic committee $\committee$, of size $m=\varphi(\secParam) \cdot \log(\secParam)$, for some $\varphi(\secParam)\in\Omega(1)$ (\eg $\varphi(\secParam)=\logstar(\secParam)$). However, instead of sharing the inputs with the committee members, the protocol considers all sufficiently large sub-committees, \ie all subsets of $\committee$ of size $n'=m-\log(\secParam) / \varphi(\secParam)$. Next, every party secret shares its input to each of the sub-committees, and each of the sub-committees computes, in parallel, the functionality $\fin{f}{n}{n'}$ with fairness and identifiable abort. It is important for each party to prove in zero knowledge that the same input value is shared across all sub-committees, in order to ensure the same output value in all computations. We show that in this case: $(1)$ there are polynomially many sub-committees, $(2)$ with overwhelming probability, no sub-committee is fully corrupted, and $(3)$ if the adversary aborts the fair computations in all sub-committees, then $\log(\secParam) / \varphi(\secParam)$ corrupted parties must be identified. It follows that after $\varphi^2(\secParam)$ iterations the protocol is guaranteed to successfully terminate.

In order to prove security of this construction, we generalize the notion of \emph{fairness with restricted identifiable abort} to the with-input setting. The ideal model is parametrized by a list of subsets $\committee_1,\ldots,\committee_\ell\subseteq[n]$, such that if one of the subsets is fully corrupted, \ie $\committee_i\subseteq\IS$ for some $i\in[\ell]$ (where $\IS$ is the set of corrupted parties), then no security is provided (the adversary gets all inputs and determines the output). If one of the subsets is fully honest, \ie $\committee_i\cap\IS=\emptyset$ for some $i\in[\ell]$, then the adversary cannot abort the computation. Otherwise, the adversary is allowed to abort the computation by revealing a corrupted party in each subset, however, only before it has learned any new information. See \cref{sec:restricted_idabort} for more details.

\rnote{Are there applications for $1/p$-security?}

\paragraph{Application to fully secure multiparty Boolean OR.}
An application of the above reductions is a fully secure protocol for $n$-party Boolean OR.
\citet{GordonK09} constructed a fully secure protocol, tolerating $t<n$ corruptions, that requires $O(t)$ rounds.
\radded{Using \cref{thm:intro:mainThmfair_to_full} we show how to achieve (any) super-constant round complexity when the fraction of corruptions is constant.}
\begin{corollary}[informal]\label{cor:intro_OR}
Under the assumptions in \cref{thm:intro:mainThmfair_to_full}, the $n$-party Boolean OR functionality can be computed with full security tolerating $t=\beta n$ corruptions, for $0<\beta<1$, with round complexity $O(\logstar(\secParam))$.
\end{corollary}

\paragraph{Application to a best-of-both-worlds type result.}
Another application is to a variant of the protocol of \citet{IKKLP11} that guarantees $t$-full-security assuming an honest majority and \emph{$t$-full-privacy} otherwise.\footnote{$t$-full-privacy means that the adversary does not learn any additional information other than what it can learn from $t+1$ invocations of the ideal functionality, with fixed inputs for the honest parties.}
Their idea is to repeatedly compute $\fout{f^n}$, using a secure protocol with identifiable abort, and use the player-elimination approach until the honest parties obtain the secret shares and reconstruct the result. It follows that the round complexity in~\cite{IKKLP11} is $O(t)$. The above reduction suggests an improvement both to the round complexity of the protocol and to the privacy it guarantees.
\begin{corollary}[informal]\label{cor:intro_IKLP}
Let $f^n$ be an $n$-party functionality and let $t=\beta n$ for $0<\beta<1$, and consider the assumptions as in \cref{thm:intro:mainThmfair_to_full}.
Then, there exists a single protocol $\pi$, with round complexity $O(\logstar(\secParam)\cdot\log(\secParam))$, such that:
\begin{enumerate}
    \item $\pi$ computes $f^n$ with $O(\logstar(\secParam)\cdot\log(\secParam))$-full-privacy.
    \item If $\beta<1/2$, then $\pi$ computes $f^n$ with full security.
\end{enumerate}
\end{corollary}

\paragraph{Application to uplifting partially identifiable abort to full security.}
Finally, we improve a recent transformation of \citet[Thm.\ 3 and 4]{IKPSY16} from partially identifiable abort\footnote{A computation has $\alpha$-partially identifiable abort~\cite{IKPSY16}, if in case the adversary aborts the computation, a subset of parties is identified, such that at least an $\alpha$-fraction of the subset is corrupted.} to full security in the honest-majority setting. In \cite{IKPSY16}, the computation of $\finout{f^n}{n}{t'}{n'}$ with partially identifiable abort is carried out iteratively by a committee, initially consisting of all the parties, until the output is obtained. In case of abort, all the identified parties (both honest and corrupted) are removed from the committee. It follows that the number of iterations in~\cite{IKPSY16} is $O(n)$.

\begin{corollary}[informal]\label{cor:intro_IKPSY}
Let $f^n$ be an $n$-party functionality, let $n'=\logstar(\secParam)\cdot\log(\secParam)$, let $0<\beta<\beta'<1/2$, let $t=\beta n$ and $t'=\beta' n'$, and let $\pi'$ be an $r$-round protocol that securely computes $\finout{f^n}{n}{t'}{n'}$ with $\beta'$-partially identifiable abort, tolerating $t'$ corruptions.
Then, $f^n$ can be computed with full security, tolerating $t$ corruptions, by a $O(t'\cdot r)$-round protocol that uses the protocol $\pi'$ in a black-box way.
\end{corollary}

\ifdefined\IsResultWithAbort
\subsubsection{Security with Abort to Full Security}
When a vast majority of the parties are honest, similar techniques can be used to efficiently uplift security with abort to full security. We emphasize that since we only consider security with abort, no corrupted parties are identified if the protocol halts, and so the player-elimination technique is not applicable in this setting. We stress that for this result, $n$ is not required to be super-constant.
We prove the following theorem.
\begin{theorem}[security with abort to full security, informal]\label{thm:intro:aborttofull}
Let $f^n$ be an $n$-party functionality, and let $t$ such that $t \cdot (2t+1)<n$. Then, $f^n$ can be computed with full security tolerating $t$ corruptions (with information-theoretic security) in the hybrid model computing $\finout{f^n}{n}{t}{2t+1}$ with abort. For $t \cdot (3t+1)<n$, the above holds with perfect security.
\end{theorem}

To the best of our knowledge, the transformation in \cref{thm:intro:aborttofull} is the first generic black-box transformations from security with abort to full security (not requiring identifiability).
\fi

\subsection{Additional Related Work}\label{sec:relatedWork}

The MPC literature contains many examples of reductions from strong security notions to weaker ones, \eg \cite{GoldreichMW87,CLOS02,IKOS07,IPS08}, to name but a few.
Recently, \citet{IKPSY16} presented a formal framework for studying (black-box) transformations between different security notions (see further examples therein). All of our results in the honest-majority setting can be stated in the framework of~\cite{IKPSY16}.

Security with identifiable abort was first explicitly used by \citet{AL10}, however, it was widely used in the literature implicitly, especially in the realm of fairness~\cite{IKKLP11,HMZ08,ZHM09,GordonK09,BeimelOO15,GK10,BLOO11,HaitnerT14,CL17,IOZ14,BOS16,AO16,BHLT17}.
\citet{KZZ16} recently defined security with \emph{publicly identifiable abort}, where a subset of the parties perform a computation, and all parties (even outside of the subset) identify a corrupted party in case of abort. This definition is different from the definition we introduce of \emph{restricted identifiable abort}, since unlike our definition, in~\cite{KZZ16} non-working parties do not provide input and do not receive output; such parties can only identify a corrupted party in case of abort.

The idea of electing a small committee to perform a computation was initially used in~\cite{Bracha87}, and has been considered in numerous  settings, such as: leakage-resilient secure computation~\cite{BGK11,BGJK12}, large-scale MPC~\cite{BGT13,BCP15,DKMSZ17,BCDH18}, leader election~\cite{KSSV06,KKKSS08}, Byzantine agreement~\cite{KS09,KS10,KLST11,BGH13}, and distributed key-generation~\cite{CS04}.

\citet{LR17} recently considered a slightly different model, in which a secure protocol must have a fixed and a priori known \emph{committal round}, where the effective inputs of the corrupted parties are determined. They showed that fair protocols cannot be constructed in this model without an honest majority. We note that our constructions of fair protocols over committed inputs do not contradict this impossibility result, since we consider the standard model, where the simulator may choose not to use the values that are committed by the corrupted parties in the protocol, and may decide on the effective inputs based on the behavior of the real-world adversary at \emph{any} round.

\subsection{Open Questions}
In the no-input, no-honest-majority setting, there is a logarithmic gap between the $\omega(\log(\secParam))$ number of sequential fair calls required by our reductions, and our $\omega(1)$ lower bound on this number. The source for this gap is in the difference between the pessimistic model we use in our reductions, and the more optimistic (yet, as far as we know, possible) model we use in our lower bound (see more details in \cref{sec:imposability:model}).
Finding the right security model to capture the power of such protocols, and finding tight reductions in this setting, is an interesting open question.

\subsection*{Paper Organization}
Basic definitions can be found in \cref{sec:Preliminaries}.
Our reductions from full security to fairness for no-input functionalities are given in \cref{sec:fairtofull_noinput}, and for functionalities with inputs in \cref{sec:fair_to_full_withinputs}. 
\ifdefined\IsResultWithAbort
Reductions from full security to security with abort are given in \cref{sec:abort_to_full}. 
\fi
The lower bound on the number of sequential fair calls is given in \cref{sec:impossibility}.

\section{Preliminaries}\label{sec:Preliminaries}

\rnote{Maybe move most of the preliminaries to the appendix and only present the basic model and the new definition?} \Inote{why?}
\rnote{not to exhaust the reader with standard definitions}

\subsection{Notations}\label{sec:notations}
We use calligraphic letters to denote sets, uppercase for random variables, lowercase for values, boldface for vectors, and sans-serif (\eg \Ac) for algorithms (\ie Turing Machines).
For $n\in\N$, let $[n]=\set{1,\cdots,n}$. Let $\poly$ denote the set all positive polynomials and let \ppt denote a probabilistic algorithm that runs in \emph{strictly} polynomial time. A function $\nu \colon \N \mapsto [0,1]$ is \emph{negligible}, denoted $\nu(\secParam) = \negl(\secParam)$, if $\nu(\secParam)<1/p(\secParam)$ for every $p\in\poly$ and large enough $\secParam$.
The statistical distance between two random variables $X$ and $Y$ over a finite set $\Uni$, denoted $\SD(X,Y)$, is defined as $\frac12 \cdot \sum_{u\in \Uni}\size{\pr{X = u}- \pr{Y = u}}$.
Given a random variable $X$, we write $x\la X$ to indicate that $x$ is selected according to $X$.

Two distribution ensembles $X=\set{X(a,\secParam)}_{a\in\zs,\secParam\in\N}$ and $Y=\set{Y(a,\secParam)}_{a\in\zs,\secParam\in\N}$ are computationally $\delta$-indistinguishable (denoted $X\deltaci Y$) if for every non-uniform polynomial-time distinguisher $\Adv$ there exists a function $\nu(\secParam) = \negl(\secParam)$, such that for every $a\in\zs$ and $\secParam\in\N$
\[
\abs{\pr{\Adv(X(a,\secParam),1^\secParam)=1} - \pr{\Adv(Y(a,\secParam),1^\secParam)=1}}\leq \delta(\secParam) + \nu(\secParam).
\]
In case $\delta$ is negligible, we say that $X$ and $Y$ are computationally indistinguishable and denote $X\ci Y$.
The distribution ensembles $X$ and $Y$ are (statistically) $\delta$-close (denoted $X\deltaclose Y$) if $\SD(X,Y)\le \delta+\nu(\secParam)$ for a negligible function $\nu$, and statistically close (denoted $X\statclose Y$) is they are $\delta$-close and $\delta$ is negligible. $X$ and $Y$ are perfectly $\delta$-close (denoted $X\deltaequiv Y$) if $\SD(X,Y)\le \delta$. In case $\delta=0$, \ie if $X$ and $Y$ are identically distributed, denote $X\equiv Y$.

We denote by $n$ the number of participating parties in a protocol, by $t$ an upper bound on the number of corrupted parties, and by $\secParam$ the security parameter. \rchanged{We}{As standard in the context of large-scale secure computation (see, \eg \cite{DI06,DIKNS08}), we} assume that $n$ and $\secParam$ are polynomially related,
\ifdefined\IsResultWithAbort
(except for \cref{sec:abort_to_full}),
\fi
\ie $n=\secParam^c$ for some constant $c>0$ (possibly $c<1$).
\radded{In particular, this means that a functionality $f$ is in fact an ensemble of functionalities $\sset{f_n}_{n\in\N}$, where $f_n$ is an $n$-party functionality; for example, in the Boolean OR functionality, for every $n\in\N$, $f_n(x_1,\ldots,x_n)=x_1\vee\ldots\vee x_n$. We refer the reader to \cite{BCDH18} for a discussion on functionalities and protocols in the large-scale setting.}

\subsection{Secret Sharing}\label{sec:secretsharing}
A (threshold) secret-sharing scheme~\cite{Shamir79} is a method in which a dealer distributes shares of some secret to $n$ parties such that $t$ colluding parties do not learn anything about the secret, and any subset of $t+1$ parties can fully reconstruct the secret.

\begin{definition}[secret sharing]\label{def:TSS}
A \emph{$(t+1)$-out-of-$n$ secret-sharing scheme} over a message space $\cM$ consists of a pair of algorithms $(\Share, \Recon)$ satisfying the following properties:
\begin{enumerate}
    \item\textbf{$t$-privacy:}
    For every $m\in \cM$, and every subset $\IS\subseteq[n]$ of size $\ssize{\IS}\leq t$, the distribution of $\sset{s_i}_{i\in \IS}$ is independent of $m$, where $(s_1,\ldots,s_n)\gets \Share(m)$.
    \item\textbf{$(t+1)$-reconstructability:}
    For every $m\in \cM$, every subset $\IS\subseteq[n]$ of size $t+1$, every $\vs=(s_1,\ldots,s_n)$ and every $\vs'=(s'_1,\ldots,s'_n)$ such that $\ppr{\vS\gets \Share(m)}{\vS=\vs}>0$, $\vs_{\IS}=\vs'_{\IS}$, and $\vs'_{\bar{\IS}} = \bot^{\ssize{\bar{\IS}}}$, it holds that $m=\Recon(\vs')$.
\end{enumerate}
\end{definition}

An error-correcting secret-sharing (ECSS) scheme is a secret-sharing schemes, in which the reconstruction is guaranteed to succeed even if up to $t$ shares are faulty.
This primitive has also been referred to as \emph{robust secret sharing} or as \emph{honest-dealer VSS}~\cite{RB89,CFOR12,CDF01}.

\begin{definition}[error-correcting secret sharing]\label{def:ECSS}
A \emph{$(t+1)$-out-of-$n$ error-correcting secret-sharing scheme} (ECSS) over a message space $\cM$ consists of a pair of algorithms $(\Share, \Recon)$ satisfying the following properties:
\begin{enumerate}
    \item\textbf{$t$-privacy:} As in \cref{def:TSS}.
    \item\textbf{Reconstruction from up to $t$ erroneous shares:}
    For every $m\in \cM$, every $\vs = (s_1,\ldots,s_n)$, and every $\vs' = (s'_1,\ldots,s'_n)$ such that $\ppr{\vS\gets \Share(m)}{\vS=\vs}>0$ and $\ssize{\sset{i \mid s_i=s'_i}}\geq n-t$, it holds that $m=\Recon(\vs')$.
\end{enumerate}
\end{definition}

ECSS can be constructed with perfect correctness when $t<n/3$ using Reed-Solomon decoding~\cite{BGW88} and with a negligible error probability when $t<n/2$ by authenticating the shares using one-time MAC~\cite{RB89}.
In case $t\geq n/2$ it is impossible to construct a $(t+1)$-out-of-$n$ ECSS scheme, or even a secret-sharing scheme that identifies cheaters~\cite{IOS12}.

\subsection{Security Definitions}\label{sec:mpc_Def}

We provide the basic definitions for secure multiparty computation according to the real/ideal paradigm, for further details see~\cite{Goldreich04}. Informally, a protocol is considered secure if whatever an adversary can do in the real execution of protocol, can be done also in an ideal computation, in which an uncorrupted trusted party assists the computation.

\begin{definition}[functionalities]\label{def:func}
An $n$-party {\sf functionality} is a random process that maps vectors of $n$ inputs to vectors of $n$ outputs.\footnote{We assume that a functionality can be computed in polynomial time.} Given an $n$-party functionality $f \colon (\zs)^n \mapsto (\zs)^n$, let $f_i(\vx)$ denote its \ith output coordinate, \ie $f_i(\vx) = f(\vx)_i$. A functionality $f$ has {\sf public output}, if the output values of all parties are the same, \ie for every $\vx\in(\zs)^n$, $f_1(\vx)=f_2(\vx)=\ldots=f_n(\vx)$, otherwise $f$ has {\sf private output}.

A no-input $n$-party functionality $f \colon (\zs)^n \mapsto (\zs)^n$ is a functionality in which the input of every party is the empty string $\lambda$. That is, $f$ is computed as $(y_1,\ldots,y_n)\gets f(\lambda,\ldots,\lambda;r)$ over random coins $r$. In case $f$ is a no-input functionality with public output, it can be defined for any number of parties; we denote by $f^n$ the functionality $f$ when defined for $n$ parties.
\end{definition}

\subsubsection{Execution in the Real World}
An $n$-party protocol $\pi= (\Party_1,\ldots,\Party_n)$ is an $n$-tuple of probabilistic polynomial-time interactive Turing machines.
The term \emph{party $\Party_i$} refers to the \ith interactive Turing machine.
Each party $\Party_i$ starts with input $x_i\in\zs$ and random coins $r_i\in\zs$.
Without loss of generality, the input length of each party is assumed to be the security parameter $\secParam$.
An \emph{adversary} \Adv is another interactive TM describing the behavior of the corrupted parties.
It starts the execution with input that contains the identities of the corrupted parties and their private inputs, and possibly an additional auxiliary input.
The parties execute the protocol in a synchronous network.
That is, the execution proceeds in rounds: each round consists of a \emph{send phase} (where parties send their messages from this round) followed by a \emph{receive phase} (where they receive messages from other parties).
The adversary is assumed to be \emph{rushing}, which means that it can see the messages the honest parties send in a round before determining the messages that the corrupted parties send in that round.

The parties can communicate in every round over a broadcast channel or using a fully connected \ptp network.
We consider two models for the communication lines between the parties: In the \emph{authenticated-channels} model (used in the computational setting), the communication lines are assumed to be ideally authenticated but not private (and thus the adversary cannot modify messages sent between two honest parties but can read them).
In the \emph{secure-channels} model (used in the information-theoretic setting), the communication lines are assumed to be ideally private (and thus the adversary cannot read or modify messages sent between two honest parties).

Throughout the execution of the protocol, all the honest parties follow the instructions of the prescribed protocol, whereas the corrupted parties receive their instructions from the adversary.
The adversary is considered to be \emph{malicious}, meaning that it can instruct the corrupted parties to deviate from the protocol in any arbitrary way.
At the conclusion of the execution, the honest parties output their prescribed output from the protocol, the corrupted parties output nothing and the adversary outputs an (arbitrary) function of its view of the computation (containing the views of the corrupted parties).
The view of a party in a given execution of the protocol consists of its input, its random coins, and the messages it sees throughout this execution.

\begin{definition} [real-model execution]\label{def:RealModel}
Let $\pi= (\Party_1,\ldots, \Party_n)$ be an $n$-party protocol and let $\IS \subseteq [n]$ denote the set of indices of the parties corrupted by $\Adv$. The {\sf joint execution of $\pi$ under $(\Adv,\IS)$ in the real model}, on input vector $\vx= (x_1,\ldots, x_n)$, auxiliary input $\aux$ and security parameter $\secParam$, denoted $\REAL_{\pi,\IS,\Adv(\aux)}(\vx,\secParam)$, is defined as the output vector of $\Party_1,\ldots,\Party_n$ and $\Adv(\aux)$ resulting from the protocol interaction, where for every $i \in \IS$, party $\Party_i$ computes its messages according to $\Adv$, and for every $j \notin \IS$, party $\Party_j$ computes its messages according to $\pi$.
\end{definition}

\paragraph{Bounded-parallel composition.}
An $\ell$-times parallel execution of a protocol $\pi'$ (for a pre-determined $\ell$) is a protocol $\pi$, in which the parties run $\ell$ independent instances of $\pi'$ in parallel. Each party receives a vector of $\ell$ input values, one input for every execution, and all the executions proceed in a synchronous manner, round by round, such that all the messages of the \ith round in all executions of $\pi'$ are guaranteed to be delivered before round $i+1$ starts. At the conclusion of the executions, every party outputs a vector of $\ell$ output values, one output value from each execution.
Note that the honest parties run each execution of $\pi$ obliviously to the other executions. (Thus, this is stateless composition.) The adversary may gather information from all the executions in order to attack any specific execution. The $\ell$-times parallel execution of $\pi'$ can also be executed by different subsets of parties, in this case every party that does not participate in the \ith execution receives an empty input $\lambda$ and outputs $\lambda$.

\subsubsection{Execution in the Ideal World}
In this section, we present standard definitions of ideal-model computations that are used to define security with abort, with identifiable abort, with fairness and with guaranteed output delivery (\ie full security).
We start by presenting the ideal-model computation for security with abort, where the adversary may abort the computation either before or after it has learned the output; other ideal-model computations are defined by restricting the power of the adversary either by forcing the adversary to identify a corrupted party in case of abort, or by allowing fair abort (\ie abort only before learning the output) or no abort (full security).

\paragraph{Ideal computation with abort.}
An ideal computation with abort of an $n$-party functionality $f$ on input $\vx=(x_1,\ldots,x_n)$ for parties $(\Party_1,\ldots,\Party_n)$ in the presence of an ideal-model adversary $\Adv$ controlling the parties indexed by $\IS\subseteq[n]$, proceeds via the following steps.
\begin{itemize}
    \item[\emph{Sending inputs to trusted party}:]
    An honest party $\Party_i$ sends its input $x_i$ to the trusted party.
    The adversary may send to the trusted party arbitrary inputs for the corrupted parties. Let $x_i'$ be the value actually sent as the input of party $\Party_i$.

	\item[\emph{Early abort}:]
    The adversary $\Adv$ can abort the computation by sending an $\abort$ message to the trusted party. In case of such an abort, the trusted party sends $\bot$ to all parties and halts.

    \item[\emph{Trusted party answers adversary}:]
    The trusted party computes $(y_1, \ldots, y_n)=f(x_1', \ldots, x_n')$ and sends $y_i$ to party $\Party_i$ for every $i\in\IS$.

	\item[\emph{Late abort}:]
    The adversary $\Adv$ can abort the computation (\emph{after seeing the outputs of corrupted parties}) by sending an \abort message to the trusted party. In case of such abort, the trusted party sends $\bot$ to all parties and halts. Otherwise, the adversary sends a $\continue$ message to the trusted party.

    \item[\emph{Trusted party answers remaining parties}:]
    The trusted party sends $y_i$ to $\Party_i$ for every $i\notin\IS$.

    \item[\emph{Outputs}:]
    Honest parties always output the message received from the trusted party and the corrupted parties output nothing.
    The adversary $\Adv$ outputs an arbitrary function of the initial inputs $\set{x_i}_{i\in\IS}$, the messages received by the corrupted parties from the trusted party and its auxiliary input.
\end{itemize}

\begin{definition}[ideal-model computation with abort]\label{def:ideal_abort}
Let $f \colon (\zs)^n \mapsto (\zs)^n$ be an $n$-party functionality, let $\IS\subseteq [n]$ be the set of indices of the corrupted parties, and let $\secParam$ be the security parameter.
Then, the \emph{joint execution of $f$ under $(\Adv, I)$ in the ideal model}, on input vector $\vx=(x_1, \ldots, x_n)$, auxiliary input $\aux$ to $\Adv$ and security parameter $\secParam$, denoted $\IDEAL^{\abort}_{f,\IS,\Adv(\aux)}(\vx,\secParam)$, is defined as the output vector of $\Party_1, \ldots, \Party_n$ and $\Adv$ resulting from the above described ideal process.
\end{definition}

\noindent
We now define the following variants of this ideal computation:
\begin{itemize}
    \item\textbf{Ideal computation with identifiable abort}.
    This ideal model proceeds as in \cref{def:ideal_abort}, with the exception that in order to abort the computation, the adversary chooses an index of a corrupted party $\is\in\IS$ and sends $(\abort,\is)$ to the trusted party. In this case the trusted party responds with $(\bot,\is)$ to all parties. This ideal computation is denoted as $\IDEAL^{\idabort}_{f,\IS,\Adv(\aux)}(\vx,\secParam)$.
    \item\textbf{Ideal computation with fairness}.
    This ideal model proceeds as in \cref{def:ideal_abort}, with the exception that the adversary is allowed to send \abort only in step \emph{Early abort}. This ideal computation is denoted as $\IDEAL^{\fair}_{f,\IS,\Adv(\aux)}(\vx,\secParam)$.
    \item\textbf{Ideal computation with fairness and identifiable abort}.
    This ideal model proceeds as the ideal model for security with fairness, with the exception that in order to abort the computation, the adversary sends $(\abort,\is)$ with $\is\in\IS$ and the trusted party responds with $(\bot,\is)$ to all parties. This ideal computation is denoted as $\IDEAL^{\idfair}_{f,\IS,\Adv(\aux)}(\vx,\secParam)$.
    \item\textbf{Ideal computation with full security (aka guaranteed output delivery)}.
    This ideal model proceeds as in \cref{def:ideal_abort}, with the exception that the adversary is not allowed to send \abort to the trusted party. This ideal computation is denoted as $\IDEAL^{\full}_{f,\IS,\Adv(\aux)}(\vx,\secParam)$.
\end{itemize}

\subsubsection{Security Definitions}

Having defined the real and ideal models, we can now define security of protocols according to the real/ideal paradigm.
\begin{definition}\label{def:SecureProtocol}
Let $\type\in \set{\full, \fair, \idfair,\abort,\idabort}$.
Let $f\colon(\zs)^n \mapsto (\zs)^n$ be an $n$-party functionality, and let $\pi$ be a probabilistic polynomial-time protocol computing $f$.
The \emph{protocol $\pi$ $(\delta,t)$-securely computes $f$ with \type (and computational security)}, if for every probabilistic polynomial-time real-model adversary \Adv, there exists a probabilistic (expected) polynomial-time adversary $\Sim$ for the ideal model, such that for every $\IS\subseteq [n]$ of size at most $t$, it holds that
\[
\set{\bigbrack \REAL_{\pi, \IS, \Adv(\aux)}(\vx, \secParam)}_{(\vx, \aux)\in(\zs)^{n+1}, \secParam\in\N}
\deltaci
\set{\bigbrack \IDEAL^\type_{f, \IS, \Sim(\aux)}(\vx, \secParam)}_{(\vx, \aux)\in(\zs)^{n+1}, \secParam\in\N}.
\]
If $\delta$ is negligible, we say that $\pi$ is a protocol that $t$-securely computes $f$ with \type and computational security.

The \emph{protocol $\pi$ $(\delta,t)$-securely computes $f$ with \type (and statistical security)}, if for every real-model adversary \Adv, there exists an adversary $\Sim$ for the ideal model, whose running time is polynomial in the running time of $\Adv$, such that for every $\IS\subseteq [n]$ of size at most $t$, it holds that
\[
\set{\bigbrack \REAL_{\pi, \IS, \Adv(\aux)}(\vx, \secParam)}_{(\vx, \aux)\in(\zs)^{n+1}, \secParam\in\N}
\deltaclose
\set{\bigbrack \IDEAL^\type_{f, \IS, \Sim(\aux)}(\vx, \secParam)}_{(\vx, \aux)\in(\zs)^{n+1}, \secParam\in\N}.
\]
If $\delta$ is negligible, we say that $\pi$ is a protocol that $t$-securely computes $f$ with \type and statistical security.

Similarly, \emph{$\pi$ is a protocol that $(\delta,t)$-securely computes $f$ with \type (and perfect security)}, if
\[
\set{\bigbrack \REAL_{\pi, \IS, \Adv(\aux)}(\vx, \secParam)}_{(\vx, \aux)\in(\zs)^{n+1}, \secParam\in\N}
\deltaequiv
\set{\bigbrack \IDEAL^\type_{f, \IS, \Sim(\aux)}(\vx, \secParam)}_{(\vx, \aux)\in(\zs)^{n+1}, \secParam\in\N}.
\]
If $\delta=0$, we say that $\pi$ is a protocol that $t$-securely computes $f$ with \type and perfect security.
\end{definition}

\subsubsection{Reactive Functionalities}\label{sec:reactive}

The previous section described non-reactive ideal computations (also referred to as secure function evaluation), where each party sends one input to the trusted party and receives back one output; the trusted party is stateless, \ie no internal state is stored between computations of different functionalities. A reactive functionality is defined as vector of $n$-party functionalities $\vf=(f_1,\ldots,f_q)$ where each functionality receives an additional input representing a state. That is, for every $i\in[q]$, the functionality $f_i$ is defined as $(y^i_1,\ldots,y^i_n)=f_i(x^i_1,\ldots,x^i_n,\state^i;r^i)$. In an ideal computation of a reactive functionality $\vf$ with security \type, the trusted party is modeled as a stateful interactive Turing machine. Initially the trusted party sets the internal state to be $\state^1=\lambda$. Next, the ideal computation proceeds in phases, where in the \ith phase every party and the adversary send their messages to the trusted party according to the non-reactive ideal-model computation of $f_i$ with security \type, with the exception that upon receiving inputs $x^i_1,\ldots,x^i_n$ from the honest parties and the adversary, the trusted party samples random coins $r^i$, computes $(y^i_1,\ldots,y^i_n)=f_i(x^i_1,\ldots,x^i_n,\state^i;r^i)$ and sets $\state^{i+1}=\state^i\circ(x^i_1,\ldots,x^i_n,r^i)$.

Looking ahead, in \cref{sec:apps_withinput} (a protocol for Boolean OR) we will make use of a weaker form of reactive functionalities.
We define a reactive functionality $\vf=(f_1,\ldots,f_q)$ to be a \textsf{single-input reactive functionality}, if for every $i\geq 2$, the function $f_i$ is deterministic and does not receive inputs from the parties, \ie whose output depends only on the state $\state^2$, which is in fact the input values sent for $f_1$. Stated differently, denote by $x_1,\ldots,x_n$ the input values provided by the parties in the first invocation, then for every $i\in[q]$ the computation is $f_i(x_1,\ldots,x_n)$.

\subsubsection{The Hybrid Model}

The \emph{hybrid model} is a model that extends the real model with a trusted party that provides ideal computation for specific functionalities.
The parties communicate with this trusted party in exactly the same way as in the ideal models described above.
The question of which ideal model is considered must be specified.
Specifically, the trusted party may work according to any of the ideal models that we have defined above.

Let $f$ be a functionality.
Then, an execution of a protocol $\pi$ computing a functionality $g$ in the $f$-hybrid model, involves the parties sending normal messages to each other (as in the real model) and in addition, having access to a trusted party computing $f$.
It is essential that the invocations of $f$ are done sequentially, meaning that before an invocation of $f$ begins, the preceding invocation of $f$ must finish.
In particular, there is at most a single call to $f$ per round, and no other messages are sent during any round in which $f$ is called. In case $f$ is a reactive functionality then $f$ must be sequentially called until the computation of $f$ is completed.

Let $\type\in \set{\full, \fair, \idfair,\abort,\idabort}$.
Let $\Adv$ be an adversary with auxiliary input $\aux$ and let $\IS\subseteq[n]$ be the set of corrupted parties.
We denote by $\HYBRID^{f,\type}_{\pi, \IS, \Adv(\aux)}(\vx, \secParam)$ the random variable consisting of the view of the adversary and the output of the honest parties, following an execution of $\pi$ with ideal calls to a trusted party computing $f$ according to the ideal model \type, on input vector $\vx=(x_1, \ldots, x_n)$, auxiliary input $\aux$ to $\Adv$, and security parameter $\secParam$.
We call this the $(f,\type)$-hybrid model.

The following proposition follows from the sequential composition theorem of \citet{Canetti00}.
\begin{proposition}\label{prop:Composition}
Let $\type_1,\type_2\in \set{\full, \fair, \idfair,\abort,\idabort}$.
Let $f$ be an $n$-party functionality.
Let $\rho$ be a protocol that $(\delta_1,t)$-securely computes $f$ with $\type_1$, and let $\pi$ be a protocol that $(\delta_2,t)$-securely computes $g$ with $\type_2$ in the $(f,\type_1)$-hybrid model, using $q$ calls to the ideal functionality.
Then protocol $\pi^\rho$, that is obtained from $\pi$ by replacing all ideal calls to the trusted party computing $f$ with the protocol $\rho$, is a protocol that $((q\cdot\delta_1+\delta_2),t)$-securely computes $g$ with $\type_2$ in the real model.
\end{proposition}

\subsection{Committee Election}\label{sec:Feige}
Feige's lightest-bin protocol~\cite{Feige99} is an elegant $n$-party, public-coin protocol, consisting of a single broadcast round, for electing a committee of size $n'<n$, in the information-theoretic setting.
Each party uniformly selects one of $\ceil{n/n'}$ bins and broadcasts it choice. The parties that selected the lightest bin are elected to participate in the committee. The protocol ensures that the ratio of corrupted parties in the elected committee is similar to their ratio in the population.
The original protocol in~\cite{Feige99} considered committees of size $\log(n)$, however, this results with a non-negligible failure probability.
\citet[Lem.\ 2.6]{BGK11} analyzed Feige's protocol for arbitrary committee sizes and proved the following lemma.

\begin{lemma}[\cite{BGK11}]\label{lem:Feige}
For integers $n'<n$ and constants $0<\beta<\beta'<1$, denote
\[
\err\left(n,n',\beta,\beta'\right)=\frac{n}{n'} \cdot e^{-\frac{(\beta'-\beta)^2 n'}{2(1-\beta)}}.
\]
Feige's lightest-bin protocol is a $1$-round, $n$-party protocol for electing a committee $\committee$, such that for any set of corrupted parties $\IS\subseteq [n]$ of size $t=\beta n$, the following holds.
\begin{enumerate}
    \item
    $\size{\committee}\leq n'$.
    \item
    $\pr{\size{\committee\setminus\IS}\leq (1-\beta')\cdot n'} < \err(n,n',\beta,\beta')$.
    \item
    $\pr{\size{\committee\cap \IS}\geq \beta' \cdot \size{\committee}} < \err(n,n',\beta,\beta')$.
\end{enumerate}
\end{lemma}

\noindent
The committee selected by the lightest-bin protocol of Feige has the desired fraction of honest parties except for probability $\err(n,n',\beta,\beta')$. We capture this guarantee by defining the following ideal functionality.

\begin{definition}[the committee-election functionality]\label{def:func-felect}~
\begin{itemize}
\item[\emph{Input:}]
In the committee-election functionality $\felect(n,n',\beta')$ parties have no private input. The common input \rchanged{include}{includes} the number of parties $n$, an upper bound $n'$ on the size of the elected committee (a subset of parties), and a bound on the fraction of corrupted parties in that committee.\footnote{\rchanged{The functionality knows the subset $\I$ of corrupted parties.}{We note that the functionality is corruption aware, \ie its actions are depended on knowing the subset $\I$ of corrupted parties. This is standard and does not pose any restrictions since Feige's protocol does not require knowing $\I$.}}

\item[\emph{Computation:}]
Let $k=n/n'$, the functionality defines $k$ subsets $\committee_1,\ldots,\committee_k$ (all initially empty). The functionality randomly selects $k\cdot \left\lceil (1-\beta')n'\right\rceil$ honest parties and partitions them evenly into the $k$ subsets. Next, the functionality assigns each of the remaining honest parties a random subset $\committee_i$. Then, the functionality informs the adversary of the resulting partition.

\item[\emph{Adversary's choice:}]
The adversary chooses a committee $\committee_i$ for some $i\in[k]$ and a subset $\JS\subseteq \I$, such that $\size{\committee_i}+\size{\JS}\le n'$.

\item[\emph{Output:}]
The output of the functionality is the subset $\committee = \committee_i \cup \JS$.
\end{itemize}
\end{definition}

%
%

The following corollary follows immediately from \cref{lem:Feige}.
\begin{corollary}\label{cor:elect}
Let $n'<n$ be integers, let $0<\beta<\beta'<1$ be constants, and let $t=\beta n$. Then, Feige's lightest-bin protocol $(\err(n,n',\beta,\beta'),t)$-securely computes $\felect(n,n',\beta')$ with full security. In particular, for $n'=\omega(\log(\secParam))$ Feige's protocol $t$-securely computes $\felect(n,n',\beta')$ with full security.
\end{corollary}


\subsection{Fairness with Restricted Identifiable Abort (With Inputs)}\label{sec:restricted_idabort}
Delegating computation to a small committee will be a useful technique throughout this work. In such a computation, we wish to allow non-members of the committee to monitor the execution of the protocol by committee members; however, non-members should never be able to disrupt the execution themselves.
To capture the required security, we introduce a variant of fairness with identifiable abort that will be used as an intermediate step in our constructions.

This definition captures the delegation of the computation to smaller committees that independently carry out the (same) fair computation, such that the adversary can only abort the computation of committees with corrupted parties.

In \cref{sec:fairtofull_noinput}, we use this security notion for the case of no-input functionalities. Clearly, this is a special case captured by the general definition.
Nevertheless, for clarity, we specify the no-input variant of this notion in \cref{sec:restricted_idabort_noinput}.
We first present a variant of the definition that does not require fairness, which, looking ahead, will turn out to be useful in some of the applications.

\paragraph{Ideal model with restricted identifiable abort.}
An ideal computation, with $\vCS$-identifiable-abort, of an $n$-party functionality $f$ on input $\vx=(x_1,\ldots,x_n)$ for parties $(\Party_1,\ldots,\Party_n)$ \wrt $\vCS=(\committee_1,\ldots,\committee_\ell)$, where $\committee_1,\ldots,\committee_\ell\subseteq [n]$, in the presence of an ideal-model adversary $\Adv$ controlling the parties indexed by $\IS\subseteq[n]$, proceeds via the following steps.
\begin{itemize}
    \item[\emph{Sending inputs to trusted party}:]
    An honest party $\Party_i$ sends its input $x_i$ to the trusted party.
    The adversary may send to the trusted party arbitrary inputs for the corrupted parties. Let $x_i'$ be the value actually sent as the input of party $\Party_i$.

	\item[\emph{Early abort}:]
    If there exists a corrupted party in every subset $\committee_j$, \ie if $\IS\cap\committee_j\neq \emptyset$ for every $j\in[\ell]$, then the adversary $\Adv$ can abort the computation by choosing an index of a corrupted party $\is_j\in\IS \cap \committee_j$ for every $j\in [\ell]$ and sending the abort message $(\abort,\sset{\is_1,\ldots,\is_\ell})$ to the trusted party. In case of such abort, the trusted party sends the message $(\bot,\sset{\is_1,\ldots,\is_\ell})$ to all parties and halts.

    \item[\emph{Trusted party answers adversary}:]
    If $\committee_j\subseteq\IS$ for some $j\in[\ell]$, the trusted party sends all the input values $x'_1,\ldots,x'_n$ to the adversary, waits to receive from the adversary output values $y'_1,\ldots,y'_n$, sends $y'_i$ to $\Party_i$ and proceeds to the \emph{Outputs} step.
    Otherwise, the trusted party computes $(y_1, \ldots, y_n)=f(x_1', \ldots, x_n')$ and sends $y_i$ to party $\Party_i$ for every $i\in\IS$.

	\item[\emph{Late abort}:]
    If there exists a corrupted party in every subset $\committee_j$, then the adversary $\Adv$ can abort the computation (\emph{after seeing the outputs of corrupted parties}) by choosing an index $\is_j\in\IS \cap \committee_j$ for every $j\in [\ell]$ and sending the abort message $(\abort,\sset{\is_1,\ldots,\is_\ell})$ to the trusted party. In case of such abort, the trusted party sends the message $(\bot,\sset{\is_1,\ldots,\is_\ell})$ to all parties and halts. Otherwise, the adversary sends a $\continue$ message to the trusted party.
		
    \item[\emph{Trusted party answers remaining parties}:]
	The trusted party sends $y_i$ to $\Party_i$ for every $i\notin\IS$.

    \item[\emph{Outputs}:]
    Honest parties always output the message received from the trusted party and the corrupted parties output nothing.
    The adversary $\Adv$ outputs an arbitrary function of the initial inputs $\set{x_i}_{i\in\IS}$, the messages received by the corrupted parties from the trusted party and its auxiliary input.
\end{itemize}

\begin{definition}[ideal-model computation with restricted identifiable abort]\label{def:ideal_ridabort}
Let $f\colon(\zs)^n \mapsto (\zs)^n$ be an $n$-party functionality, let $\IS\subseteq [n]$, and let $\vCS=(\committee_1,\ldots,\committee_\ell)$, where $\committee_1,\ldots,\committee_\ell\subseteq [n]$.
The {\sf joint execution of $f$ with $\vCS$ under $(\Adv, I)$ in the ideal model}, on input vector $\vx=(x_1, \ldots, x_n)$, auxiliary input $\aux$ to $\Adv$, and security parameter $\secParam$, denoted $\IDEAL^{\vCS\mhyphen\idabort}_{f,\IS,\Adv(\aux)}(\vx,\secParam)$, is defined as the output vector of $\Party_1, \ldots, \Party_n$ and $\Adv(\aux)$ resulting from the above described ideal process.
\end{definition}

\noindent
To keep notation short, in case $\vCS = \set{\committee_1}$, i.e., $\ell=1$, we denote $\committee_1\mhyphen\idabort$ instead of $\vCS\mhyphen\idabort$.
The ideal model presented above defines security with $\vCS$-identifiable-abort. We define the fair variant of this ideal computation as follows:

\paragraph{Ideal model with fairness and $\vCS$-identifiable-abort.}
This ideal model proceeds as in \cref{def:ideal_ridabort} with the exception that in step \emph{Late abort}, the adversary is not allowed to abort the computation. This ideal computation is denoted as $\IDEAL^{\vCS\mhyphen\idfair}_{f,\IS,\Adv(\aux)}(\vx,\secParam)$.

\paragraph{Security definitions.}
We present the security definition according to the ideal model computing $f$ with fairness and $\vCS$-identifiable-abort. The definitions for security with $\vCS$-identifiable-abort follow in a similar way.

\begin{definition}\label{def:SecureProtocol_ridfair}
Let $f\colon(\zs)^n \mapsto (\zs)^n$ be an $n$-party functionality and let $\pi$ be a probabilistic polynomial-time protocol computing $f$.
The {\sf protocol $\pi$ $(\delta,t)$-securely computes $f$ with fairness and $(\ell, n',t')$-identifiable-abort (and computational security)}, if for every probabilistic polynomial-time real-model adversary \Adv, there exists a probabilistic polynomial-time adversary $\Sim$ for the ideal model, such that for every $\IS\subseteq [n]$ of size at most $t$ and subsets $\committee_1,\ldots,\committee_\ell\subseteq [n]$ satisfying $\ssize{\committee_j}=n'$ and $\ssize{\IS\cap\committee_j}\leq t'$, for every $j\in[\ell]$, it holds that
\[
\set{\bigbrack \REAL_{\pi, \IS, \Adv(\aux)}(\vx, \secParam)}_{(\vx, \aux)\in(\zs)^{n+1}, \secParam\in\N}
\deltaci
\set{\bigbrack \IDEAL^{(\committee_1,\ldots,\committee_\ell)\mhyphen\idfair}_{f, \IS, \Sim(\aux)}(\vx, \secParam)}_{(\vx, \aux)\in(\zs)^{n+1}, \secParam\in\N}.
\]
If $\delta$ is negligible, we say that $\pi$ is a protocol that $t$-securely computes $f$ with fairness and $(\ell, n',t')$-identifiable-abort and computational security.

The {\sf protocol $\pi$ $(\delta,t)$-securely computes $f$ with fairness and $(\ell, n',t')$-identifiable-abort (and
statistical security)}, if for every real-model adversary \Adv, there exists an adversary $\Sim$ for the ideal model, whose running time is polynomial in the running time of $\Adv$, such that for every $\IS\subseteq [n]$ of size at most $t$, and subsets $\committee_1,\ldots,\committee_\ell\subseteq [n]$ satisfying $\ssize{\committee_j}=n'$ and $\ssize{\IS\cap\committee_j}\leq t'$, for every $j\in[\ell]$, it holds that
\[
\set{\bigbrack \REAL_{\pi, \IS, \Adv(\aux)}(\vx, \secParam)}_{(\vx, \aux)\in(\zs)^{n+1}, \secParam\in\N}
\deltaclose
\set{\bigbrack \IDEAL^{(\committee_1,\ldots,\committee_\ell)\mhyphen\idfair}_{f, \IS, \Sim(\aux)}(\vx, \secParam)}_{(\vx, \aux)\in(\zs)^{n+1}, \secParam\in\N}.
\]
If $\delta$ is negligible, we say that $\pi$ is a protocol that $t$-securely computes $f$ with fairness and $(\ell, n',t')$-identifiable-abort and statistical security.

Similarly, {\sf $\pi$ is a protocol that $(\delta,t)$-securely computes $f$ with fairness and $(\ell, n',t')$-identifiable-abort (and perfect security)}, if
\[
\set{\bigbrack \REAL_{\pi, \IS, \Adv(\aux)}(\vx, \secParam)}_{(\vx, \aux)\in(\zs)^{n+1}, \secParam\in\N}
\deltaequiv
\set{\bigbrack \IDEAL^{(\committee_1,\ldots,\committee_\ell)\mhyphen\idfair}_{f, \IS, \Sim(\aux)}(\vx, \secParam)}_{(\vx, \aux)\in(\zs)^{n+1}, \secParam\in\N}.
\]
If $\delta=0$, we say that $\pi$ is a protocol that $t$-securely computes $f$ with fairness and $(\ell, n',t')$-identifiable-abort and perfect security.
\end{definition}

\subsubsection{Ideal functionalities delegated to committees}\label{sec:ideal_functionalities}
We prove \radded{some of} our constructions \radded{(\cref{prot:fair_to_ridfair_noinput,prot:fair_to_ridfair})} in a hybrid model that computes (a variant of) the augmented coin-tossing functionality $\faugct$ (see \cite{Goldreich04,CLOS02}) and the one-to-many zero-knowledge proof of knowledge functionality $\zkmany$ (see \cite{CLOS02,Pass04}).

\begin{definition}[delegated augmented coin-tossing]\label{def:augFunc}
The augmented coin-tossing functionality, denoted $\faugct^{\committee}$, is an $n$-party no-input functionality. It is parametrized by a subset $\committee\subseteq[n]$ of size $n'$ and a commitment scheme $\Com$. The output of each party $\Party_i\in\committee$ in this functionality is a private random string $\rndaugct_i$ together with a decommitment information $\decomval_i$, and in addition, all $n$ parties receive as public output a vector of commitments to the random strings. That is,
\[
\faugct^{\committee}(\lambda,\ldots,\lambda) = ((y_1,\vcomaugct),\ldots,(y_n,\vcomaugct)),
\text{ where }
y_i = \left\{
        \begin{array}{ll}
          (\rndaugct_i,\decomaugct_i), & \text{if } \Party_i\in\committee \\
          \lambda, & \text{otherwise}
        \end{array}
      \right.
\]
and $\vcomaugct=(\comaugct_1,\ldots,\comaugct_{n'})$ such that $\comaugct_j=\Com(\rndaugct_j;\decomaugct_j)$ for every $\Party_j\in\committee$. 
\end{definition}

\begin{definition}[one-to-many zero-knowledge]\label{def:zkmany}
The one-to-many zero-knowledge proof of knowledge functionality $\zkmany$ for a language $L$ is an $n$-party functionality with a special party, called the prover. The common input of all parties is a statement $x$. The prover sends a statement-witness pair $(x,w)$ to the functionality.  The common output is $(x,1)$ if $(x,w)\in R_L$ and $(x,0)$ otherwise.

\end{definition}

\begin{lemma}
\label{lem:instantiating_zkmanmy_noinput}
Assume that \TDP and \CRH exist and let \rchanged{$m=m(\secParam)$ be an a priori fixed polynomial}{$t<n$}.
Then,
\begin{enumerate}
    \item
    $\zkmany$ can be $t$-securely computed with full security using a constant-round protocol.
    \item
    Let $\committee\subseteq[n]$, then $\faugct^{\committee}$ can be $t$-securely computed with $\committee$-identifiable-abort using a constant-round protocol.
\end{enumerate}
\end{lemma}
\begin{proof}
The first part of the lemma follows from \citet{Pass04}, who showed that under the assumptions in the lemma there exist constant-round protocols for a simulation-sound,\footnote{Zero-knowledge protocols are simulation sound if the soundness of each of the protocols is preserved even when the other protocol is simulated at the same time with the roles of the prover and verifier reversed.} bounded-concurrent composition of $\zkmany$. We consider the variant of the protocol by \citet{BeimelOO15} that adjusted the protocol from~\cite{Pass04}, originally designed for an asynchronous point-to-point network, to the synchronous setting where all messages are sent over a broadcast channel.

We now show that $\faugct^{\committee}$ can be securely computed with $\committee$-identifiable abort in the $\zkmany$-hybrid model. We use the protocol from \citet{CLOS02} with a small adjustment. In order to obtain a committed random string for $\Party'_j\in\committee$, all parties in $\committee$ proceed as follows.
\begin{enumerate}
    \item
    Every party $\Party'_i\in\committee$ uniformly samples a random string $r_{ij}$, computes $\comval_{ij}=\Com(r_{ij};\decomval_{ij})$, and sends $(\comval_{ij},(r_{ij},\decomval_{ij}))$ to $\zkmany$ (parameterized with the NP-relation consisting of statement-witness pairs $(\comval,(r,\decomval))$ that satisfy $\comval=\Com(r;\decomval)$). In case $\zkmany$ returns $(\comval,0)$ for some $\is\in\committee$, all parties output $(\bot,\is)$ (for the smallest such $\is$) and halt.
    \item
    Every party $\Party'_i\in\committee$, except for $\Party'_j$, broadcasts $(r_{ij},\decomval_{ij})$, and all parties verify that $\comval_{ij}=\Com(r_{ij};\decomval_{ij})$; if not the parties output $(\bot,i)$ (for the smallest such $i$) and halt.
    \item
    Party $\Party'_j$ computes $r_j=\bigoplus_i r_{ij}$ and outputs $(r_j,(r_{jj},\decomval_{jj}))$ (\ie the random string is $r_j$ and the decommitment is $(r_{jj},\decomval_{jj})$), whereas all other parties compute $r_j^\ast=\bigoplus_{i\neq j} r_{ij}$ and output $(\comval_{jj}, r_j^\ast)$ as the commitment for $r_j$.
\end{enumerate}
The proof of the protocol follows from~\cite{CLOS02}.
\end{proof}

\section{Fairness to Full Security for No-Input Functionalities}\label{sec:fairtofull_noinput}

In this section, we present a reduction from a fully secure computation to a fair computation for functions without inputs (\eg coin flipping).
This serves as a first step before presenting the more complex case of functionalities with input in \cref{sec:fair_to_full_withinputs}.
In \cref{sec:fairtofull_noinput_noHM}, we consider a reduction that does not assume an honest majority, and in \cref{sec:fairtofull_noinput_HM} a more efficient reduction in the honest-majority setting.
In \cref{sec:apps_noinput}, we show applications regarding coin-flipping protocols.

Our reductions are two-phased. Initially, we show how to reduce full security to fairness with restricted identifiable abort (defined in \cref{sec:restricted_idabort}, see \cref{sec:restricted_idabort_noinput} for the no-input case) in a round-efficient manner. Next, we show how to reduce fairness with restricted identifiable abort to fairness.

Recall that if a no-input functionality $f$ has public output, then it can be defined for any number of parties. We denote by $f^n$ the functionality $f$ when defined for $n$ parties, and show how to compile any fair protocol computing $f^{n'}$ to a protocol that fairly computes $f^n$ with restricted identifiable abort (for $n'<n$).
For integers $n'<n$ and for $0<\beta<\beta'<1$ we define $\err(n,n',\beta,\beta')=\frac{n}{n'} \cdot e^{-\frac{(\beta'-\beta)^2 n'}{2(1-\beta)}}$.
In addition, denote by $\foutss{f^n}{t}{n}$ the $n$-party functionality that computes $f^n$ and outputs shares of the result using a $(t+1)$-out-of-$n$ error-correcting secret-sharing scheme (ECSS, see \cref{def:ECSS}).
We prove the following theorem.

\begin{theorem}[restating \cref{thm:intro:mainInlessfair_to_full}]\label{thm:mainThm_fairtofull_noinput}
Assume that \TDP and \CRH exist.
Let $f$ be a no-input functionality with public output, let $n'<n$ be integers, let $0<\beta<\beta'<1$, and let $t=\beta n$ and $t'=\beta' n'$.
\begin{enumerate}
    \item
    If $f^{n'}$ can be \rchanged{$(\delta,t')$}{$(\delta',t')$}-securely computed with fairness by an \rchanged{$r$}{$r'$}-round protocol, then $f^n$ can be \rchanged{$(t'\cdot\delta + \err(n,n',\beta,\beta'),t)$}{$(t'\cdot\delta' + \err(n,n',\beta,\beta'),t)$}-securely computed with full security by an \rchanged{$O(t'\cdot r)$}{$O(t'\cdot r')$}-round protocol.
    \item \label{thm:mainThm_fairtofull_noinput_step2}
    If $f^{n'}$ can be \rchanged{$(\delta,t')$}{$(\delta',t')$}-securely computed with full security by an \rchanged{$r$}{$r'$}-round protocol, then $f^n$ can be \rchanged{$(\delta + \err(n,n',\beta,\beta'),t)$}{$(\delta' + \err(n,n',\beta,\beta'),t)$}-securely computed with full security by an \rchanged{$O(t'+r)$}{$O(t'+r')$}-round protocol.
    \item
    For $\beta'<1/2$ and $n'=\min(n,\log(\secParam)\cdot\varphi(\secParam))$ with $\varphi=1/\sqrt{1-2\beta'}+\Omega(1)$,\footnote{By $\varphi=1/\sqrt{1-2\beta'}+\Omega(1)$ we mean that for sufficiently large $\secParam$ it holds that $\varphi(\secParam)>1/\sqrt{1-2\beta'}$.} the following holds unconditionally.
    If $\foutss{f^{n'}}{t'}{n'}$ can be \rchanged{$(\delta,t')$}{$(\delta',t')$}-securely computed with abort by an \rchanged{$r$}{$r'$}-round protocol, $\ell$-times in parallel, for $\ell=\secParam^\uglyExp$, then $f^n$ can be \rchanged{$(\varphi(\secParam)^2\cdot\ell\cdot\delta + \err(n,n',\beta,\beta'),t)$}{$(\varphi(\secParam)^2\cdot\ell\cdot\delta' + \err(n,n',\beta,\beta'),t)$}-securely computed with full security by an \rchanged{$O(\varphi(\secParam)^2\cdot r)$}{$O(\varphi(\secParam)^2\cdot r')$}-round protocol.
\end{enumerate}
\end{theorem}

The proof of \cref{thm:mainThm_fairtofull_noinput} is given in the sections below, where the first part follows from a combination of \cref{thm:ridfair_to_full_noinput} and \cref{thm:idfair_to_ridfair_noinput}; the second part from \cref{thm:ridfair_to_full_noinput} and \cref{cor:full_to_full_noinput}; and the third part from \cref{thm:ridfair_to_full_parallel_noinput} and \cref{thm:idfair_to_ridfair_noinput_HM}.

\subsection{Fairness to Full Security without an Honest Majority (No Inputs)}\label{sec:fairtofull_noinput_noHM}
We now   present a reduction from full security to fairness for no-input functionalities when an honest majority is not assumed.
In \cref{sec:ridfair_to_full_noinput_noHM}, we show how to compute $f^n$ with full security in the hybrid model computing $f^n$ with fairness and restricted identifiable abort.
In \cref{sec:fair_to_ridfair_noinput_noHM}, we show how to compile a fair protocol for $f^{n'}$ to a fair protocol for $f^n$ with restricted identifiable abort.

\subsubsection{Fairness with Restricted Identifiable Abort to Full Security}\label{sec:ridfair_to_full_noinput_noHM}

We start by showing how to reduce full security to fairness with restricted identifiable abort. A single committee $\committee$ is considered in this setting (\ie $\ell=1$). The idea is quite simple: initially, a committee $\committee$ is elected using Feige's lightest-bin protocol~\cite{Feige99} such that the ratio of corrupted parties in the committee is approximately the same as in the original party-set (see \cref{sec:Feige}).
Next, the parties sequentially call the fair computation with $\committee$-identifiable-abort \radded{(where only corrupted parties in $\committee$ can abort the computation, in which case they are identified by all parties)}, until receiving the output.

\begin{theorem}\label{thm:ridfair_to_full_noinput}
Let $f$ be a no-input, $n$-party functionality with public output, let $n'<n$, let $0<\beta<\beta'<1$, and let $t=\beta n$ and $t'=\beta' n'$.
Then, $f$ can be $(\err(n,n',\beta,\beta'),t)$-securely computed with full security in a hybrid model that computes $f$ with fairness and $(n',t')$-identifiable-abort, by using $t'+1$ sequential calls to the ideal functionality.
\end{theorem}

\begin{proof}
We present the protocol in the hybrid model that computes the committee-election functionality $\felect$ with full security ($\felect$ is defined in \cref{sec:Feige}).

\begin{protocol}(fairness with restricted identifiable abort to full security (no inputs))\label{prot:ridfair_to_full_noinput}
\begin{itemize}
    \item\textbf{Hybrid Model:}
    The protocol is defined in the hybrid model computing $\felect$ with full security, and $f$ with fairness and $(n',t')$-identifiable-abort.
    \item\textbf{Common Input:}
    The values $t',n'\in\N$.
    \item\textbf{The Protocol:}
\end{itemize}
\begin{enumerate}
    \item \label{step:rid_elect_noinput}
    All the parties invoke $\felect(n,n',\beta')$ and elect a committee $\committee_1\subset[n]$ of size $n'$.
    \item
    For $i=1,\ldots,t'+1$ do
    \begin{enumerate}
        \item
        All parties that have not been previously identified call the trusted party computing $(f,\committee_1\mhyphen\idfair)$,
        where the party with the lowest index in $\committee_i$ simulates all parties in $\committee_1\setminus \committee_i$.
        Denote the output $\Party_j$ receives by $y_j$.
        \item
        Every party $\Party_j$ checks if $y_j$ is a valid output, if so $\Party_j$ outputs $y_j$ and halts.
        Otherwise, all parties received $(\bot,\is)$ as output, where $\is\in\committee_1\cap\IS$.
        If $\is\notin\committee_i$ (and so $\Party_{\is}$ is a previously identified corrupted party), then all parties set $\is$ to be the lowest index in $\committee_i$.
        \item
        All parties set $\committee_{i+1}=\committee_i\setminus\set{\is}$.
    \end{enumerate}
\end{enumerate}
\end{protocol}

Let $\Adv$ be an adversary attacking \cref{prot:ridfair_to_full_noinput} and let $\IS$ be the set of corrupted parties.
We construct a simulator $\Sim$ for the ideal model computing $f$ with full security, as follows.
$\Sim$ starts by emulating $\Adv$ on its auxiliary input $\aux$.
Initially, $\Sim$ emulates the committee-election functionality $\felect(n,n',\beta')$; that is, $\Sim$ partitions the honest parties to $n/n'$ random subsets, each containing at least $(1-\beta')\cdot n$ parties, hands them to $\Adv$ and receives back a committee $\committee_1$ of size $n'$ that contains exactly one of the subsets (and does not intersect any of the others).

Next, in every iteration, $\Sim$ simulates towards $\Adv$ the computation of $f$ with $\committee_1$-identifiable-abort.
If $\Adv$ sends the message $(\abort,\is)$ in the \ith iteration, with $\is\in\IS\cap\committee_1$, then $\Sim$ checks if $\Party_{\is}$ has not been previously identified (otherwise set $\is$ to be the smallest element in $\committee_i$), simulates sending the response $(\bot,\is)$ to all parties, sets $\committee_{i+1}=\committee_i\setminus\sset{\is}$ and proceeds to the next iteration.
In the first iteration in which no \abort\ is sent, $\Sim$ calls the trusted party in the fully secure ideal model computing $f$.
Upon receiving the output from its trusted party, $\Sim$ hands it to $\Adv$ as if it were the output of the corrupted parties in the iteration of $\pi$, and outputs whatever $\Adv$ outputs.

The simulation in the hybrid model is perfect since $\Sim$ can perfectly simulate the trusted party for all iterations in which $\abort$ is sent.
Furthermore, in the first iteration for which $\abort$ is not sent, $\Sim$ sends to \Adv the output of the function $f$ as computed in the protocol. After $t'$ iterations it is guaranteed that there are no corrupted parties left in $\committee_1$ and the protocol will complete. By \cref{cor:elect}, Feige's lightest-bin protocol $(\err(n,n',\beta,\beta'),t)$-securely computes $\felect(n,n',\beta')$ with full security and the theorem follows.
\end{proof}

\subsubsection{Fairness to Fairness with Restricted Identifiable Abort}\label{sec:fair_to_ridfair_noinput_noHM}

We next present a reduction from a fair computation with restricted identifiable abort of $f^n$ to a fair computation of $f^{n'}$. More specifically, let $\pi'$ be a fair protocol computing $f^{n'}$ by a subset of $n'$ parties $\committee$. We show that $\pi'$ can be compiled into a protocol $\pi$ that computes $f$ with fairness and $\committee$-identifiable-abort. The underlying idea is to let the committee $\committee$ prove that every step in the execution is correct (in a similar way to the GMW compiler~\cite{GoldreichMW87}) such that when $\pi'$ terminates the parties in $\committee'$ either obtain the output or identify a corrupted party. Next, every party in the committee broadcasts the result and proves that it is indeed the correct result to all $n$ parties.

The above is formally stated in the theorem below, using the following notations. Let $f$ be a no-input functionality with public output, let $t,n'<n$, let $t'<n'$, and let $\committee\subseteq[n]$ of size $n'$.

\begin{theorem}\label{thm:idfair_to_ridfair_noinput}
Assume that \TDP and \CRH exist, and let $f$ be a no--input functionality with public output.
Then, there exists a \ppt algorithm $\CompilerNHMNI{n'}{n}$ such that for any $n'$-party, \rchanged{$r$}{$r'$}-round protocol $\pi'$ computing $f^{n'}$, the protocol $\pi = \CompilerNHMNI{n'}{n}(\pi',\committee)$ is an $n$-party, \rchanged{$O(r)$}{$O(r')$}-round protocol computing $f^n$ with the following guarantee.
If the number of corrupted parties in $\committee$ is at most $t'$, and $\pi'$ is a protocol that \rchanged{$(\delta,t')$}{$(\delta',t')$}-securely computes $f^{n'}$ with fairness, then $\pi$ is a protocol that \rchanged{$(\delta,t)$}{$(\delta',t)$}-securely computes $f^n$ with fairness and $\committee$-identifiable-abort.
\end{theorem}

\begin{proof}
We construct (in \cref{prot:fair_to_ridfair_noinput}) the protocol compiler $\CompilerNHMNI{n}{n'}(\pi',\committee)$ in the $(\faugct,\zkmany)$-hybrid model (see \cref{def:augFunc,def:zkmany}), and prove the security properties of this compiler in~\cref{lem:idfair_to_ridfair_noinput}. \cref{lem:instantiating_zkmanmy_noinput}, proves that the ideal functionalities $\faugct$ and $\zkmany$ can be instantiated in a round-preserving manner.
The proof follows from the sequential composition theorem (\cref{prop:Composition}).
\end{proof}

\paragraph{Compiling a fair protocol to a protocol with fairness and restricted identifiable abort.}
Towards proving \cref{thm:idfair_to_ridfair_noinput}, we next describe the compiler from fairness to fairness and restricted identifiable abort. We begin by making two remarks regarding the construction.

\begin{remark}[fairness to fairness with identifiable abort]\label{remark:compiler-restrictions}
The compiler described in \cref{prot:fair_to_ridfair_noinput} assumes that the protocol $\pi'$ (the input of the compiler) is fair with identifiable abort and that all communication in the protocol $\pi'$ is sent over the broadcast channel. These assumptions are \wlg. Following~\cite[Lem.\ 3]{CL17} any fair protocol can be compiled in a round-preserving manner into a protocol that provides fairness with identifiable abort, tolerating any number of corrupted parties, in the $(\faugct,\zkmany)$-hybrid model.\footnote{The original proof in~\cite{CL17} assumes one-way functions and is not round preserving, however, using the techniques from~\cite{Pass04}, the compilation will blow-up the round complexity only by a constant factor.}
\end{remark}

\begin{remark}[Using functionality $\zkmany$ in the compiler]
We will consider the relation $R_j$ (for the \jth party in $\committee$), parametrized by a commitment scheme $\Com$, that contains pairs $((\vm,\comaugct,\outvalue),(\rndaugct,\decomaugct))$, and validates that $\outvalue$ is the output value of protocol $\pi'$ using randomness $\rndaugct$ and messages $\vm=(m_1,\ldots,m_p)$, and that $\comaugct=\Com(\rndaugct;\decomaugct)$.
\end{remark}
The $n$-party protocol $\pi = \CompilerNHMNI{n'}{n}(\pi',\committee)$ is defined as follows.
\begin{construction}(fairness to fairness with restricted identifiable abort)\label{prot:fair_to_ridfair_noinput}
\begin{itemize}
    \item\textbf{Hybrid Model:}
    The protocol is defined in the hybrid model computing $\faugct$ with restricted identifiable abort and $\zkmany$ with full security.
    \item\textbf{Common Input:}
    A subset $\committee\subseteq[n]$ of size $n'$ and an $n'$-party protocol $\pi'$, computing the functionality $f^{n'}$ with fairness and identifiable abort, using only a broadcast channel (see \cref{remark:compiler-restrictions}). We use the notation $\Party'_j$ to refer to the \jth party in $\committee$.
    \item\textbf{The Protocol:}
\end{itemize}
\begin{enumerate}
    \item\label{step:rid_coinflip_noinput}
    All parties invoke $\faugct^{\committee}$ with $\committee$-identifiable-abort, every $\Party'_j\in\committee$ receives back $(\rndaugct_j,\decomaugct_j,\vcomaugct)$ where $\vcomaugct=(\comaugct_1,\ldots,\comaugct_{n'})$ is common to all $n$ parties.
    In case the computation for $\committee$ aborts with the identity of party $\Party_{\is}=\Party'_{\js}\in\committee$, all $n$ parties output $(\bot,\is)$ and halt.

    \item\label{step:rid_gmw_noinput}
    The parties in $\committee$ execute the protocol $\pi'$ for computing $f^{n'}$ over the broadcast channel, where $\Party'_j$ uses $\rndaugct_j$ as its random coins.
    Let $\outvalue_j$ be the output $\Party'_j$ received (either a valid value $y$ or $(\bot,\is)$ with $\Party_{\is}=\Party'_{\js}\in\committee$). Denote by $\vm_j=(m_1^j,\ldots,m_p^j)$ the messages $\Party'_j$ received during the protocol.

    \item\label{step:rid_output_noinput}
    Every $\Party'_j\in\committee$ invokes $\zkmany$ and proves to all parties that $\outvalue_j$ is indeed the correct output value generated by $\pi'$ using the committed randomness $r_j$ and messages $\vm_j$, \ie $\Party'_j$ sends $((\vm_j,\comaugct_j,\outvalue_j),(\rndaugct_j,\decomaugct_j))$ to $\zkmany$ (parametrized by the relation $R_j$).
    Once party $\Party_i$ receives from $\zkmany$ an accepting proof for $\Party'_j$'s output value $\outvalue_j$, party $\Party_i$ outputs this value (invalid proofs are ignored).
\end{enumerate}
\end{construction}

\begin{lemma}\label{lem:idfair_to_ridfair_noinput}
Assume that commitment schemes exist and consider the same notations as in \cref{thm:idfair_to_ridfair_noinput}.
If $\pi'$ is a \rchanged{$(\delta,t')$}{$(\delta',t')$}-secure protocol computing $f^{n'}$ with fairness, then the protocol $\pi = \CompilerNHMNI{n'}{n}(\pi',\committee)$ is an $n$-party protocol that \rchanged{$(\delta,t)$}{$(\delta',t)$}-securely computes $f^n$ with fairness and $\committee$-identifiable-abort, in the $(\faugct,\zkmany)$-hybrid model.
\end{lemma}

\begin{proof}
Let $\Adv$ be an adversary attacking the execution of protocol $\pi$ in the $(\faugct,\zkmany)$-hybrid model and let $\IS\subseteq[n]$ be a subset of size at most $t$, satisfying $\ssize{\IS\cap\committee}\leq t'$.
We construct the following adversary $\Sim$ for the ideal model computing $f$ with fairness and $\vCS$-identifiable-abort.
$\Sim$ starts by emulating $\Adv$ on the auxiliary input $\aux$.
The simulator interacts with $\Adv$, playing the roles of the honest parties and the ideal functionalities $\faugct$ and $\zkmany$.

To simulate Step~\ref{step:rid_coinflip_noinput}, the simulator $\Sim$ plays $\faugct^{\committee}$ honestly by sampling random strings $(\rndaugct_j,\decomaugct_j)$ for parties $\Party'_j\in\committee$, computing $\comaugct_j=\com(\rndaugct_j;\decomaugct_j)$ and setting $\vcomaugct=(\comaugct_1,\ldots,\comaugct_{n'})$. Next, \Sim hands $(\rndaugct_j,\decomval_j,\vcomaugct)$ to corrupted parties $\Party'_j\in\committee$ and $(\lambda,\vcomaugct)$ to corrupted parties outside of $\committee$. In case $\Sim$ receives $(\abort,\is)$ with $\is\in\IS\cap\committee$ from $\Adv$, it forwards $(\abort,\is)$ to the trusted party, responds with $(\bot,\is)$ to \Adv, outputs whatever \Adv outputs and halts.

Next, the simulator $\Sim$ uses the simulator $\tilde{\Sim}$ that is guaranteed to exist for $\pi'$ when interacting with the residual adversary of $\Adv$ in Step~\ref{step:rid_gmw_noinput} \radded{(\ie the adversary against the protocol $\pi'$ that is induced from the behaviour of $\Adv$)}.
The simulator $\Sim$ invokes $\tilde{\Sim}$ on auxiliary information containing $z$ and the view of the adversary in the simulation until this point.
If $\tilde{\Sim}$ sends $(\abort,\is)$ with $\is\in\IS\cap\committee$, the simulator forwards $(\abort,\is)$ to the trusted party and sets $\outvalue=(\bot,\is)$.
Otherwise, \Sim sends empty strings as the input of the corrupted parties, receives back the output $y$ and sets $\outvalue=y$.
Next, The simulator $\Sim$ forwards $y$ to $\tilde{\Sim}$ as the output of the computation, receives back the output from $\tilde{\Sim}$, which contains the view of the adversary, and interacts with $\Adv$ accordingly.

To simulate Step~\ref{step:rid_output_noinput}, the simulator $\Sim$ simulates $\zkmany$. The simulator sends, on behalf of every honest party $\Party'_j\in\committee$, the message $((\vm_j,\comaugct_j,\outvalue),1)$ to every corrupted party, where $\vm_j$ is obtained from the output of $\tilde{\Sim}$. In addition, $\Sim$ receives $((\vm'_j,\comaugct'_j,\outvalue'_j),(\rndaugct'_j,\decomaugct'_j))$ from $\Adv$ on behalf of every corrupted party $\Party'_j\in\committee$ and verifies according to the relation $R_j$ (incorrect proofs are ignored). Finally, $\Sim$ outputs whatever $\Adv$ outputs and halts.

Computational indistinguishability between the real execution of the compiled protocol $\pi$ running with adversary \Adv and the ideal computation of $f$ running with \Sim follows directly from the security of $\pi'$.
\end{proof}

The proof of \cref{thm:idfair_to_ridfair_noinput} can be easily adjusted to the case where $\pi'$ is a fully secure protocol for computing $f^{n'}$. In this case, since the augmented coin-tossing functionality $\faugct^{\committee}$ is secure with $\committee$-identifiable-abort (see \cref{sec:ideal_functionalities}),
the adversary can force to restart it $t'+1$ times. Once $\faugct^{\committee}$ completes, the adversary cannot abort the computation. This yields the following corollary.
\begin{corollary}\label{cor:full_to_full_noinput}
Assume that \TDP and \CRH exist.
Then, there exists a \ppt algorithm $\CompilerNHMNI{n'}{n}$ such that for any $n'$-party, \rchanged{$r$}{$r'$}-round protocol $\pi'$ computing $f^{n'}$, the protocol $\pi = \CompilerNHMNI{n'}{n}(\pi',\committee)$ is an $n$-party, \rchanged{$O(t'+r)$}{$O(t'+r')$}-round protocol computing $f^n$ with the following guarantee.
If the number of corrupted parties in $\committee$ is at most $t'$, and $\pi'$ is a protocol that \rchanged{$(\delta,t')$}{$(\delta',t')$}-securely computes $f^{n'}$ with full security, then $\pi$ is a protocol that \rchanged{$(\delta,t)$}{$(\delta',t)$}-securely computes $f^n$ with full security.
\end{corollary}

\subsection{Fairness to Full Security with an Honest Majority (No Inputs)}\label{sec:fairtofull_noinput_HM}
We now turn to the honest-majority setting, and present a reduction from full security to fairness for no-input functionalities.
In \cref{sec:ridfair_to_full_noinput_HM}, we show how to compute $f^n$ with full security in the hybrid model computing $f^n$ with fairness with restricted identifiable abort using $\Omega(1)$ calls.
In \cref{sec:fair_to_ridfair_noinput_HM}, we show how to compile a fair protocol for $f^{n'}$ to a fair protocol for $f^n$ with restricted identifiable abort, unconditionally, and in a black-box manner.

\subsubsection{Reducing the Round Complexity with an Honest Majority}\label{sec:ridfair_to_full_noinput_HM}

In the honest-majority setting, we are able to utilize parallel computations in many committees in order to reduce the number of calls to the ideal computation from $\Omega(\log(\secParam))$ to $\Omega(1)$.
The idea is to start as in the no-honest-majority setting (\cref{prot:ridfair_to_full_noinput}) by electing a (super-)logarithmic committee $\committee$, but instead of sequentially invoking the computation with $\committee$-identifiable-abort, we consider multiple sub-committees $\vCS=(\committee_1,\ldots,\committee_\ell)$, where $\committee_j\subseteq \committee$, and invoke a computation with $\vCS$-identifiable-abort. By defining the subsets in $\vCS$ appropriately, we can ensure that every sub-committee has an honest majority, and that the adversary must reveal the identity of many corrupted parties in order to abort the computation.

\begin{theorem}\label{thm:ridfair_to_full_parallel_noinput}
Let $f$ be a no-input functionality with public output, let $0<\beta<\beta'<1/2$, let $n\in\N$, let $n'=\log(\secParam)\cdot(\varphi(\secParam)-1/\varphi(\secParam))$ with $\varphi=1/\sqrt{1-2\beta'}+\Omega(1)$, let $\ell=\secParam^\uglyExp$, and let $t=\beta n$ and $t'=\beta' n'$.
Then, $f^n$ can be $(\err(n,n',\beta,\beta'),t)$-securely computed with full security in a hybrid model that computes $f^n$ with fairness and $(\ell, n',t')$-identifiable-abort, by invoking the ideal functionality in $\varphi(\secParam)^2$ rounds.
\end{theorem}

\begin{proof}
Assume for simplicity of exposition that $n=\Omega(\log(\secParam)\cdot \varphi(\secParam))$.
We modify \cref{prot:ridfair_to_full_noinput} as follows. First, the parties elect a committee $\committee\subseteq[n]$ of size $m=\log(\secParam)\cdot \varphi(\secParam)$ by invoking $\felect(n,m,\beta')$. Consider all subsets of $\committee$ of size $n'=m-n''$ for $n''=\log(\secParam)/\varphi(\secParam)$, denoted $\vCS=(\committee_1,\ldots,\committee_{\ell'})$ with $\ell'=\binom{m}{n'}$. Next, the parties proceed in iterations, where in each iteration they invoke the ideal functionality computing $f$ with fairness and $\vCS$-identifiable-abort. In case the output value is a subset of identified corrupted parties, they are removed and another iteration is carried out. Otherwise, the parties halt with the output value.

We start by showing that the protocol indeed runs in polynomial time, \ie that the number of sub-committees is polynomial.

\begin{claim}\label{claim:par_poly_noinput}
$\binom{m}{n'}\leq n^{\log{e}\cdot\left(\frac{2}{e} + \frac{1}{\varphi(\secParam)}\right)}$.
\end{claim}
\begin{proof}
For every $1<k<m$ it holds that $\binom{m}{k} \leq (em/k)^k$, therefore,
\begin{align*}
\binom{m}{n'} = \binom{m}{n''} \ \leq \ \left(\frac{em}{n''}\right)^{n''}
 &  = \ \left(\frac{e\cdot\log(\secParam)\cdot \varphi(\secParam)}{\log(\secParam)/\varphi(\secParam)}\right)^{\log(\secParam)/\varphi(\secParam)} \\
~ & \ = \ \ \left(e\cdot\varphi(\secParam)^2\right)^{\log(\secParam)/\varphi(\secParam)} \\
~ & \ = \ \ 2^{\frac{\log{\left(e\cdot\varphi(\secParam)^2\right)}\cdot\log(\secParam)}{\varphi(\secParam)}} \\
~ & \ = \ \ \secParam^{\left(\frac{\log{e}}{\varphi(\secParam)}+\frac{2\log{\varphi(\secParam)}}{\varphi(\secParam)}\right)} \\
~ & \ \stackrel{(\ast)}{\leq} \ \ \secParam^{\left(\frac{\log{e}}{\varphi(\secParam)}+\frac{2\log{e}}{e}\right)},
\end{align*}
where $(\ast)$ follows by \citet{Steiner50} who showed that $x^{1/x}$ is bounded from above by $e^{1/e}$, for every positive $x$, hence
\[
\frac{\log{\varphi(\secParam)}}{\varphi(\secParam)} =
\log\left(\varphi(\secParam)^{1/\varphi(\secParam)}\right) \leq
\log\left(e^{1/e}\right) =
\frac{\log{e}}{e}.
\]
\end{proof}

Next, we show that every sub-committee still has an honest majority.
\begin{claim}\label{claim:par_honest_noinput}
In every sub-committee $\committee_i$ there exists an honest majority.
\end{claim}
\begin{proof}
We require that $n'>2\cdot\ssize{\committee\cap\IS}$.
By the definition of $\felect(n,m,\beta')$ it holds that $\ssize{\committee\cap\IS}\leq\beta'\cdot m$. The claim will therefore follow if $n'>2\cdot\beta'\cdot m$, \ie
\begin{equation}\label{eq:hm}
\log(\secParam)\cdot(\varphi(\secParam)-1/\varphi(\secParam)) >2\cdot \beta'\cdot\log(\secParam)\cdot \varphi(\secParam).
\end{equation}

\noindent
Since $0<\beta'<1/2$ and $\varphi=\Omega(1)$, \cref{eq:hm} holds for
\[
\varphi(\secParam)>\frac{1}{\sqrt{1-2\beta'}}.
\]
\end{proof}

Finally, we show that after calling the ideal functionality in $\varphi(\secParam)^2$ rounds, the adversary cannot abort the computation.
\begin{claim}\label{claim:par_term_noinput}
After $\varphi(\secParam)^2$ iterations the protocol terminates.
\end{claim}
\begin{proof}
Denote by $\JS$ the set of identified corrupted parties in some iteration of the protocol. If $\ssize{\JS}<n''$ then $\ssize{\committee\setminus\JS}>m-n''$; hence, there exists a sub-committee $\committee_i\subseteq \committee\setminus\JS$ of size $m-n''$ such that the adversary didn't identify any corrupted parties from $\committee_i$.
It follows that in every iteration either the computation completes or at least $n''$ corrupted parties are identified. Therefore, after $m/n'' = \varphi(\secParam)^2$ iterations, the computation will complete.
\end{proof}

The simulator proceeds similarly to the simulator in \cref{thm:ridfair_to_full_noinput}, where in every iteration, if the adversary sends $(\abort,\ID)$, with $\ID\subseteq\IS$ and $\ID\cap\committee_l\neq\emptyset$ for every $l\in[\ell']$, the simulator $\Sim$ simulates sending $(\bot,\ID)$ to the parties. If $\Adv$ does not send \abort, $\Sim$ calls the trusted party in the fully secure ideal model computing $f$, obtains the output $y$, forwards $y$ to the adversary, and outputs whatever $\Adv$ outputs.
\end{proof}

\subsubsection{Fairness with Restricted Identifiable Abort with an Honest Majority}\label{sec:fair_to_ridfair_noinput_HM}
Now, we show that in the honest-majority setting, the reduction from fairness with restricted identifiable abort to fairness can be much more elegant, and more importantly, be based on much simpler tools. Specifically, we devise a compiler, similar to the one for the no-honest-majority case, that is based solely on error-correcting secret sharing schemes (ECSS), which exist unconditionally (see \cref{def:ECSS}).

Given a no-input, $n$-party functionality with public output, we consider the no-input, $n'$-party functionality, denoted $\foutss{f}{t'}{n'}$ that computes $f^n$ and outputs shares of the result using $(t'+1)$-out-of-$n'$ ECSS, for some $t'<n/2$. Note that a secure computation of $\foutss{f}{t'}{n'}$ with abort is in fact fair assuming $t'$ corruptions, since an adversary aborting the computation does not learn any new information.

We prove the theorem below, using the following notations. Let $f$ be a no-input functionality with public output, let $t<n/2$, let $n'<n$, let $t'<n'/2$, and let $\vCS=(\committee_1,\ldots,\committee_\ell)$, where $\committee_l\subseteq [n]$ for every $l\in[\ell]$.
\begin{theorem}\label{thm:idfair_to_ridfair_noinput_HM}
There exists a \ppt algorithm $\CompilerHMNI{n'}{n}$ such that if the number of corrupted parties in every $\committee_j$ is at most $t'$, then the following holds with information-theoretic security.
For any $n'$-party, \rchanged{$r$}{$r'$}-round protocol that \rchanged{$(\delta,t')$}{$(\delta',t')$}-securely computes $\foutss{f}{t'}{n'}$ with abort, $\ell$ times in parallel, the protocol $\pi = \CompilerHMNI{n'}{n}(\pi',\vCS)$ is an $n$-party, \rchanged{$O(r)$}{$O(r')$}-round protocol that \rchanged{$(\ell\cdot\delta,t)$}{$(\ell\cdot\delta',t)$}-securely computes $f$ with fairness and $\vCS$-identifiable-abort.
Furthermore, the compiler is black-box \wrt the protocol $\pi'$.
\end{theorem}

\begin{proof}
In \cref{lem:idfair_to_ridfair_noinput_HM} we prove the theorem under the assumption that $\pi'$ is secure with identifiable abort. In \cref{lem:abort_to_idabort_noinput_HM} we explain how to compile, in a round-preserving manner, every protocol that is secure with abort into a protocol that is secure with identifiable abort in the honest-majority setting, with information-theoretic security and using only a black-box access to the underlying protocol.
\end{proof}

\begin{lemma}\label{lem:idfair_to_ridfair_noinput_HM}
Consider the same notations as in \cref{thm:idfair_to_ridfair_noinput_HM}.
If $\pi'$ is protocol that \rchanged{$(\delta,t')$}{$(\delta',t')$}-securely computes $\foutss{f}{t'}{n'}$ with identifiable abort, $\ell$ times in parallel, then the protocol $\pi = \CompilerHMNI{n'}{n}(\pi',\vCS)$ is an $n$-party protocol that \rchanged{$(\ell\cdot\delta,t)$}{$(\ell\cdot\delta',t)$}-securely computes $f^n$ with fairness and $\vCS$-identifiable-abort.
\end{lemma}

\begin{proof}
Given the protocol $\pi'$ and $\vCS=(\committee_1,\ldots,\committee_\ell)$, the $n$-party protocol $\pi = \CompilerHMNI{n'}{n}(\pi',\vCS)$ is defined as follows.

\begin{construction}(security with abort to fairness with restricted identifiable abort)\label{prot:abort_to_ridfair_noinput_hmn}
\begin{itemize}
    \item\textbf{Common Input:}
    An $n'$-party protocol $\pi'$ and $\vCS=(\committee_1,\ldots,\committee_\ell)$, where each $\committee_l\subseteq[n]$ \radded{is} of size $n'$. We use the notation $\Party^l_j$ to refer to the \jth party in $\committee_l$.
    \item\textbf{The Protocol:}
\end{itemize}
\begin{enumerate}
    \item\label{step:idabort_to_ridfair_protocol}
    For every committee $\committee_l$, the parties in $\committee_l$ execute the protocol $\pi'$ for computing $\foutss{f}{t'}{n'}$ with identifiable abort.
    Let $\outvalue_{j,l}$ be the output $\Party_j^l\in\committee_l$ received (either a valid share $y_j^l$ or $(\bot,\is)$ with $\Party_{\is}=\Party_{\js}^l\in\committee_l$).
    If $\outvalue_{j,l}$ is of the form $(\bot,\is)$, then $\Party_j^l$ broadcasts $(l,\outvalue_{j,l})$.
    \item\label{step:idabort_to_ridfair_announce}
    If for every committee $\committee_l$, more than $t'$ parties broadcasted a value, then every party $\Party_i$ takes the value $(l,\is_l)$ that appears the most for every $\committee_l$ (if there is no unique majority value, choose arbitrarily) and outputs $(\bot,\sset{\is_1,\ldots,\is_\ell})$. Otherwise, denote by $\ls$ the minimal index $l$ such that at most $t'$ values $(l,\is)$ were broadcasted.
    \item\label{step:idabort_to_ridfair_shares}
    Every party \rdeleted{in} $\Party_j^{\ls}\in\committee_{\ls}$ broadcasts its output value $\outvalue_{j,\ls}$.
    \item\label{step:idabort_to_ridfair_reconstruct}
    Denote by $y_j^{\ls}$ the value broadcasted by $\Party_j^{\ls}$. Every party computes $y=\Recon(y_1^{\ls},\ldots,y_{n'}^{\ls})$ and outputs $y$.
\end{enumerate}
\end{construction}

Let $\Adv$ be a computationally unbounded adversary attacking the execution of $\pi$, and let $\IS\subseteq[n]$ satisfying $\ssize{\IS\cap\committee_l}\leq t'$ for every $l\in[\ell]$.
We construct an adversary $\Sim$ (simulator) for the ideal model computing $f$ with fairness and $\vCS$-identifiable-abort.
The simulator $\Sim$ starts by emulating $\Adv$ on the auxiliary input $\aux$. The simulator $\Sim$ interacts with $\Adv$, playing the roles of the honest parties. For simplicity, assume that the output value of $f$ are elements in $\zo^\secParam$.

To simulate Step~\ref{step:idabort_to_ridfair_protocol}, the simulator uses the simulator $\tilde{\Sim}$ that is guaranteed to exist for the $\ell$-times parallel execution of $\pi'$ with the residual adversary of $\Adv$ in Step~\ref{step:idabort_to_ridfair_protocol}. $\Sim$ invokes $\tilde{\Sim}$ on the auxiliary $\aux$ and receives from $\tilde{\Sim}$ early messages $(\abort,\is_l)$ for the \lth computation. If $\tilde{\Sim}$ did not abort the computation of a committee $\committee_l$, the simulator $\Sim$ hands secret shares of $0^\secParam$ to the corrupted parties in $\committee_l\cap\IS$. Next, $\tilde{\Sim}$ may send late-abort messages $(\abort,\is_l)$. Finally, the simulator \Sim receives the output from $\tilde{\Sim}$ that contains the simulated view of the adversary, and interacts with \Adv accordingly.

To simulate Step~\ref{step:idabort_to_ridfair_announce}, for every $\committee_l$ for which $\tilde{\Sim}$ aborted the computation, \Sim simulates broadcasting $(l,\is_l)$ by all honest parties in $\committee_l$.
In case $\tilde{\Sim}$ aborted the computations for every committee $\committee_l$, the simulator \Sim sends $(\abort,\sset{\is_1,\ldots,\is_\ell})$ to the trusted party (in the ideal model with fairness and $\vCS$-identifiable-abort). Otherwise, \Sim receives the output $y$ from the trusted party, secret shares $y$ and simulates the honest parties in $\committee_{\ls}$ broadcasting their shares, for simulating Step~\ref{step:idabort_to_ridfair_shares}. Finally, $\Sim$ outputs whatever \Adv outputs and halts.

The proof follows in a straightforward manner based on the unconditional security of the ECSS scheme and of the protocol $\pi'$.
\end{proof}

\begin{lemma}\label{lem:abort_to_idabort_noinput_HM}
Let $f$ be an $n$-party functionality and let $t<n/2$. The following holds with information-theoretic security.
There exists a \ppt compiler $C$ such that if $\pi$ is an $r$-round protocol that $(\delta,t)$-securely computes $f$ with abort, then $C(\pi)$ is an $O(r)$-round protocol that $(\delta,t)$-securely computes $f$ with identifiable abort. Moreover, the compiler uses the protocol $\pi$ in a black-box manner.
\end{lemma}
\begin{proof}
In the proof we combine existing results from~\cite{IOZ14,CCGZ17}.
We start by using the compiler of \citet[Thm.\ 6]{IOZ14}, from security with abort to security with identifiable abort, that satisfies the requirements in the lemma, in the Setup-Commit-then-Prove hybrid model.\footnote{The Setup-Commit-then-Prove is a reactive functionality, parametrized by a vector of NP-relations, that in the first invocation hands every party a secret witness, and in all future invocation allows each party to prove NP-statements to all other parties, using the secret witness. This functionality extends the Commit-then-Prove functionality from~\cite{CLOS02}.} In \cite{IOZ14}, the Setup-Commit-then-Prove functionality was realized in a hybrid model that gave correlated randomness to the parties.
\radded{
Namely, they proved the following claim.
\begin{claim}
Let $\pi$ be an $r$-round protocol which $(\delta,t)$-securely computes $f$ with abort and with information-theoretic security in the correlated-randomness model for a distribution $D$ (\ie where a trusted dealer samples correlated randomness for the parties in the setup phase). There exists a \ppt compiler $C$ such that $C(\pi)$ is an $O(r)$-round protocol that $(\delta,t)$-securely computes $f$ with identifiable abort in the Setup-Commit-then-Prove hybrid model and the correlated randomness model. Moreover, the compiler uses the protocol $\pi$ in a black-box manner.
\end{claim}
}
\citet[Lem.\ 6.2]{CCGZ17} showed how to security compute (without abort) the required correlated randomness in the honest-majority setting by a constant-round protocol.
\radded{
Namely, they proved the following claim.
\begin{claim}
Let $\pi$ be a constant-round protocol which $(\delta,t)$-securely computes $f$ with identifiable abort and with information-theoretic security in the Setup-Commit-then-Prove hybrid model and the correlated-randomness model for an efficiently sampleable distribution $D$ in $NC^0$. Then $f$ can be $(\delta,t)$-securely computed with identifiable abort and with information-theoretic in constant rounds in the broadcast model with secure point-to-point channels.
\end{claim}
}
The lemma follows from these results.
\end{proof}

\subsection{Applications}\label{sec:apps_noinput}

We next give a few applications of the above security uplifting reductions \wrt coin-flipping protocols.
\begin{definition}\label{def:coinflip}
The $n$-party coin-flipping functionality is defined as $\fcf{n}(\lambda,\ldots,\lambda)=(b,\ldots,b)$, where $b\in\zo$ is a uniformly distributed bit. A $\delta$-bias, $t$-secure coin-flipping protocol is a protocol that $(\delta,t)$-securely computes $\fcf{n}$ with full security.
\end{definition}

For any $n\in\N$, $t<n$, and $r=r(\secParam)$, \citet{ABCGM85} presented an $r$-round, $O(t/\sqrt{r})$-bias coin-flipping protocol tolerating up to $t$ corrupted parties (assuming OWF). \citet{BeimelOO15} improved this result, giving an $r$-round, $O(1/\sqrt{r-t})$-bias coin-flipping protocol tolerating up to $t = \beta\cdot n$ corrupted parties, for $0<\beta<1$ (assuming OT).
We now show how to reduce the dependency on $t$.

\begin{corollary}\label{cor:main_cf-coro}
Assume that \TDP and \CRH exist.
Let $n'<n$ be integers, and let $0<\beta<\beta'<1$ be constants.
If there exists an $n'$-party, \rchanged{$\delta$}{$\delta'$}-bias, \rchanged{$r$}{$r'$}-round coin-flipping protocol $\pi'$ tolerating $t'=\beta'n'$ corrupted parties, then there exists an $n$-party, \rchanged{$(\delta+\err(n,n',\beta,\beta'))$}{$(\delta'+\err(n,n',\beta,\beta'))$}-bias, \rchanged{$O(t'+r)$}{$O(t'+r')$}-round coin-flipping protocol, tolerating $t=\beta n$ corrupted parties.
\end{corollary}

\begin{proof}
Note that $\pi'$ is an $n'$-party, \rchanged{$r$}{$r'$}-round protocol that \rchanged{$(\delta,t')$}{$(\delta',t')$}-securely computes $\fcf{n'}$ with full security. Hence, the proof follows from \cref{thm:mainThm_fairtofull_noinput_step2} of \cref{thm:mainThm_fairtofull_noinput}.
\end{proof}

\begin{corollary}[restating \cref{cor:intro_CF}]
Assume that \TDP and \CRH exist. Let $n\in \N$, let $t=\beta n$ for $0<\beta<1$, and let $r:\N\mapsto\N$ be an efficiently computable function.
\begin{enumerate}
    \item\label{cf:itemone}
    There exists an $n$-party, $r(\secParam)$-round,
    $O\Big(\frac{1}{\sqrt{r(\secParam)-\log(\secParam)}} + \frac{1}{\secParam\cdot\log(\secParam)}\Big)$-bias,
    $t$-secure coin-flipping protocol.
    \item\label{cf:itemtwo}
    Let $\varphi=\omega(1)$.
    There exists an $n$-party, $r(\secParam)$-round,
    $O\Big(\frac{1}{\sqrt{r(\secParam)-\varphi(\secParam)\cdot\log(\secParam)}}\Big)$-bias,
    $t$-secure coin-flipping protocol.
\end{enumerate}
\end{corollary}

\begin{proof}
For $n$ and $\beta$ as above, fix $\beta'=(1+\beta)/2$.
For Item~\ref{cf:itemone} consider $n'=\log(\secParam)$ and for Item~\ref{cf:itemtwo} consider $n'=\varphi(\secParam)\cdot\log(\secParam)$.\footnote{If $n<\varphi(\secParam)\cdot\log(\secParam)$, the parties can simply consider $n'=n$ and $t'=t$.}
Set $t'=\beta' n'$ and let $\pi'$ be the protocol of~\cite{BeimelOO15} that is an $n'$-party, \rchanged{$r$}{$r'$}-round, \rchanged{$O(1/\sqrt{r-t'})$}{$O(1/\sqrt{r'-t'})$}-round coin-flipping protocol, tolerating up to $t'$ corrupted parties. The proof follows by~\cref{cor:main_cf-coro}.
\end{proof}

\section{Fairness to Full Security for With-Input Functionalities}\label{sec:fair_to_full_withinputs}
In this section, we present a reduction from a fully secure computation to a fair computation for functionalities with inputs. The main additional challenge compared to no-input functionalities (\cref{sec:fairtofull_noinput}) is enforcing the small committees to carry out the computation on the inputs of the honest parties, while preserving privacy.
In \cref{sec:fairtofull_withinput_noHM}, we show how to reduce full security to fairness with restricted identifiable abort in a round-efficient manner.
In \cref{sec:fairtofull_withinput_HM}, we present an analogue result in the honest-majority setting that is unconditional and black-box in the underlying fair protocol.
Applications are found in \cref{sec:apps_withinput}.

Since we consider delegating a computation of functionalities \emph{with inputs} to a small committee, we use a protocol for computing the function over secret-shared inputs. In the honest-majority setting, the parties use a $(t'+1)$-out-of-$n'$ ECSS scheme (\cref{def:ECSS}) to distribute their inputs among the committee members. We denote by $\finout{f}{n}{t'}{n'}$ the Reconstruct-Compute-Share variant of $f$, that receives secret shares of the inputs, computes $f$ over the reconstructed $n$-tuple, and outputs secret shares of the result.
In the no-honest-majority setting, every party initially commits to its input (in a somewhat non-trivial way, in order to identify cheating committee members). We denote by $\fin{f}{n}{n'}$ the Verify-Reconstruct-Compute variant of $f$, that receives $n'$-out-of-$n'$ secret shares for the decommitment of each party, verifies that all the commitments can be opened, reconstructs the $n$-tuple, computes $f$, and outputs the result. Formal definitions of these functionalities can be found in \cref{sec:shared_inputs}.
As before, for integers $n'<n$ and for $0<\beta<\beta'<1$ we define $\err(n,n',\beta,\beta')=\frac{n}{n'} \cdot e^{-\frac{(\beta'-\beta)^2 n'}{2(1-\beta)}}$.

We prove the following theorem.

\begin{theorem}[restating \cref{thm:intro:mainThmfair_to_full}]\label{thm:mainThmfair_to_full}
Assume that \TDP, \CRH, and non-interactive perfectly binding commitment schemes exist.
Let $f$ be an $n$-party functionality with public output, let $0<\beta<\beta'<1$, let $n'=\min(n,\log(\secParam)\cdot\varphi(\secParam))$ with $\varphi=1/\sqrt{1-\beta'}+\Omega(1)$, and let $t=\beta n$ and $t'=\beta' n'$.
\begin{enumerate}
    \item
    If $\fin{f}{n}{n'}$ can be \rchanged{$(\delta,t')$}{$(\delta',t')$}-securely computed with fairness by an \rchanged{$r$}{$r'$}-round protocol, then $f$ can be \rchanged{$(t'\cdot\delta + \err(n,n',\beta,\beta'),t)$}{$(t'\cdot\delta' + \err(n,n',\beta,\beta'),t)$}-securely computed with full security by an \rchanged{$O(t'\cdot r)$}{$O(t'\cdot r')$} protocol.
    \item
    For a deterministic functionality $f$, if $\fin{f}{n}{n'}$ can be \rchanged{$(\delta,n'-1)$}{$(\delta',n'-1)$}-securely computed with fairness by an \rchanged{$r$}{$r'$}-round protocol, $\ell$-times in parallel, for $\ell=\secParam^\uglyExp$, then $f$ can be \rchanged{$(\varphi(\secParam)^2\cdot \ell\cdot\delta + \err(n,n',\beta,\beta'),t)$}{$(\varphi(\secParam)^2\cdot \ell\cdot\delta' + \err(n,n',\beta,\beta'),t)$}-securely computed with full security by an \rchanged{$O(\varphi(\secParam)^2\cdot r)$}{$O(\varphi(\secParam)^2\cdot r')$} protocol.
    \item
    For $\beta'<1/2$, the following holds unconditionally.
    If $\finout{f}{n}{t'}{n'}$ can be \rchanged{$(\delta,t')$}{$(\delta',t')$}-securely computed with abort by an \rchanged{$r$}{$r'$}-round protocol, $\ell$-times in parallel, for $\ell=\secParam^\uglyExp$, then $f$ can be \rchanged{$(\varphi(\secParam)^2\cdot \ell\cdot\delta + \err(n,n',\beta,\beta'),t)$}{$(\varphi(\secParam)^2\cdot \ell\cdot\delta' + \err(n,n',\beta,\beta'),t)$}-securely computed with full security by an \rchanged{$r$}{$r'$}-round protocol.
\end{enumerate}
\end{theorem}

The proof of \cref{thm:mainThmfair_to_full} is given in the sections below, where the first part follows from a combination of \cref{thm:ridfair_to_full_withinput} and \cref{thm:idfair_to_ridfair}; the second part from \cref{thm:ridfair_to_full_parallel_withinput} and \cref{thm:idfair_to_ridfair}; and the third part from \cref{thm:ridfair_to_full_parallel_withinput} and \cref{thm:idfair_to_ridfair_hm}.

\subsection{Fairness to Full Security without an Honest Majority (With Inputs)}\label{sec:fairtofull_withinput_noHM}
We start by constructing a reduction from full security to fairness for functionalities with inputs, when an honest majority is \emph{not} assumed.
In \cref{sec:ridfair_to_full_withinput_noHM}, we show how to compute $f$ with full security in the hybrid model computing $f$ with fairness and restricted identifiable abort using $\omega(1)$ calls.
In \cref{sec:ridfair_to_full_parallel}, we show how to compile a fair protocol for $\fin{f}{n}{n'}$ to a fair protocol for $f$ with restricted identifiable abort.

\subsubsection{Fairness with Restricted Identifiable Abort to Full Security}\label{sec:ridfair_to_full_withinput_noHM}

We now show how to reduce full security to fairness with restricted identifiable abort.
Similarly to \cref{sec:ridfair_to_full_noinput_noHM}, we start with the simpler case where a single committee $\committee$ is considered (\ie $\ell=1$). For deterministic functionalities, we can use the technique from \cref{sec:ridfair_to_full_noinput_HM} and reduce the round complexity by using multiple committees.

\begin{theorem}\label{thm:ridfair_to_full_withinput}
Let $f$ be an $n$-party, public-output functionality, let $n'<n$, let $0<\beta<\beta'<1$, and let $t=\beta n$ and $t'=\beta'n'$.
Then, $f$ can be $(\err(n,n',\beta,\beta'),t)$-securely computed with full security in a hybrid model that computes $f$ with fairness and $(n',t')$-identifiable-abort, by using $t'+1$ sequential calls to the ideal functionality.
\end{theorem}

\begin{proof}
The protocol is very similar to \cref{prot:ridfair_to_full_noinput}, we present it here for completeness.
\begin{protocol}(fairness with restricted identifiable abort to full security)\label{prot:ridfair_to_full_withinput_noHM}
\begin{itemize}
    \item\textbf{Hybrid Model:}
    The protocol is defined in the hybrid model computing $\felect$ with full security, and $f$ with fairness with $(n',t')$-identifiable-abort.
    \item\textbf{Common Input:}
    The values $t',n'\in\N$.
    \item\textbf{Private Input:}
    Every party $\Party_i$ has private input $x_i\in\zs$, for $i\in[n]$.
    \item\textbf{The Protocol:}
\end{itemize}
\begin{enumerate}
    \item \label{step:rid_elect}
    All the parties invoke $\felect(n,n',\beta')$ and elect a committee $\committee_1\subset[n]$ of size $n'$.
    \item
    For $i=1,\ldots,t'+1$ do
    \begin{enumerate}
        \item
        All parties that have not been previously identified send their inputs to the trusted party computing $(f,\committee_1\mhyphen\idfair)$, where the party with the lowest index in $\committee_i$ simulates all parties in $\committee_1\setminus \committee_i$, using their predetermined default input values.
        Denote the output $\Party_j$ receives by $y_j$.
        \item
        Every party $\Party_j$ checks if $y_j$ is a valid output, if so $\Party_j$ outputs $y_j$ and halts.
        Otherwise, all parties received $(\bot,\is)$ as output, where $\is\in\committee_1\cap\IS$.
        If $\is\notin\committee_i$ (and so $\Party_{\is}$ is a previously identified corrupted party), then all parties set $\is$ to be the party with the lowest index in $\committee_i$.
        \item
        All parties set $\committee_{i+1}=\committee_i\setminus\set{\is}$.
    \end{enumerate}
\end{enumerate}
\end{protocol}

Proving the security of \cref{prot:ridfair_to_full_withinput_noHM} follows almost identically as the proof of \cref{thm:ridfair_to_full_noinput}. The difference is that when simulating the ideal functionality computing $f$ with $\committee$-identifiable-abort in a non-aborting iteration, the simulator gets input values from the adversary, which it sends to the trusted party in the fully secure ideal computation in order to get the output.
\end{proof}

\subsubsection{Reducing the Round Complexity}\label{sec:ridfair_to_full_parallel}

We next show that the technique, used in \cref{sec:ridfair_to_full_noinput_HM} for the honest-majority setting, to reduce the number of rounds for invoking the ideal computation from $\omega(\log(\secParam))$ to $\omega(1)$, can be applied also to the no-honest-majority case. However, this improvement will turn out to be meaningful only for deterministic functionalities (with public output). It is well known that every functionality can be adjusted to be deterministic and with public output via standard techniques; however, the resulting functionality has inputs (even if the original functionality had no inputs), therefore, it is only relevant for the with-input setting.

\begin{theorem}\label{thm:ridfair_to_full_parallel_withinput}
Let $f$ be an $n$-party functionality, let $0<\beta<\beta'<1$, let $n'=\log(\secParam)\cdot(\varphi(\secParam)-1/\varphi(\secParam))$ with $\varphi(\secParam)=1/\sqrt{1-\beta'}+\Omega(1)$, let $t'=n'-1$, let $\ell=\secParam^\uglyExp$, and let $t=\beta n$.
Then, $f$ can be $(\err(n,n',\beta,\beta'),t)$-securely computed with full security in a hybrid model that computes $f$ with fairness and $(\ell, n',t')$-identifiable-abort, by invoking the ideal functionality in $\varphi(\secParam)^2$ rounds.
\end{theorem}

\begin{proofsketch}
Assume for simplicity of exposition that $n=\Omega(\log(\secParam)\cdot \varphi(\secParam))$.
Similarly to the proof of \cref{thm:ridfair_to_full_parallel_noinput}, the parties initially invoke $\felect(n,m,\beta')$ to elect a committee $\committee\subseteq[n]$ of size $m=\min(n,\log(\secParam)\cdot \varphi(\secParam))$, and consider all subsets of $\committee$ of size $n'=m-n''$ for $n''=\log(\secParam)/\varphi(\secParam)$, denoted $\vCS=(\committee_1,\ldots,\committee_{\ell'})$. Following \cref{claim:par_poly_noinput}, there are polynomially many sub-committees, more precisely, $\ell'\leq\ell$. In addition, following \cref{claim:par_term_noinput} the protocol will complete within $\varphi(\secParam)^2$. It is left to show that in every sub-committee there is an honest party. This is done by slightly adjusting \cref{claim:par_honest_noinput}.

\begin{claim}\label{claim:par_honest}
In every sub-committee $\committee_i$ there exists at least one honest party.
\end{claim}
\begin{proof}
We require that $n'>\ssize{\committee\cap\IS}$.
By the definition of $\felect(n,m,\beta')$ it holds that $\ssize{\committee\cap\IS}\leq\beta'\cdot m$. The claim will therefore follow if $n'>\beta'\cdot m$, \ie
\begin{equation}\label{eq:nohm}
\log(\secParam)\cdot(\varphi(\secParam)-1/\varphi(\secParam)) >\beta'\cdot\log(\secParam)\cdot \varphi(\secParam).
\end{equation}

\noindent
Since $0<\beta'<1$ and $\varphi=\Omega(1)$, \cref{eq:nohm} holds for
\[
\varphi(\secParam)>\frac{1}{\sqrt{1-\beta'}}.
\]
\end{proof}

The simulator proceeds similarly to the simulator in \cref{thm:ridfair_to_full_noinput}, where in every iteration, if the adversary sends $(\abort,\ID)$, with $\ID\subseteq\IS$ and $\ID\cap\committee_l\neq\emptyset$ for every $l\in[\ell']$, the simulator $\Sim$ simulates sending $(\bot,\ID)$ to the parties. If $\Adv$ sends input values $\sset{x'_i}_{i\in\IS}$ to the computation, $\Sim$ forwards these input values to the trusted party and obtains the output $y$. Next, $\Sim$ forwards $y$ to the adversary and outputs whatever $\Adv$ outputs.
\end{proofsketch}

\subsubsection{Fairness to Fairness with Restricted Identifiable Abort}\label{sec:fair_to_ridfair}
In this section, we present a reduction from fair computation with identifiable abort to fair computation for functionalities with inputs. More specifically, let $\pi'$ be a protocol for computing $\fin{f}{n}{n'}$ with fairness, even when run $\ell$ times in parallel by subsets of parties $\vCS= (\committee_1,\ldots,\committee_\ell)$. We show that $\pi'$ can be compiled into a protocol $\pi$ that computes $f$ with fairness and $\vCS$-identifiable-abort. For $\ell>1$ we require that $f$ is deterministic.

The above is formally stated in the theorem below, using the following notations. Let $f$ be a deterministic $n$-party functionality with public output, let $t,n'<n$, let $t'<n'$ and $\ell\in\poly(n)$, and let $\vCS=(\committee_1,\ldots,\committee_\ell)$, where $\committee_1,\ldots,\committee_\ell\subseteq[n]$ of size $n'$.

\begin{theorem}\label{thm:idfair_to_ridfair}
Assume that \TDP, \CRH, and non-interactive perfectly binding commitment schemes exist.
Then, there exists a \ppt algorithm $\CompilerNHM{n'}{n}$ such that for any $n'$-party, \rchanged{$r$}{$r'$}-round protocol $\pi'$ computing $\fin{f}{n}{n'}$, the protocol $\pi = \CompilerNHM{n'}{n}(\pi',\vCS)$ is an $n$-party, \rchanged{$O(r)$}{$O(r')$}-round protocol computing $f$ with the following guarantee.
If the number of corrupted parties in every $\committee_j$ is at most $t'$, and $\pi'$ is a protocol that \rchanged{$(\delta,t')$}{$(\delta',t')$}-securely computes $\fin{f}{n}{n'}$ with fairness, $\ell$ times in parallel, then $\pi$ is a protocol that \rchanged{$(\ell\cdot\delta,t)$}{$(\ell\cdot\delta',t)$}-securely computes $f$ with fairness and $\vCS$-identifiable-abort.
\end{theorem}

\begin{proof}
In \cref{lem:idfair_to_ridfair}, we construct the protocol compiler $\CompilerNHM{n}{n'}(\pi',\vCS)$ in the $(\faugct,\zkmany)$-hybrid model (explained below), and in \cref{lem:instantiating_zkmanmy_noinput}, we show how to instantiate the ideal functionalities $\faugct$ and $\zkmany$ in a round-preserving manner.
The proof follows from the sequential composition theorem (\cref{prop:Composition}).
\end{proof}

\paragraph{A high-level description of the compiler.}
We now describe the compiler at a high level. For simplicity, we consider a single committee $\committee$ that computes $\fin{f}{n}{n'}$ with fairness and identifiable abort. The compiler consists of three phases: initially, every party shares its input among the committee members; next, the committee computes $f$ over the shared inputs; and finally, the committee members distribute the output to all the parties.

In the first phase, it is important that the input values of the parties are shared in a way that forces each committee member to use the actual value it received from each party (otherwise, corrupted parties might learn the value of $f$ on inputs that are correlated to the actual input values of the honest parties). In case a corrupted committee member uses a different value, it should be identified \emph{without} learning the output. Note that the guarantee provided by a fair computation is that if the adversary learns the output then so do all honest parties, but there is no restriction on the adversary from entering arbitrary values to the computation. One way to solve this issue is by having the functionality $\fin{f}{n}{n'}$ itself identify committee members that modify their shares to all the parties \emph{before} computing the function $f$.
However, now it should be ensured that a corrupted party cannot falsely incriminate an honest committee member, and claim that the committee member used different values than the one it gave him.

We solve both issues by having each party first publicly commit the values it sends to the committee members, in the following way. Every party secret shares its input value and publicly commits to each share. Next, the party encrypts each share using a committing public-key encryption scheme\footnote{A committing public-key encryption scheme is an encryption scheme with the property that it is computationally infeasible to find two pairs of (different) plaintext and randomness that are encrypted to the same ciphertext.} (every share is encrypted with the public key of the recipient committee member), and broadcasts all ciphertexts. Finally, each party proves to all other parties that it behaved honestly.

In the second phase, the committee members compute $\fin{f}{n}{n'}$ over the broadcast channel. This computation is compiled to enforce a semi-honest behavior in a similar way to the GMW compiler~\cite{GoldreichMW87}. The committee members run an augmented coin-tossing protocol, where each committee member receives a random string, and all other parties receive a commitment to this string. Next, the committee members execute the fair protocol computing $\fin{f}{n}{n'}$ with identifiable abort until it completes with an output value $y$ or with an identity of a corrupted party in the committee.

In the third phase, each committee member broadcasts its output value and proves to all parties that it behaved honestly, \ie that it followed the protocol using as input the shares it decrypted from the ciphertexts in the first phase, and the committed randomness generated in the second phase. If a party receives an accepting proof, it outputs the corresponding output value. Assuming that there exists an honest party in the committee, at least one of the proofs will be accepting.

If during the first two phases a party in the committee is identified as corrupted, then every party outputs its identity and the protocol halts. If a party outside of the committee is identified, \radded{then} all the parties remove this party from the computation and resume the execution.

\paragraph{Ideal functionalities used in the compiler.}
We construct the compiler in the hybrid model computing $\faugct^{\committee}$ and $\zkmany$ (see \cref{def:augFunc,def:zkmany}).
In the compiler below we consider two types of NP relations. The first is used when a party sends its shared input to the committees, and the second when a committee member sends its output value to all the parties.
These relation are formally described as follows:
\begin{itemize}
    \item
    In the relation $\Renc$, parametrized by a public-key encryption scheme $(\EncGen,\Enc,\Dec)$ and a non-interactive commitment scheme $\Com$, the public statement $(\vek,\vencval,\vcomval)$ consists of a vector of encryption keys $\vek=(\ek_{1,1},\ldots,\ek_{n',\ell})$, a vector of ciphertexts $\vencval=(\encval_{1,1},\ldots,\encval_{n',\ell})$ and a vector of commitments $\vcomval=(\comval_{1,1},\ldots,\comval_{n',\ell})$;
    the witness $(\vsvalhat,\vr)$ consists of vectors of strings $\vsvalhat=(\svalhat_{1,1},\ldots,\svalhat_{n',\ell})$, where $\svalhat_{j,l}=(\sval_{j,l},\decomval_{j,l})$, and randomness $\vr=(\rval_{1,1},\ldots,\rval_{n',\ell})$.
    For every $j\in[n']$ and $l\in[\ell]$, it holds that
    $\encval_{j,l}=\Enc_{\ek_{j,l}}(\svalhat_{j,l};\rval_{j,l})$,
    that $\comval_{j,l}=\Com(\sval_{j,l};\decomval_{j,l})$,
    and in addition, by denoting $x_l=\bigoplus_{j\in[n']}{\sval_{j,l}}$, that $x_1=\ldots=x_\ell$.
    \item
    In the following vector of NP-relations $(R_1,\ldots,R_{n'})$, parametrized by a public-key encryption scheme $(\EncGen,\Enc,\Dec)$ and a commitment scheme $\Com$, for every $j\in[n']$, the relation $R_j$ contains pairs $((\vencval,\vm,\comaugct,\outvalue),(\dk,\rndaugct,\decomaugct))$, where the public instance consists of a vector of ciphertexts $\vencval=(\encval_1,\ldots,\encval_n)$, a vector of messages $\vm=(m_1,\ldots,m_p)$, a commitment $\comaugct$, and a value $\outvalue$. The witness consists of a decryption key $\dk$, a random string $\rndaugct$, and a decommitment information $\decomaugct$. It holds that $\outvalue$ is the output value under the next-message function of $\Party_j$ in protocol $\pi'$ on input $(\svalhat_1,\ldots,\svalhat_n)=(\Dec_{\dk}(\encval_1),\ldots,\Dec_{\dk}(\encval_n))$, randomness $\rndaugct$ and messages $\vm$; in addition, $\comaugct=\Com(\rndaugct;\decomaugct)$.
\end{itemize}

\begin{lemma}\label{lem:idfair_to_ridfair}
Assume that committing public-key encryption schemes and non-interactive perfectly binding commitment schemes exist, and consider the same notations as in \cref{thm:idfair_to_ridfair}.
If $\pi'$ is a \rchanged{$(\delta,t')$}{$(\delta',t')$}-secure protocol computing $\fin{f}{n}{n'}$ with fairness, $\ell$ times in parallel,
then the protocol $\pi = \CompilerNHM{n'}{n}(\pi',\vCS)$ is an $n$-party protocol that \rchanged{$(\ell\cdot\delta,t)$}{$(\ell\cdot\delta',t)$}-securely computes $f$ with fairness and $\vCS$-identifiable-abort, in the $(\faugct,\zkmany)$-hybrid model.
\end{lemma}
\begin{proof}
The proof proceeds in a similar way to the proof of \cref{thm:idfair_to_ridfair_noinput}, where the main difference relates to sending the shared inputs to each committee. As in \cref{thm:idfair_to_ridfair_noinput}, we assume \wlg that $\pi'$ is fair with identifiable abort and that all messages are sent over the broadcast channel.
The $n$-party protocol $\pi = \CompilerNHM{n'}{n}(\pi',\vCS)$ is defined as follows.
\begin{construction}(fairness to fairness with restricted identifiable abort)\label{prot:fair_to_ridfair}
\begin{itemize}
    \item\textbf{Hybrid Model:}
    The protocol is defined in the hybrid model computing $\faugct$ with restricted identifiable abort and $\zkmany$ with full security.
    \item\textbf{Common Input:}
    A public-key encryption scheme $(\EncGen,\Enc,\Dec)$, a non-interactive commitment scheme $\com$, an $n'$-party protocol $\pi'$, and $\vCS=(\committee_1,\ldots,\committee_\ell)$. We use the notation $\Party^l_j$ to refer to the \jth party in $\committee_l$.
    \item\textbf{Private Input:}
    Every party $\Party_i$, for $i\in[n]$, has private input $x_i\in\zs$.
    \item\textbf{The Protocol:}
    Let $\ID$ denote the (initially empty) set of identified corrupted parties. At any point in the protocol, we call a committee $\committee_l$ \emph{active} if $\committee_l\cap\ID=\emptyset$.
\end{itemize}
\begin{enumerate}
    \item\label{step:rid_pk}
    Every $\Party_j\in\bigcup_{l\in[\ell]} \committee_l$ generates $(\dk_j,\ek_j)\gets \EncGen(1^\secParam)$ and broadcasts $\ek_j$.
    If some $\Party_{\is}\in\bigcup \committee_l$ did not broadcast, every party $\Party_i$ adds $\is$ to $\ID$.
    (The key-pair generated by $\Party_j^l\in\committee_l$ is denoted by $(\dk_{j,l},\ek_{j,l})$.) Denote $\vek=(\ek_{1,1},\ldots,\ek_{n',\ell})$.

    \item\label{step:rid_input}
    Every party $\Party_i$ proceeds as follows:
    \begin{enumerate}
        \item\label{step:rid_input_encrypt}
        For every active $\committee_l$, \radded{$\Party_i$} samples random values $(\sval_i^{1,l},\ldots,\sval_i^{n',l})$, conditioned on $x_i= \bigoplus_j{\sval_i^{j,l}}$. For every $j\in[n']$, \rchanged{commit}{$\Party_i$ commits} to the \jth share as $\comval_i^{j,l}=\Com(\sval_i^{j,l};\decomval_i^{j,l})$, denote $\svalhat_i^{j,l}=(\sval_i^{j,l},\decomval_i^{j,l})$, and encrypt $\encval_i^{j,l}=\Enc_{\ek_{j,l}}(\svalhat_i^{j,l};\rval_i^{j,l})$.
        \item\label{step:rid_input_prove}
        Send the values $((\vek,\vencval_i,\vcomval_i),(\vsvalhat_i,\vr_i))$ to $\zkmany$ (parametrized by $\Renc$),\footnote{More formally, we consider an ideal world for $n$-bounded parallel computation of $\zkmany$, where $\Party_i$ acts as the prover in the \ith computation.} where
        $\vencval_i=(\encval_i^{1,1},\ldots,\encval_i^{n',\ell})$,
        $\vcomval_i=(\comval_i^{1,1},\ldots,\comval_i^{n',\ell})$,
        $\vsvalhat_i=(\svalhat_i^{1,1},\ldots,\svalhat_i^{n',\ell})$, and
        $\vr_i=(\rval_i^{1,1},\ldots,\rval_i^{n',\ell})$.

        If $\zkmany$ returned $((\vek,\vencval_{\is},\vcomval_{\is}),0)$ for some $\Party_{\is}$, then all parties hard-wire to $f$ the default value for $\Party_{\is}$ and remove $\Party_{\is}$ from the party-set; if $\is\in\bigcup \committee_l$, then every $\Party_i$ adds $\is$ to $\ID$.
    \end{enumerate}

    \item\label{step:rid_decrypt}
    For every active $\committee_l$, every party $\Party^l_j\in\committee_l$ decrypts the vector $\vencval_{j,l}=(\encval_1^{j,l},\ldots,\encval_n^{j,l})$, by computing $\svalhat_i^{j,l}=\Dec_{\dk_{j,l}}(\encval_i^{j,l})$, for every $i\in[n]$. In addition, parse $\svalhat_i^{j,l}=(\sval_i^{j,l},\decomval_i^{j,l})$.

    \item\label{step:rid_coinflip}
    For every active $\committee_l$, all parties invoke $\faugct^{\committee_l}$ with $\committee_l$-identifiable-abort (in parallel);\footnote{More formally, the parties invoke the functionality computing at once $\ell$ instances of $\faugct$, where the \lth instance has $\committee_l$-identifiable-abort.} every $\Party^l_j\in\committee_l$ receives back $(\rndaugct_{j,l},\decomaugct_{j,l},\vcomaugct_l)$ where $\vcomaugct_l=(\comaugct_{1,l},\ldots,\comaugct_{n',l})$ is a public output.
    In case the computation for $\committee_l$ aborts and some party $\Party_{\is} = \Party^l_{\js}\in \committee_l$ is identified as corrupted, every $\Party_i$ adds $\is$ to $\ID$.

    \item\label{step:rid_gmw}
    For every active $\committee_l$, the parties in $\committee_l$ execute the protocol $\pi'$ for computing $\fin{f}{n}{n'}$, parametrized by $\vcomval=(\vcomval_1,\ldots,\vcomval_n)$, over the broadcast channel, where $\Party_j^l\in\committee_l$ uses $(\svalhat_1^{j,l},\ldots,\svalhat_n^{j,l})$ as its input and $\rndaugct_{j,l}$ as its random coins.
    Denote by $\vm_{j,l}=(m_1^{j,l},\ldots,m_p^{j,l})$ the messages $\Party_j^l\in\committee_l$ received during the protocol and let $\outvalue_{j,l}$ be the output $\Party_j^l\in\committee_l$ received (either a valid value $y$ or $(\bot,\is)$ with $\is\in\committee_l$).

    \item\label{step:rid_output}
    For every active $\committee_l$, every $\Party_j^l\in\committee_l$ proves to all parties that $\outvalue_{j,l}$ is indeed its output value, \ie $\Party_j^l$ sends $((\vencval_{j,l},\vm_{j,l},\comaugct_{j,l},\outvalue_{j,l}),(\dk_j,\rndaugct_{j,l},\decomaugct_{j,l}))$ to $\zkmany$ (parametrized by $R_j$).\footnote{More formally, we consider an ideal world for $(n'\cdot\ell)$-bounded parallel computation of $\zkmany$, where every $\Party_j^l\in\committee_l$ acts as the prover in a separate computation.}
    Once party $\Party_i$ receives $((\vencval_{j,l},\vm_{j,l},\comaugct_{j,l},\outvalue_{j,l}),1)$ where the values $(\vencval_{j,l},\vm_{j,l},\comaugct_{j,l})$ match the common view, and $\outvalue_{j,l}= (\bot,\is)$ (for some $\Party_{\is} = \Party^l_{\js}\in \committee_l$), $\Party_i$ adds $\is$ to $\ID$; otherwise, $\Party_i$ adds $\outvalue_{j,l}$ to the (initially empty) set $\OUT_i$.
    If $\Party_i$ receives an invalid proof from some party $\Party_{\is} = \Party^l_{\js}\in \committee_l$, $\Party_i$ adds $\is$ to $\ID$.

    \item\label{step:rid_term}
    If $\OUT_i\neq\emptyset$, party $\Party_i$ arbitrarily chooses $y\in\OUT_i$ and outputs $y$.
    Otherwise, $\Party_i$ outputs $(\bot,\ID)$.
\end{enumerate}
\end{construction}


Let $\Adv$ be an adversary attacking the execution of protocol $\pi$ in the $(\faugct,\zkmany)$-hybrid model and let $\IS\subseteq[n]$ be a subset of size at most $t$, satisfying $\ssize{\IS\cap\committee_l}\leq t'$, for every $l\in[\ell]$.
We construct the following adversary $\Sim$ for the ideal model computing $f$ with fairness and $\vCS$-identifiable-abort.
On inputs $\set{x_i}_{i\in\IS}$ and auxiliary input $\aux$, the simulator $\Sim$ starts by emulating $\Adv$ on these inputs.
$\Sim$ plays towards $\Adv$ the roles of the honest parties and the ideal functionalities $\faugct$ and $\zkmany$.
\Sim initializes empty sets $\ID$ and $\OUT$.
For simplicity, assume that all input values and random strings are elements in $\zo^\secParam$.

To simulate Step~\ref{step:rid_pk}, the simulator $\Sim$ sends to \Adv public keys for honest parties in $\bigcup\committee_l$ and receives from \Adv public keys for corrupted parties in $\bigcup \committee_l$; if \Adv does not provide a public key for $\Party_{\is}=\Party_{\js}^l\in\bigcup\committee_l$, the simulator $\Sim$ adds $\is$ to $\ID$. Denote by $\ek_{j,l}$ the public key for $\Party_j^l\in\committee_l$ and let $\vek=(\ek_{1,1},\ldots,\ek_{n',\ell})$.

To simulate Step~\ref{step:rid_input}, the simulator \Sim proceeds as follows for every honest party $\Party_i$.
In Step~\ref{step:rid_input_encrypt}, for every active $\committee_l$ (satisfying $\committee_l\cap\ID=\emptyset$), compute secret shares of zero, $0^{2\secParam}=\bigoplus_{j\in[n']}{\sval_i^{j,l}}$, commit to each share as $\comval_i^{j,l}=\Com(\sval_i^{j,l};\decomval_i^{j,l})$, denote $\svalhat_i^{j,l}=(\sval_i^{j,l},\decomval_i^{j,l})$, and encrypt $\encval_i^{i,j}\gets\Enc_{\ek_{j,l}}(\svalhat_i^{j,l})$.
In Step~\ref{step:rid_input_prove}, for every active $\committee_l$, denote $\vencval_i=(\encval_i^{1,1},\ldots,\encval_i^{n',\ell})$, and $\vcomval_i=(\comval_i^{1,1},\ldots,\comval_i^{n',\ell})$.
Next, $\Sim$ simulates $\zkmany$ as follows. On behalf of every honest party $\Party_i$, the simulator $\Sim$ sends $((\vek,\vencval_i,\vcomval_i),1)$ to $\Adv$; on behalf of every corrupted party $\TParty_i$, the simulator $\Sim$ receives values $((\vek,\vencval_i,\vcomval_i),(\vsvalhat_i,\vr_i))$ from $\Adv$, verifies the relation $\Renc$ and answers $\Adv$ accordingly.

To simulate Step~\ref{step:rid_coinflip}, for every active $\committee_l$, the simulator $\Sim$ emulates $\faugct^{\committee_l}$ by sampling random strings $(\rndaugct_{j,l},\decomaugct_{j,l})$ for parties $\Party_j^l\in\committee_l$, computing $\comaugct_{j,l}=\com(\rndaugct_{j,l};\decomaugct_{j,l})$ and setting $\vcomaugct_l=(\comaugct_{1,l},\ldots,\comaugct_{n',l})$. Next, \Sim hands $(\rndaugct_{j,l},\decomval_{j,l},\vcomaugct_l)$ to corrupted parties $\TParty_j\in\committee_l$ and $(\lambda,\vcomaugct_l)$ to corrupted parties outside of $\committee_l$. In case $\Sim$ receives $(\abort,\is)$ with $\is\in\IS\cap\committee_l$ from $\Adv$, it responds with $(\bot,\is)$ and adds $\is$ to $\ID$.

Next, the simulator $\Sim$ uses the simulator $\tilde{\Sim}$ that is guaranteed to exist for the $\ell$-times parallel execution of $\pi'$ when interacting with the residual adversary of $\Adv$ in Step~\ref{step:rid_gmw}.
The simulator $\Sim$ invokes $\tilde{\Sim}$ on input values $\sset{\vsvalhat_{j,l}}_{j\in\committee_l\cap\IS}$ for the \lth execution of $\pi'$ with parties in $\committee_l$ (where the values $\vsvalhat_{j,l}=(\svalhat_1^{j,l},\ldots,\svalhat_n^{j,l})$ are obtained during the simulation of Step~\ref{step:rid_input_prove}) and on auxiliary information containing $z$, the input values $\sset{x_i}_{i\in\IS}$, and the transcript of the simulation until this point.
If $\tilde{\Sim}$ sends $(\abort,\is)$ for the \lth computation, with $\is\in\IS\cap\committee_l$, the simulator $\Sim$ adds $\is$ to $\ID$.
Otherwise, $\tilde{\Sim}$ sends input values $\sset{\tilde{\vs}_{j,l}}_{j\in\committee_l\cap\IS}$ with $\tilde{\vs}_{j,l}=(\tilde{\sval}_1^{j,l},\ldots,\tilde{\sval}_n^{j,l})$.
The simulator \Sim verifies that $\tilde{\sval}_i^{j,l}=(\sval_i^{j,l},\decomval_i^{j,l})$ and that $\comval_i=\Com(\sval_i^{j,l};\decomval_i^{j,l})$; if the verifications fails, \Sim responds to $\tilde{\Sim}$ with $(\bot,\js)$ (for the smallest such $\js$) and adds $\is$ to $\ID$, such that $\Party_{\is}=\Party_{\js}^l\in\committee_l$.
In case no computation has valid inputs from $\Adv$, the simulator \Sim sends the set $\ID$ to its trusted party and receives back $(\bot,\ID)$.
Otherwise, $\Sim$ reconstructs the input values $\sset{x'_i}_{i\in\IS}$, by computing $x'_i=\bigoplus_j{\sval_i^{j,l}}$ for an arbitrary computation $l$ with valid inputs.
Next, \Sim sends $\sset{x'_i}_{i\in\IS}$ to the trusted party, receives output $y$ and forwards $y$ to $\tilde{\Sim}$ for every computation with valid inputs.
In both cases, $\Sim$ receives the output from $\tilde{\Sim}$, which contains the view of the adversary, and interacts with $\Adv$ accordingly.

To simulate Step~\ref{step:rid_output}, the simulator $\Sim$ simulates $\zkmany$. The simulator $\Sim$ sends, on behalf of every honest party $\Party_j^l$ in an active $\committee_l$, the message $((\vencval_{j,l},\vm_{j,l},\comaugct_{j,l},\outvalue_{j,l}),1)$ to every corrupted party, where $\vm_{j,l}$ is obtained from the output of $\tilde{\Sim}$ and $\outvalue_{j,l}$ is either the output value $y$ or $(\bot,\is)$ where $\Party_{\is}=\Party_{\js}^l\in\committee_l$ is the identified corrupted party. In addition, $\Sim$ receives $((\vencval_{j,l},\vm_{j,l},\comaugct_{j,l},\outvalue_{j,l}),(\dk_j,\rndaugct_{j,l},\decomaugct_{j,l}))$ from $\Adv$ on behalf of every corrupted party $\Party_j$ in $\committee_l$ and verifies the relation according to $R_j$. Finally, $\Sim$ outputs whatever $\Adv$ outputs and halts.

\paragraph{Proving indistinguishability.}
We prove computational indistinguishability between the real execution of the compiled protocol $\pi$ running with adversary \Adv and the ideal computation of $f$ running with \Sim via a series of hybrids experiments. The output of each experiment is the output of the honest parties and of the adversary.

\paragraph{The game $\HYB^1_{\pi, \IS, \Adv(z)}(\vx,\secParam)$.}
This game is defined to be the execution of the protocol $\pi$ in the $(\faugct,\zkmany)$-hybrid model on inputs $\vx\in(\zs)^n$ and security parameter $\secParam$ with adversary \Adv running on auxiliary information $z$ and controlling parties in $\IS$.

\paragraph{The game $\HYB^2_{\pi, \IS, \Adv(z)}(\vx,\secParam)$.}
In this game, we modify $\HYB^1_{\pi, \IS, \Adv(z)}(\vx,\secParam)$ as follows.
Whenever an honest party invokes $\zkmany$ with $(x,w)$ (in Steps~\ref{step:rid_input_prove} and~\ref{step:rid_output}), all parties receive output $(x,1)$ without checking if $w$ is a witness for $x$.

\begin{claim}\label{claim:hyb1}
$\sset{\HYB^1_{\pi, \IS, \Adv(z)}(\vx,\secParam)}_{\vx,z,\secParam} \approx \sset{\HYB^2_{\pi, \IS, \Adv(z)}(\vx,\secParam)}_{\vx,z,\secParam}$.
\end{claim}
\begin{proof}
This is immediate since honest parties always send valid witnesses to $\zkmany$.
\end{proof}

\paragraph{The game $\HYB^3_{\pi, \IS, \Adv(z)}(\vx,\secParam)$.}
In this game, we modify $\HYB^2_{\pi, \IS, \Adv(z)}(\vx,\secParam)$ as follows.
In Step~\ref{step:rid_gmw}, instead of running $\ell$ instances of protocol $\pi'$ in parallel, use the simulator $\tilde{\Sim}$ that is guaranteed to exists for the residual adversary \Adv.

More specifically, $\tilde{\Sim}$ is invoked on input values $\sset{\vs_{j,l}}_{j\in\committee_l\cap\IS}$ for the \lth execution of $\pi'$ with parties in $\committee_l$ (where $\vs_{j,l}=(\sval_1^{j,l},\ldots,\sval_n^{j,l})$ are the values sent by $\Adv$ to $\zkmany$ in Step~\ref{step:rid_input}) and on auxiliary information containing $z$, the input values $\sset{x_i}_{i\in\IS}$ and the transcript of the experiment until this point.
If $\tilde{\Sim}$ sends $(\abort,\is)$ for the \lth computation, with $\is\in\IS\cap\committee_l$, add $\is$ to $\ID$.
Otherwise, $\tilde{\Sim}$ sends input values $\sset{\tilde{\vs}_{j,l}}_{j\in\committee_l\cap\IS}$ with $\tilde{\vs}_{j,l}=(\tilde{\sval}_1^{j,l},\ldots,\tilde{\sval}_n^{j,l})$; If for some $j\in\committee_l\cap\IS$ the signed values are not verified, reply with $(\bot,j)$ to $\tilde{\Sim}$ and add $j$ to $\ID$.
If there exist instances with verified inputs, reconstruct the input values $\sset{x'_i}_{i\in\IS}$, where $x'_i=\Recon(\tilde{\sval}_i^{1,l}, \ldots, \tilde{\sval}_i^{n',l})$ for an arbitrary computation $l$ with valid inputs (where for $j\notin \committee_l\cap\IS$, set $\tilde{\sval}_i^{j,l}=\sval_i^{j,l}$), and compute $f$ over these inputs and the honest parties' inputs to obtain the output $y$ and forward $y$ to $\tilde{\Sim}$ for every computation with valid inputs.
Next, use the output from $\tilde{\Sim}$, which contains the view of the adversary, in order to interacts with $\Adv$.

\begin{claim}\label{claim:hyb2}
$\sset{\HYB^2_{\pi, \IS, \Adv(z)}(\vx,\secParam)}_{\vx,z,\secParam} \ci \sset{\HYB^3_{\pi, \IS, \Adv(z)}(\vx,\secParam)}_{\vx,z,\secParam} $.
\end{claim}
\begin{proof}
This follows from the security of the simulator $\tilde{\Sim}$.
Indeed, the ability to distinguish between $\HYB^2$ and $\HYB^3$ with non-negligible advantage implies the same advantage in distinguishing between the simulator $\tilde{\Sim}$ and the execution of the $\ell$-times parallel execution of $\pi'$ when interacting with the residual adversary of $\Adv$ in Step~\ref{step:rid_gmw}.
\end{proof}

\paragraph{The games $\HYB^{4,\is,\js,\ls}_{\pi, \IS, \Adv(z)}(\vx,\secParam)$, for $1\leq \is\leq n$, $1\leq \js\leq n'$ and $1\leq \ls\leq \ell$.}
In these games, we modify $\HYB^3_{\pi, \IS, \Adv(z)}(\vx,\secParam)$ as follows.
For every honest party $\Party_i$, in addition to generating $(\sval_i^{1,l},\ldots,\sval_i^{n',l})\gets\Share(x_i)$, generate secret shares of zero as $(\tilde{\sval}_i^{1,l},\ldots,\tilde{\sval}_i^{n',l})\gets\Share(0)$.
Next, for $(i,j,l)< (\is,\js,\ls)$ (\ie for $i< \is$, or $i=\is$ and $j<\js$, or $i=\is$, $j=\js$ and $l<\ls$), compute $c_i^{j,l}\gets\Enc_{pk_j}(\tilde{\sval}_i^{j,l})$, and for $(i,j,l)\geq(\is,\js,\ls)$, compute $c_i^{j,l}\gets\Enc_{pk_j}(\sval_i^{j,l})$.
\begin{claim}\label{claim:hyb3}
$\sset{\HYB^3_{\pi, \IS, \Adv(z)}(\vx,\secParam)}_{\vx,z,\secParam} \ci \sset{\HYB^{4,n.n',\ell}_{\pi, \IS, \Adv(z)}(\vx,\secParam)}_{\vx,z,\secParam}$.
\end{claim}
\begin{proof}
Note that $\HYB^3_{\pi, \IS, \Adv(z)}(\vx,\secParam) \approx \HYB^{4,1,1,1}_{\pi, \IS, \Adv(z)}(\vx,\secParam)$.
For every $i\in[n]$, $j\in[n']$ and $l\in[\ell-1]$, it holds that $\HYB^{4,i,j,l}_{\pi, \IS, \Adv(z)}(\vx,\secParam) \ci \HYB^{4,i,j,l+1}_{\pi, \IS, \Adv(z)}(\vx,\secParam)$;
similarly, for every $i\in[n]$, $j\in[n'-1]$ it holds that $\HYB^{4,i,j,\ell}_{\pi, \IS, \Adv(z)}(\vx,\secParam) \ci \HYB^{4,i,j+1,1}_{\pi, \IS, \Adv(z)}(\vx,\secParam)$ and for every $i\in[n-1]$ it holds that $\HYB^{4,i,n',\ell}_{\pi, \IS, \Adv(z)}(\vx,\secParam) \ci \HYB^{4,i+1,1,1}_{\pi, \IS, \Adv(z)}(\vx,\secParam)$.
Otherwise, since the simulated computation of $\pi'$ is independent of the honest parties' inputs, there exists an attack on the semantic security of the encryption scheme.
The claim follows using a standard hybrid argument.
\end{proof}

\paragraph{The game $\HYB^5_{\pi, \IS, \Adv(z)}(\vx,\secParam)$.}
In this game, we modify $\HYB^{4,n.n',\ell}_{\pi, \IS, \Adv(z)}(\vx,\secParam)$ as follows.
Instead of computing $f$ on $\sset{x'_i}_{i\in\IS}$ and on the input values of the honest parties, send $\sset{x'_i}_{i\in\IS}$ (or the set $\ID$ in case all computations are aborted) to the ideal functionality computing $f$ in the fair ideal model with $(\committee_1,\ldots,\committee_\ell)$-identifiable abort and get back the output $y$.

\begin{claim}\label{claim:hyb4}
$\sset{\HYB^{4,n.n',\ell}_{\pi, \IS, \Adv(z)}(\vx,\secParam)}_{\vx,z,\secParam} \approx \sset{\HYB^5_{\pi, \IS, \Adv(z)}(\vx,\secParam)}_{\vx,z,\secParam} $.
\end{claim}
\begin{proof}
This is immediate since the ideal model computes $f$ on the honest parties' inputs as required.
\end{proof}

The proof of \cref{lem:idfair_to_ridfair} now follows since $\HYB^5_{\pi, \IS, \Adv(z)}(\vx,\secParam)$ exactly describes the simulation done by $\Sim$, and in particular, does not depend on the input values of honest parties.
\end{proof}

\begin{remark}\label{remark:reactive}
The result in \cref{thm:idfair_to_ridfair}, of computing an $n$-party functionality $f$ with fairness and restricted identifiable abort, is achieved by having small committees compute $\fin{f}{n}{n'}$, over committed values, with fairness and identifiable abort. This technique can be extended in a straightforward way to construct a secure computation of $f$ with restricted identifiable abort, by having each of the committees securely compute $\fin{f}{n}{n'}$ without fairness, only with identifiable abort (where the adversary can abort the computation in each committee after learning the output). Moreover, the technique can be extended to a computation with restricted identifiable abort of \emph{single-input reactive functionalities} (see \cref{sec:reactive}) that receive inputs from the parties only at the first call, \ie the output of all proceeding calls is determined by the input for the first call. Indeed, this translates to having the committees compute $\fin{f_i}{n}{n'}$ (where $f_i$ is the \ith function of the single-input reactive functionality) over the same committed values that were provided by all the parties, and prove in zero knowledge that the outcome of every computation is correct.
Looking ahead, this observation will turn out useful in \cref{sec:apps:or}.
\end{remark}

\subsection{Fairness to Full Security with an Honest Majority (With Inputs)}\label{sec:fairtofull_withinput_HM}

\subsubsection{Fairness with Restricted Identifiable Abort with an Honest Majority}\label{sec:fair_to_ridfair_HM}
In this section, we show that in the case that an honest majority is guaranteed, the reduction from fairness with restricted identifiable abort to fairness can be much more elegant, and more importantly, be based on much simpler tools. Specifically, we devise a compiler, similar to the one for the no-honest-majority case that is based solely on error-correcting secret-sharing schemes (ECSS), which exist unconditionally (see \cref{def:ECSS}).
\begin{theorem}\label{thm:idfair_to_ridfair_hm}
Let $f$, $n'$, $\ell$, and $\vCS$ as in \cref{thm:idfair_to_ridfair}, in addition, let $t<n/2$ and $t'<n'/2$.
There exists a \ppt algorithm $\CompilerHM{n'}{n}$ such that for any $n'$-party, \rchanged{$r$}{$r'$}-round protocol $\pi'$ computing $\finout{f}{n}{t'}{n'}$, the protocol $\pi = \CompilerHM{n'}{n}(\pi',\vCS)$ is an $n$-party, \rchanged{$O(t'\cdot r)$}{$O(t'\cdot r')$}-round protocol computing $f$ with the following guarantee, in the information-theoretic (statistical) setting.

If the number of corrupted parties in every $\committee_j$ is smaller than $t'$, and $\pi'$ is a protocol that \rchanged{$(\delta,t')$}{$(\delta',t')$}-securely computes $\finout{f}{n}{t'}{n'}$ with abort, $\ell$ times in parallel, then $\pi$ is a protocol that \rchanged{$(\ell\cdot\delta,t)$}{$(\ell\cdot\delta',t)$}-securely computes $f$ with fairness and $\vCS$-identifiable-abort.
Furthermore, the compiler is black-box \wrt the protocol $\pi'$.
\end{theorem}

\begin{proof}
Initially, following \cref{lem:abort_to_idabort_noinput_HM} we can assume that $\pi'$ is a protocol that \rchanged{$(\delta,t')$}{$(\delta',t')$}-securely computes $\finout{f}{n}{t'}{n'}$ with identifiable abort.
Next, we construct $\CompilerHM{n'}{n}$ by adjusting the compiler for the no-honest-majority setting (\cref{prot:fair_to_ridfair}) as follows.

The compiler is initially defined in the Setup-Commit-then-Proof hybrid model (defined in~\cite{IOZ14}). Upon first invocation of the trusted party (the setup-commit phase), every party $\Party_i$ obtains correlated randomness $\vr_i=(\rinput^i,\vrmask^i, \vrprot^i)$, where $\vrmask^i=(\rmask^{i,1},\ldots,\rmask^{i,n'})$ and $\vrprot^i=(\rprot^{i,1},\ldots,\rprot^{i,\ell})$. The committed correlated randomness is used as follows:
\begin{itemize}
    \item
    $\rinput^i$ is used to commit to the input of $\Party_i$ by broadcasting $x_i\xor \rinput^i$.
    \item
    $\vrmask^i$ is used to mask the communication (over the broadcast channel) between $\Party_i$ and $\Party_j$. These are pairwise correlated values, \ie $\rmask^{i,j}=\rmask^{j,i}$.
    \item
    $\vrprot^i$ consists of randomness for executing the protocol $\pi'$. If $\Party_i\in\committee_l$, then $\Party_i$ will use $\rprot^{i,l}$ as its committed randomness for the protocol executed by $\committee_l$.
\end{itemize}
Upon future calls to the trusted party (the prove phase), every party $\Party_i$ can prove statements using the secret witness $\vr_i$.
For simplicity we assume that the functionality $f$ is deterministic; randomized functionalities can be made deterministic using standard techniques, by having the trusted party for the setup-commit phase add shares of the same random string (using $(t'+1)$-out-of-$n'$ ECSS) to the members of every committee.
The $n$-party protocol $\pi = \CompilerHM{n'}{n}(\pi',\vCS)$ is defined as follows.
\begin{construction}(security with abort to fairness with restricted identifiable abort)\label{prot:abort_to_ridfair_hm}
\begin{itemize}
    \item\textbf{Hybrid Model:}
    The protocol is defined in the Setup-Commit-then-Prove hybrid model with full security.
    \item\textbf{Common Input:}
    An $n'$-party protocol $\pi'$ and $\vCS=(\committee_1,\ldots,\committee_\ell)$. We use the notation $\Party^l_j$ to refer to the \jth party in $\committee_l$.
    \item\textbf{Private Input:}
    Every party $\Party_i$, for $i\in[n]$, has private input $x_i\in\zs$.
    \item\textbf{The Protocol:}
    Let $\ID$ denote the (initially empty) set of identified corrupted parties. At any point in the protocol, we call a committee $\committee_l$ \emph{active} if $\committee_l\cap\ID=\emptyset$.
\end{itemize}
\begin{enumerate}
    \item\label{step:rid_hm_input}
    Every party $\Party_i$ proceeds as follows:
    \begin{enumerate}
        \item\label{step:rid_hm_setupcommit}
        Call the trusted party to obtain the correlated randomness $\vr_i=(\rinput^i,\vrmask^i, \vrprot^i)$.
        \item\label{step:rid_hm_broadcast}
        Broadcast $x_i \xor \rinput^i$.
        \item\label{step:rid_hm_input_encrypt}
        For every active $\committee_l$, secret share its input as $(\sval_i^{1,l},\ldots,\sval_i^{n',l})\gets\Share(x_i)$, and send $\sval_i^{j,l}$ (masked using the pairwise correlated randomness) to $\Party_j^l$.
        \item\label{step:rid_hm_input_prove}
        Prove that it sent to every committee shares of its committed input value, using the committed correlated randomness. If a party $\Party_{\is}$ fails to provide an accepting proof, it is removed from the party set, and if $\is\in\bigcup \committee_l$ then $\is$ is added to the list of identified parties $\ID$.
    \end{enumerate}

    \item\label{step:rid__hm_compute}
    For every active $\committee_l$, every party $\Party^l_j\in\committee_l$ unmasks the vector $\vsval_{j,l}=(\sval_1^{j,l},\ldots,\sval_n^{j,l})$ (using $\vrmask$), and participates in the execution of protocol $\pi'$ for computing $\finout{f}{n}{t'}{n'}$, over the broadcast channel, where $\Party_j^l\in\committee_l$ uses $\vsval_{j,l}$ as its input and its random coins from $\vrprot$.
    Denote by $\vm_{j,l}=(m_1^{j,l},\ldots,m_p^{j,l})$ the messages $\Party_j^l\in\committee_l$ received during the protocol and let $\outvalue_{j,l}$ be the output $\Party_j^l\in\committee_l$ received (either a valid value $y$ or $(\bot,\is)$ with $\is\in\committee_l$).

    \item\label{step:rid_hm_output}
    For every active $\committee_l$, every $\Party_j^l\in\committee_l$ broadcasts $\outvalue_{j,l}$ and proves to all parties that this is indeed its output value and that $\outvalue_{j,l}$ is consistent with the common view and its committed correlated randomness.
    If party $\Party_i$ receives output values $(\outvalue_{1,l},\ldots,\outvalue_{n',l})$ from a committee $\committee_l$, where at least $n'-t'$ are valid, $\Party_i$ reconstructs the output $y$ and outputs it (take the minimal $l$ if several committees send valid output values). Otherwise, every committee identified a corrupted party $\Party_{\is} = \Party^l_{\js}\in \committee_l$; every party adds each such $\is$ to $\ID$, and outputs $(\bot,\ID)$.
\end{enumerate}
\end{construction}

The simulation follows standard techniques. Given an adversary $\Adv$, the simulator $\Sim$ first simulates honestly the first call to the Setup-Commit-then-Prove functionality. Next, $\Sim$ commits to zero for every honest party (\ie sends a random string to $\Adv$), receives from $\Adv$ commitments to the input values of the corrupted parties, and extracts the values $\sset{x'_i}_{i\in\IS}$. The simulator $\Sim$ continues by handing $\Adv$ (masked) secret shares of zero on behalf of the honest parties to every corrupted party which is a member on some committee, and verifies that the proof is correct. In addition, $\Sim$ receives from $\Adv$ the values on behalf of the corrupted parties and verifies their validity.

In order to simulate Step~\ref{step:rid__hm_compute}, the simulator $\Sim$ uses the simulator $\tilde{\Sim}$ for the $\ell$-bounded parallel composition of the protocol $\pi'$. For every committee that $\tilde{\Sim}$ does not abort its computations, $\Sim$ hands secret shares of zero as the output values for the corrupted parties. If there exist committees for which $\tilde{\Sim}$ did not abort the computation, $\Sim$ sends the values $\sset{x'_i}_{i\in\IS}$ to the trusted party computing $f$ and receives back the output $y$. Next, $\Sim$ computes secret sharing of $y$ for every non-aborting committee, and sends them to $\tilde{\Sim}$ on behalf of honest committee members. In case $\tilde{\Sim}$ aborts all computation, $\Sim$ updates the set of identified parties $\ID$, and hands the trusted party $(\abort,\ID)$. Finally, $\Sim$ interacts with $\Adv$ according to $\tilde{\Sim}$ and outputs whatever $\Adv$ outputs. The proof follow in a standard way via the security of the protocol $\pi'$ and of the ECSS scheme.

To complete the proof of the theorem, we use the fully secure, constant-round implementation of the Setup-Commit-then-Prove functionality in the honest-majority setting, given in~\cite[Lem.\ 6.2]{CCGZ17}.
\end{proof}

\subsection{Applications}\label{sec:apps_withinput}

In \cref{sec:apps:or}, we show how to reduce the round complexity for computing Boolean Or from linear to super-constant. In \cref{sec:apps:iklp}, we show how to improve the round complexity in the best-of-both-worlds result of~\cite{IKKLP11}. Finally, in \cref{sec:apps:ikpsy}, we show how to improve the round complexity in a recent result from~\cite{IKPSY16}.

\rnote{Are there applications for $1/p$ security???}

\subsubsection{Multiparty Boolean OR}\label{sec:apps:or}

\citet[Thm.\ 9]{GordonK09} constructed a fully secure protocol that computes the $n$-party Boolean OR functionality and tolerates an arbitrary number of corruptions.
Loosely speaking, every party initially broadcasts a commitment of its input (all parties output $1$ if some party does not commit). Next, all parties iteratively run a protocol computing the \emph{committed OR functionality} $\fcomor$ with identifiable abort, until the result is obtained.
The functionality $\fcomor$ verifies that each party provided the same set of commitments as well as a valid decommitment of its own input. If so the functionality computes Boolean OR; otherwise, it notifies to each party which other parties are not consistent with him, where in the latter case every party continues the protocol only with parties that are consistent with him.
In case the protocol outputs $\bot$ and a corrupted party is identified, all parties proceed to the next iteration without the identified party (this can be simulated due to the special properties of the Boolean OR function).
It follows that the protocol in~\cite{GordonK09} requires $t+1$ sequential calls to $\fcomor$.

\begin{nfbox}{The reactive committed OR functionality}{fig:fcomor}
\small
\begin{center}
    \textbf{The functionality} $\fcomor$
\end{center}
The functionality $\fcomor$ is run by party-set $\PS$ of size $n$, parametrized by a commitment scheme $\Com$ and a vector of commitments $\vc=(c_1,\ldots,c_n)$.

\begin{itemize}
    \item
    Upon the first invocation, with input $(x_i,\rho_i)$ for party $\Party_i\in\PS$, proceed as follows:
    \begin{enumerate}
        \item
        For each party $\Party_i\in\PS$, if $\Com(x_i; \rho_i) \neq c_i$, add $\Party_i$ to the (initially empty) set $\MS$.
        \item
        For every $\Party_i\in\MS$, set $x'_i=0$; for every $\Party_i\in\PS\setminus\MS$, set $x'_i=x_i$.
        \item
        Compute $y=x'_1\vee\ldots\vee x'_n$.
        \item
        Return $\MS$ to all parties and store $y$ as the internal secret state.
    \end{enumerate}
    \item
    Upon the second invocation, return $y$ to all parties.
\end{itemize}

\end{nfbox}

We next show how to drastically improve the round complexity of computing Boolean OR, when a constant fraction of parties are honest, by computing $\fcomor$ with restricted identifiable abort. Note that a direct application of \cref{thm:ridfair_to_full_withinput} does not help, since the protocol in~\cite{GordonK09} is already fully secure and cannot be made secure with fair abort in a meaningful way.\footnote{In fact, as pointed out in~\cite{CL17}, every fair protocol for Boolean OR can be immediately transformed to a fully secure protocol.} The idea is to iteratively invoke $\fcomor$ with restricted identifiable abort, and eliminate identified corrupted parties. Intuitively, for appropriate parameters, the adversary can abort the computation only in a limited (super-constant) number of invocations, after which the honest parties are guaranteed to obtain the output. However, a closer look shows that \emph{any} party can in fact provide an invalid input to $\fcomor$ and force the output to be its identity (as it disagrees with all other parties). This will result in a linear number of invocations of $\fcomor$ before terminating the protocol.

We overcome this obstacle by slightly modifying the definition of $\fcomor$ from~\cite{GordonK09}. Initially, we parametrize $\fcomor$ by the vector of commitments that were sent over the broadcast channel. Second, we would like to ensure that $\fcomor$ will compute the Boolean OR of the inputs values even if some parties provide invalid decommitments as their inputs. We cannot simply replace the input bit of parties that failed to provide a valid decommitment with $0$ (\ie ignore those parties), since this will lead to the following attack. $\Adv$ commits on behalf of some corrupted $\Party_i$ to $1$ and for all other corrupted parties to $0$; upon the first call to $\fcomor$, it gives an invalid decommitment for $\Party_i$ and aborts the computation after receiving the output bit by identifying a corrupted $\Party_j$. Since the bits of all corrupted parties are treated as $0$ in this case, \Adv learns whether all honest parties have $0$ or not. Next, $\fcomor$ is invoked without $\Party_j$, and \Adv sends valid decommitments for all corrupted parties; in this case the output will be $1$. Clearly, such an attack cannot be simulated in the fully secure ideal model.
Therefore, we consider a two-phase single-input reactive version of $\fcomor$, formally defined in \cref{fig:fcomor}. In the first call, the functionality outputs the set of parties that provided invalid decommitments, and stores the result of Boolean OR without those parties as its internal secret state. In the second invocation all parties receive the output bit.

\begin{corollary}[restating \cref{cor:intro_OR}]
Assume that \TDP, \CRH, and non-interactive perfectly binding commitment schemes exist.
Then, the $n$-party Boolean OR functionality can be computed with full security in $\omega(1)$ rounds facing $t=\beta n$ corruptions, for $0<\beta<1$.
\end{corollary}

\begin{proof}
Assume for simplicity of exposition that $n=\Omega(\log(\secParam)\cdot \varphi(\secParam))$.
Let $m=\log(\secParam)\cdot\varphi(\secParam)$ with $\varphi(\secParam)=\omega(1)$, $n'=m-\log(\secParam)/\varphi(\secParam)$, and $\ell=\secParam^\uglyExp$. We construct the Boolean OR protocol in the $(\felect,\fcomor)$-hybrid model, where $\felect$ is computed with full security and $\fcomor$ with (unfair) $(\ell, n',n'-1)$-identifiable-abort. The corollary will follow from \cref{lem:Feige} and \cref{thm:idfair_to_ridfair}, since under the assumptions in the corollary and using the protocol from~\cite{Pass04}, the functionality $\fin{\fcomor}{n}{n'}$ can be computed in constant rounds with identifiable abort, $\ell$ times in parallel, facing $n'-1$ corruptions; hence $\fcomor$ can be computed with $(\ell, n',n'-1)$-identifiable-abort. Note that the reactive functionality $\fcomor$ receives inputs from the parties only for the first call, and the output of the second call is determined deterministically by the input values to the first call; therefore, following \cref{remark:reactive}, the functionality $\fin{\fcomor}{n}{n'}$ can indeed be computed with $(\ell, n',n'-1)$-identifiable-abort.

The protocol proceeds as follows with party-set $\PS=\sset{\Party_1,\ldots,\Party_n}$. Initially, every party $\Party_i$ broadcasts a commitment of its inputs $c_i=\Com(x_i;\rho_i)$; in case some party didn't broadcast, all parties output $1$ and halt. Next, denote $\beta'=(1+\beta)/2$, the parties invoke $\felect(n,m,\beta')$ to elect a committee $\committee\subseteq[n]$ of size $m$. Denote by $\vCS=(\committee_1,\ldots,\committee_\ell)$ all subsets of $\committee$ of size $n'$ (following \cref{claim:par_poly_noinput} there are at most $\ell$ such subsets).

The parties proceed by iteratively invoking the two-phase ideal functionality computing $\fcomor(\PS,\vc)$ with $\vCS$-identifiable-abort. In the first call, every party $\Party_i$ sends $(x_i,\rho_i)$ as its input and receives back a subset $\MS\subseteq\PS$ or $(\abort,\ID)$; in the latter case $\Party_i$ sets $\PS\setminus\ID$ and proceeds to the next iteration. If the first call completed successfully without abort, every party invokes $\fcomor$ for the second time and receives a bit $y\in\zo$ or $(\abort,\ID)$; in the former case party $\Party_i$ outputs $y$ and halts whereas in the latter case, $\Party_i$ sets $\PS=\PS\setminus(\ID\cup\MS)$ and proceeds to the next iteration.

Similarly to~\cite{GordonK09}, the idea behind the simulation is that if the adversary aborts a computation of $\fcomor$, then it learns new information on the honest parties' inputs only if it sets all the corrupted parties' inputs to $0$. Following the definition of $\fcomor$, the adversary can learn the result using inputs $0$ for all corrupted parties only if it committed to zeros in the first round or if it sends invalid decommitments in the first invocation of $\fcomor$ in some iteration and did not abort it; in both cases \Adv cannot use other inputs and force honest parties to output $1$ in later invocations.

Let $\Adv$ be an adversary attacking the protocol in the $(\felect,\fcomor)$-hybrid model and let $\IS$ be the set of corrupted parties.
We construct a simulator $\Sim$ for the ideal model computing Boolean OR with full security, as follows.
On inputs $\set{x_i}_{i\in\IS}$ and auxiliary input $\aux$, the simulator $\Sim$ starts by emulating $\Adv$ on these inputs.
Initially, $\Sim$ broadcasts commitments to $0$ on behalf of honest parties and receives commitments from $\Adv$ on behalf of corrupted parties. In case \Adv didn't send commitments for all corrupted parties, \Sim send to the trusted party $1$ on behalf of every corrupted party, outputs whatever \Adv outputs and halts. In the second round, \Sim emulates towards $\Adv$ the committee-election functionality $\felect(n,m,\beta')$, \ie $\Sim$ partitions the honest parties to $n/m$ subsets, under the condition that every subset has at least $(1-\beta')n$ parties, hands them to $\Adv$ and receives back a committee $\committee$ of size $m$ that contains exactly one of the subsets (and doesn't intersects the other).

Next, $\Sim$ simulates $\fcomor$ to $\Adv$ in every iteration as follows:
\begin{itemize}
    \item \textbf{Simulating the first call:}
    \begin{itemize}
        \item
        If $\Adv$ sends $(\abort,\JS)$, the simulator \Sim responds with $(\bot,\ID)$ and proceeds to the simulation of the next iteration without the paries in $\ID$.
        \item
        Otherwise, $\Adv$ provided $(x_i,\rho_i)$ for every corrupted party $\Party_i$ that was not yet identified. \Sim prepares the list $\MS$ of corrupted parties with invalid decommitments, and sends $\MS$ to \Adv.
        \item
        If \Adv responds with $(\abort,\JS)$, the simulator \Sim answers with $(\bot,\ID)$ and proceeds to the simulation of the next iteration without the paries in $\ID$.
        \item
        If \Adv responds with \continue, the simulator \Sim proceeds to the simulation of the second call.
    \end{itemize}
    \item \textbf{Simulating the second call:}
    \begin{itemize}
        \item
        If $\Adv$ sent $(\abort,\JS)$, the simulator \Sim responds with $(\bot,\ID)$ and proceeds to the simulation of the next iteration without the paries in $\ID\cup\MS$.
        \item
        Otherwise, if one of the inputs (with valid decommitment) sent by \Adv is $1$, \Sim responds to \Adv with $1$; if all (valid) inputs received from \Adv are $0$, \Sim responds to \Adv with the output bit $y\in\zo$, where in case $y$ is not set yet, \Sim sends to the trusted party $0$ on behalf of every corrupted party and receives back the output $y$.
        \item
        If \Adv responds with $(\abort,\ID)$, the simulator \Sim answers with $(\bot,\ID)$ and proceeds to the simulation of the next iteration without the paries in $\ID\cup\MS$.
        \item
        If \Adv responds with \continue, the simulator \Sim stops simulating $\fcomor$ towards \Adv.
    \end{itemize}
\end{itemize}
Next, if \Sim did not send inputs to the trusted party yet, it send $1$ on behalf of every corrupted party (and receives back the output $1$). Finally, \Sim outputs whatever \Adv outputs and halts.

Proving computational indistinguishability between the output of the honest parties and of \Adv in the execution of the protocol and the output of the honest parties and of \Sim in the fully secure ideal model follows in similar lines to the proof of \citet{GordonK09}.
\end{proof}

\subsubsection{Combining Full Security and Privacy}\label{sec:apps:iklp}

Given an $n$-party functionality, the ideal model for computing $f$ with $t$-full-privacy is defined in a similar way to the ideal model computing $f$ with abort (\cref{def:ideal_abort}) with the exception that the adversary can invoke the trusted party $t+1$ times, and learn the result of $f$ computed on the same inputs for the honest parties and different inputs for the corrupted parties in each time.
This notion was introduced by \citet{IKKLP11}, who showed that assuming the existence of \TDP and \CRH, for every $n$-party functionality $f$ there exists a single protocol $\pi$ that securely computes $f$ with $t$-full-privacy, tolerating $t<n$ corrupted parties, and achieves $t$-full-security if $t<n/2$; however, the round complexity in~\cite{IKKLP11} is $O(t)$.
We next show how to reduce the round complexity and obtain better privacy guarantees, when a constant number of the parties are honest.

\begin{corollary}[restating \cref{cor:intro_IKLP}]\label{cor:IKLP}
Assume that \TDP, \CRH, and non-interactive perfectly binding commitment schemes exist. Let $f$ be an $n$-party functionality, let $t=\beta n$ for $0<\beta<1$, and let $n'=\omega(\log(\secParam))$.
Then, there exists a single protocol $\pi$, with round complexity $O(n')$, such that:
\begin{enumerate}
    \item
    $\pi$ computes $f$ with $n'$-full-privacy.
    \item
    If $t<n/2$ then $\pi$ computes $f$ with $t$-full-security.
\end{enumerate}
\end{corollary}

\begin{proofsketch}
Consider a variant of the functionality $\fin{f}{n}{n'}$ that instead of outputting the output value $y$ \rchanged{on}{in} the clear, secret shares $y$ using a $\ceil{n'/2}$-out-of-$n'$ ECSS scheme, denoted $(\Share',\Recon')$. Let $\pi'$ be a constant-round protocol computing this variant of $\fin{f}{n}{n'}$ with abort (\eg the protocol from~\cite{Pass04}) and let $\committee\subseteq[n]$ of size $n'$. Consider a variant of the compiler presented in \cref{prot:fair_to_ridfair}, where in case the computation completes without abort, each party locally reconstructs the shares it received from the committee members in $\committee$ and obtains the output value.
By adjusting the proof of \cref{thm:idfair_to_ridfair}, the protocol $\pi=\Comp(\pi',\committee)$ is a constant-round protocol computing $f$ with $\committee$-identifiable-abort and $\floor{(n'-1)/2}$-fairness or $(n'-1)$-abort, \ie if $t'<n'/2$ then $\pi$ is fair and if $n'/2\leq t'\leq n'-1$ it is secure with abort, and in case of abort a corrupted party in $\committee$ is identified.
\end{proofsketch}

\subsubsection{Partially Identifiable Abort to Full Security}\label{sec:apps:ikpsy}

\citet{IKPSY16} introduced security with $\alpha$-partially identifiable abort as security with identifiable abort, such that upon abort, a subset of parties is identified, where at least an $\alpha$-fraction of the subset is corrupted. Next, they presented a transformation, in the honest-majority setting, from $\alpha$-partially identifiable abort, for $\alpha\leq 1/2$, to full security.
The transformation is based on the player-elimination approach, and the idea is to compute $\finout{f^n}{n}{t'}{n'}$ with $\alpha$-partially identifiable abort by a committee that initially consists of all the parties (\ie $n'=n$), where in case of abort, all the identified parties (both honest and corrupted) are removed from the committee. It follows that the number of iterations is $O(n)$.
We next show how to reduce the round complexity when a constant fraction of the parties are honest.

\begin{corollary}[restating \cref{cor:intro_IKPSY}]\label{cor:IKPSY}
Let $f^n$ be an $n$-party functionality, let $n'=\logstar(\secParam)\cdot\log(\secParam)$, let $0<\beta<\beta'<1/2$, let $t=\beta n$, let $t'=\beta' n'$, and let $\pi'$ be an \rchanged{$r$}{$r'$}-round protocol that securely computes $\finout{f^n}{n}{t'}{n'}$ with $\beta'$-partially identifiable abort, tolerating $t'$ corruptions.
Then, $f^n$ can be computed with full security, tolerating $t$ corruptions, by a \rchanged{$O(t'\cdot r)$}{$O(t'\cdot r')$}-round protocol that uses the protocol $\pi'$ in a black-box way.
\end{corollary}

\begin{proofsketch}
The proof follows in similar lines to proof of \cref{thm:ridfair_to_full_withinput,thm:idfair_to_ridfair_hm}, with the exception that upon abort in a committee, all identified parties are removed from the committee. In more detail, the parties initially elect a committee $\committee$ of size $\log(\secParam)\cdot\logstar(\secParam)$. Next, every party secret shares its input to the committee members, that execute the protocol $\pi'$. Upon completion, every committee member broadcasts the output value it received. If the output consists of at most $t'$ messages of the form $(\bot,\cdot)$, every party reconstructs the output value and halts. Otherwise, if more than $t'$ messages are of the form $(\bot,\ID)$, the parties consider the subset $\ID$ that appears in most messages, set $\committee=\committee\setminus \ID$, and re-iterate.
\end{proofsketch}

\subsection{Computation over Shared Inputs}\label{sec:shared_inputs}
\radded{In this section, we provide the definitions of the functionalities $\fin{f}{n}{n'}$ and $\finout{f}{n}{t'}{n'}$ that were used earlier.}

A basic idea underlying our reductions is to delegate the computation to a small committee. Intuitively, this is done by having each party secret share its input value among the committee members, and letting the committee compute the function over the shared inputs. The type of the secret-sharing scheme that we use depends on whether an honest majority is assumed or not. In the honest-majority setting, we consider an error-correcting secret-sharing scheme (ECSS, see \cref{def:ECSS}), allowing honest committee members to reconstruct the correct input values even if corrupted committee members arbitrarily modify the shares they received.

The case that an honest majority cannot be guaranteed is somewhat more subtle. In this case, ECSS schemes do not exist~\cite{IOS12}, and we need to use an $n'$-out-of-$n'$ secret sharing.
Furthermore, to prevent corrupted parties from computing the function on wrong inputs (by changing some of their shares), we let each party broadcast commitments to its shares.
The functionality computed by the committee, parametrized by all the commitments, first verifies that the decommitments are valid and only later reconstructs the input values and evaluates the function.


\vspace{-2ex}
\paragraph{The Verify-Reconstruct-Compute functionality.}
Given a public-output $n$-party functionality $f$ and $n'<n$, the $n'$-party functionality $\fin{f}{n}{n'}$ is parametrized by $n$ vectors of commitments $\vcomval_i=(\comval_i^1,\ldots,\comval_i^{n'})$, for $i\in[n]$.
Each of the commitments $c_i^j$ commits a unique value $s_i^j$; denote $x_i=\bigoplus_{j\in[n']}{s_i^j}$.
Each input value to $\fin{f}{n}{n'}$ consists of a vector of $n$ values, such that the \ith value for the \jth party is the decommitment to $c_i^j$.
The functionality validates that all of the decommitments are valid, reconstructs all $x_i$'s, evaluates $y=f(x_1,\ldots,x_n)$ and outputs the result $y$ to all $n'$ parties. A formal description of the functionality appears in \cref{fig:ssin}.

\begin{nfbox}{The Verify-Reconstruct-Compute functionality}{fig:ssin}
\small
\begin{center}
    \textbf{The functionality} $\fin{f}{n}{n'}$
\end{center}
The $n'$-party functionality $\fin{f}{n}{n'}$ is parametrized by an $n$-party public-output functionality $f$, a non-interactive commitment scheme $\com$ and vectors of commitments $\vcomval_i=(\comval_i^1,\ldots,\comval_i^{n'})$, for every $i\in[n]$.
$\fin{f}{n}{n'}$ is formally defined as follows, on inputs $(\vs_1,\ldots,\vs_{n'})$.
\begin{enumerate}
    \item
    For every $j\in[n']$, parse $\vs_j$ as an $n$-vector $(\hat{s}^j_1,\ldots,\hat{s}^j_n)$;
    for every $i\in[n]$, parse $\hat{s}^j_i$ as a pair $\hat{s}^j_i=(s^j_i,\decomval^j_i)$ and verify that $\comval^j_i=\Com(s^j_i;\decomval^j_i)$.
    If there exists a malformed $\vs_j$, output $(\bot,j)$ to all parties (for the smallest such $j$).
    \item
    For every $i\in[n]$, compute $x_i=\bigoplus_{j\in[n']}{s_i^j}$.
    \item
    Compute $y=f(x_1,\ldots,x_n)$ and output $y$ to all parties.
\end{enumerate}
\end{nfbox}

\vspace{-2ex}
\paragraph{The Reconstruct-Compute-Share functionality}
This functionality is parametrized by a $(t'+1)$-out-of-$n'$ error-correcting secret-sharing scheme (ECSS, \cref{def:ECSS}) that can correct up to $t'$ errors during the reconstruction procedure.
Given an $n$-party functionality $f$, denote by $\finout{f}{n}{t'}{n'}$ the $n'$-party functionality, that receives as input secret shares of an $n$-tuple $(x_1,\ldots,x_n)$, evaluates $f$ on the reconstructed values as $y=f(x_1,\ldots,x_n)$,\footnote{For the sake of simplicity, we consider $f$ as a public-output functionality. Adjusting the protocol for private-output functionalities is straightforward.} and outputs secret shares of $y$. A formal description of the functionality appears in \cref{fig:RecCompShare}.

\begin{nfbox}{The Reconstruct-Compute-Share functionality}{fig:RecCompShare}
\small
\begin{center}
    \textbf{The functionality} $\finout{f}{n}{t'}{n'}$
\end{center}
The $n'$-party functionality $\finout{f}{n}{t'}{n'}$, parametrized by an $n$-party functionality $f$ and a $(t'+1)$-out-of-$n'$ ECSS scheme $(\Share,\Recon)$, is formally defined as follows, on inputs $(\vs_1,\ldots,\vs_{n'})$.
\begin{enumerate}
    \item
    For every $j\in[n']$ parse $\vs_j$ as an $n$-vector $(s^j_1,\ldots,s^j_n)$ (in case $\vs_j$ cannot be parsed, set to the vector $(0,\ldots,0)$; the value $0$ is arbitrary).
    \item
    For every $i\in[n]$, compute $x_i=\Recon(s_i^1, \ldots, s_i^{n'})$ (in case the reconstruction fails, set $x_i=\tilde{x}_i$ for some default value $\tilde{x}_i$).
    \item
    Compute $y=f(x_1,\ldots,x_n)$.
    \item
    Secret share the result $(y_1,\ldots, y_{n'})\gets\Share(y)$.
    \item
    Set the output for every $j\in[n']$ to be $y_j$.
\end{enumerate}
\end{nfbox}

\ifdefined\IsResultWithAbort
\section{Security with Abort to Full Security}\label{sec:abort_to_full}

In this section, we show a black-box reduction from full security to security with abort when sufficiently many parties are honest.
We stress that we do not assume that the underlying functionality also provides identifiability; that is, an adversary may abort the computation after seeing the output, in which case, all honest parties receive a special $\bot$ symbol as output but no malicious party should be identified.
The results in this section are unconditional, \ie we consider computationally unbounded adversaries and that parties can communicate over secure point-to-point channels as well as over a broadcast channel.
Security with fairness and restricted abort is defined in \cref{sec:restricted_fair_def}.
In \cref{sec:fair_to_full_hm} we show how to uplift a fair computation of $f$ with restricted abort to a fully secure computation of $f$, when sufficiently many parties are honest.
In \cref{sec:abort_to_rfair} we show how to reduce a fair computation of $f$ with restricted abort to security with abort of the Reconstruct-Compute-Share variant of a function~$f$ (see \cref{sec:shared_inputs}), when a majority of the parties are honest.

\begin{theorem}[restating \cref{thm:intro:aborttofull}]\label{thm:aborttofull}
Let $f$ be an $n$-party functionality.
\begin{enumerate}
    \item
    Let $t$ be such that $t \cdot (2t+1)<n$ and denote $n'=2t+1$. Then, $f$ can be $t$-securely computed with full security, with statistical security, in the hybrid model computing $\finout{f}{n}{t}{n'}$ with abort, using $t$ invocations of the ideal functionality.
    For $n'=3t+1$ and $t \cdot (3t+1)<n$, the above holds with perfect security.
    \item
    Let $t$ be such that $(t+1) \cdot (2t+1) \leq n$ and denote $n'=2t+1$. Then, $f$ can be $t$-securely computed with full security, with statistical security, in the hybrid model computing $\finout{f}{n}{t}{n'}$ with abort, $t+1$ times in parallel, using a single invocation of the ideal functionality.
    For $n'=3t+1$ and $(t+1) \cdot (3t+1)<n$, the above holds with perfect security.
\end{enumerate}
\end{theorem}

\noindent
The proof of \cref{thm:aborttofull} follows from \cref{thm:rfair_to_full,thm:abort_to_restrict_fair} appearing in \cref{sec:abort_to_rfair,sec:fair_to_full_hm} (respectively).

\subsection{Fairness with Restricted Abort}\label{sec:restricted_fair_def}

We start by modifying the security definition of \emph{fairness with restricted identifiable abort} (\cref{def:ideal_ridabort}), such that in case the adversary aborts the computation, no corrupted party is identified.

\paragraph{Ideal model with fairness and restricted abort.}
An ideal computation, with fairness and $\vCS$-abort, of an $n$-party functionality $f$ on input $\vx=(x_1,\ldots,x_n)$ for parties $(\Party_1,\ldots,\Party_n)$ \wrt $\vCS=(\committee_1,\ldots,\committee_\ell)$, where $\committee_1,\ldots,\committee_\ell\subseteq [n]$, in the presence of an ideal-model adversary $\Adv$ controlling the parties indexed by $\IS\subseteq[n]$, proceeds via the following steps.
\begin{itemize}
    \item[\emph{Sending inputs to trusted party}:]
    An honest party $\Party_i$ sends its input $x_i$ to the trusted party.
    The adversary may send to the trusted party arbitrary inputs for the corrupted parties. Let $x_i'$ be the value actually sent as the input of party $\Party_i$.

	\item[\emph{Early abort}:]
    If there exists a corrupted party in every subset $\committee_j$, \ie if $\IS\cap\committee_j\neq \emptyset$ for every $j\in[\ell]$, then the adversary $\Adv$ can abort the computation by sending the \abort message to the trusted party. In case of such abort, the trusted party sends the message $\bot$ to all parties and halts.

    \item[\emph{Trusted party answers the parties}:]
    If $\committee_j\subseteq\IS$ for some $j\in[\ell]$, the trusted party sends all the input values $x'_1,\ldots,x'_n$ to the adversary, waits to receive from the adversary output values $y'_1,\ldots,y'_n$, sends $y'_i$ to $\Party_i$ and proceeds to the \emph{Outputs} step.
    Otherwise, the trusted party computes $(y_1, \ldots, y_n)=f(x_1', \ldots, x_n')$ and sends $y_i$ to party $\Party_i$ for every $i\in[n]$.

    \item[\emph{Outputs}:]
    Honest parties always output the message received from the trusted party and the corrupted parties output nothing.
    The adversary $\Adv$ outputs an arbitrary function of the initial inputs $\set{x_i}_{i\in\IS}$, the messages received by the corrupted parties from the trusted party and its auxiliary input.
\end{itemize}

\begin{definition}[ideal-model computation with fairness and restricted abort]\label{def:ideal_rfair}
Let $f\colon(\zs)^n \mapsto (\zs)^n$ be an $n$-party functionality, let $\IS\subseteq [n]$, and let $\vCS=(\committee_1,\ldots,\committee_\ell)$, where $\committee_1,\ldots,\committee_\ell\subseteq [n]$.
The {\sf joint execution of $f$ with $\vCS$ under $(\Adv, I)$ in the ideal model}, on input vector $\vx=(x_1, \ldots, x_n)$, auxiliary input $\aux$ to $\Adv$, and security parameter $\secParam$, denoted $\IDEAL^{\vCS\mhyphen\abort}_{f,\IS,\Adv(\aux)}(\vx,\secParam)$, is defined as the output vector of $\Party_1, \ldots, \Party_n$, and $\Adv(\aux)$ resulting from the above described ideal process.
\end{definition}

\paragraph{Security definition.}

\begin{definition}\label{def:SecureProtocol_rfair}
Let $f\colon(\zs)^n \mapsto (\zs)^n$ be an $n$-party functionality and let $\pi$ be a \ppt protocol computing $f$.
The {\sf protocol $\pi$ $(\delta,t)$-securely computes $f$ with fairness and $(\ell, n',t')$-abort (and
statistical security)}, if for every real-model adversary \Adv, there exists an adversary $\Sim$ for the ideal model, whose running time is polynomial in the running time of $\Adv$, such that for every $\IS\subseteq [n]$ of size at most $t$, and subsets $\vCS=(\committee_1,\ldots,\committee_\ell)$, where for every $j\in[\ell]$, $\committee_j\subseteq [n]$ satisfies $\ssize{\committee_j}=n'$ and $\ssize{\IS\cap\committee_j}\leq t'$, it holds that
\[
\set{\bigbrack \REAL_{\pi, \IS, \Adv(\aux)}(\vx, \secParam)}_{(\vx, \aux)\in(\zs)^{n+1}, \secParam\in\N}
\deltaclose
\set{\bigbrack \IDEAL^{\vCS\mhyphen\fair}_{f, \IS, \Sim(\aux)}(\vx, \secParam)}_{(\vx, \aux)\in(\zs)^{n+1}, \secParam\in\N}.
\]
If $\delta$ is negligible, we say that $\pi$ is a protocol that $t$-securely computes $f$ with fairness and $(\ell, n',t')$-abort and statistical security.

Similarly, {\sf $\pi$ is a protocol that $(\delta,t)$-securely computes $f$ with fairness and $(\ell, n',t')$-abort (and perfect security)}, if
\[
\set{\bigbrack \REAL_{\pi, \IS, \Adv(\aux)}(\vx, \secParam)}_{(\vx, \aux)\in(\zs)^{n+1}, \secParam\in\N}
\deltaequiv
\set{\bigbrack \IDEAL^{\vCS\mhyphen\fair}_{f, \IS, \Sim(\aux)}(\vx, \secParam)}_{(\vx, \aux)\in(\zs)^{n+1}, \secParam\in\N}.
\]
If $\delta=0$, we say that $\pi$ is a protocol that $t$-securely computes $f$ with fairness and $(\ell, n',t')$-abort and perfect security.
\end{definition}

We denote fairness with $(1,n',t')$-abort by fairness with $(n',t')$-abort.

\subsection{Fairness with Restricted Abort to Full Security}\label{sec:fair_to_full_hm}

Having defined fairness with restricted abort, we show how to compute $f$ with full security in the hybrid model computing $f$ with fairness and restricted abort (recall that no corrupted parties are identified in case of abort). These reductions apply when sufficiently many parties are honest, roughly $t=O(\sqrt{n})$, and are valid for any number of parties. The main idea is to partition a subset of the parties into $t+1$ committees, where each committee has an honest majority. It is guaranteed in this situation that at least one committee is fully honest, hence, the adversary cannot abort all $t+1$ computations.
\begin{theorem}\label{thm:rfair_to_full}
Let $f$ be an $n$-party functionality.
\begin{enumerate}
    \item
    Let $t$ be such that $t \cdot (2t+1)<n$ and consider a hybrid model that securely computes $f$ with
    fairness and $(2t+1,t)$-abort. Then, $f$ can be
    $t$-securely computed with full security, using $t$ sequential calls to the ideal functionality.
    \item
    Let $t$ be such that $(t+1) \cdot (2t+1) \leq n$ and consider a hybrid model that securely computes $f$ with
    fairness and $(t+1,2t+1,t)$-abort. Then, $f$ can be
    $t$-securely computed with full security, using a single call to the ideal functionality.
\end{enumerate}
\end{theorem}

\begin{proof}
We start by proving the first part of the theorem.
Since $t\cdot(2t+1)<n$, the first $t\cdot(2t+1)$ parties can be partitioned into $t$ disjoint committees $\committee_1,\ldots,\committee_t$ each of size $2t+1$, by setting for every $j\in[t]$
\[
\committee_j=\set{\Party_{(j-1)(2t+1)+1}, \ldots, \Party_{j\cdot(2t+1)}}.
\]
In each of the committees $\committee_1,\ldots,\committee_t$ it is guaranteed that a majority of the parties are honest.
The idea now is to sequentially compute $f$, where the \jth computation is fair with $\committee_j$-abort.
In case all $t$ computations abort, it is guaranteed that party $\Party_{t(2t+1)+1}$ is honest, and can complete the computation on its own.

As for the second part of the theorem, since $(t+1)\cdot(2t+1)\leq n$, the first $(t+1)\cdot(2t+1)$ parties can be partitioned into disjoint committees $\committee_1,\ldots,\committee_{t+1}$ each of size $2t+1$, by setting for every $j\in[t+1]$ the set
$
\committee_j=\sset{\Party_{(j-1)(2t+1)+1}, \ldots, \Party_{j\cdot(2t+1)}}.
$
As before, in each of the committees $\committee_1,\ldots,\committee_{t+1}$ it is guaranteed that a majority of the parties are honest.
Each party $\Party_i$ sends its input $x_i$ to the trusted party computing $f$ with
fairness and $\vCS$-abort, receives back a value $y$, outputs $y$ and halts.
Since $\Adv$ controls at most $t$ parties and there are $t+1$ committees, it is guaranteed that honest parties will receive valid output value.
The simulation is straightforward. \Sim receives input values $\sset{x'_i}_{i\in\IS}$ from \Adv (and sets missing values to default), sends them to its trusted party, receives back an output value $y$, forwards $y$ to \Adv and outputs whatever \Adv outputs. It is immediate that the simulation is perfect.

\end{proof}

\subsection{Security with Abort to Fairness with Restricted Abort}\label{sec:abort_to_rfair}

At the basis of our reduction is the $n'$-party functionality Reconstruct-Compute-Share $\finout{f}{n}{t'}{n'}$ that receives shared input values using a $(t'+1)$-out-of-$n'$ ECSS scheme (\cref{def:ECSS}), computes the function $f$ over the reconstructed inputs and outputs shares of the result. We show that, in the honest-majority setting, if $\finout{f}{n}{t'}{n'}$ can be computed with abort $\ell$ times in parallel by subsets $\vCS=(\committee_1,\ldots,\committee_\ell)$ of $n'$ parties, where each $\committee_l$ has at most $t'<n'/2$ corrupted parties, then $f$ can be computed with fairness and $\vCS$-abort.

\begin{theorem}\label{thm:abort_to_restrict_fair}
Let $f$ be an $n$-party functionality, let $t',n'\in\N$ be such that $2t'<n'<n$, let $\ell=\poly(n)$, and let $\vCS=(\committee_1,\ldots,\committee_\ell)$, where $\committee_1,\ldots,\committee_\ell\subseteq[n]$ be subsets of size $n'$.
Then, $f$ can be $t$-securely computed with fairness and $\vCS$-abort and statistical security in a hybrid model that securely computes $\finout{f}{n}{t'}{n'}$ with abort, $\ell$ times in parallel, whenever $t<n$ and the number of corrupted parties in every $\committee_i$ is at most $t'$. Furthermore, the above holds with perfect security whenever $t'<n'/3$.
\end{theorem}

\begin{proof}
The protocol in the $\finout{f}{n}{t'}{n'}$-hybrid model is defined as follows.
\begin{protocol}(security with abort to fairness with restricted abort)\label{prot:abort_to_rfair}
\begin{itemize}
    \item\textbf{Hybrid Model:}
    The protocol is defined in the hybrid model computing $\finout{f}{n}{t'}{n'}$ with $\vCS$-abort.
    \item\textbf{Common Input:}
    A $(t'+1)$-out-of-$n'$ ECSS scheme $(\Share,\Recon)$, and $\vCS=(\committee_1,\ldots,\committee_\ell)$. We use the notation $\Party^l_j$ to refer to the \jth party in $\committee_l$.
    \item\textbf{Private Input:}
    Every party $\Party_i$, for $i\in[n]$, has private input $x_i\in\zs$.
    \item\textbf{The Protocol:}
\end{itemize}
\begin{enumerate}
    \item\label{step:abort_secret_share}
    For every $l\in[\ell]$, every party $\Party_i$ secret shares its input $(s_i^{1,l},\ldots,s_i^{n',l})\gets\Share(x_i)$ and sends $s_i^{j,l}$ to the \jth party in $\committee_l$ (over a private channel).

    \item\label{step:abort_ideal_functionality}
    Every party $\Party^l_j\in\committee_l$, sends $(s_1^{j,l},\ldots,s_n^{j,l})$ to the ideal functionality computing $\finout{f}{n}{t'}{n'}$ (in case $\Party^l_j$ did not receive a value from $\Party_i$ it sets $s_i^{j,l}=\lambda$).
    In return, $\Party^l_j$ receives $y_{j,l}$ as its output value. If $y_{j,l}=\bot$, then $\Party^l_j$ broadcasts $(l,\bot)$.

    \item\label{step:abort_output}
    If more than $t'$ values $(l,\bot)$ were broadcasted for every $\committee_l$, then all parties output $\bot$ and halt. Otherwise, let $\ls$ be the smallest $l\in[\ell]$ such that no more than $t'$ values $(l,\bot)$ were broadcasted.
    Every party $\Party^{\ls}_j\in\committee_{\ls}$ broadcasts $y_{j,\ls}$.

    \item\label{step:abort_output_recon}
    Every party $\Party_i$ computes $y=\Recon(y_{1,\ls},\ldots,y_{n',\ls})$ (where in case $\Party^{\ls}_j\in\committee_{\ls}$ did not send a value, $\Party_i$ sets $y_{j,\ls}=\lambda$) and outputs $y$.
\end{enumerate}
\end{protocol}

Let $\Adv$ be a computationally unbounded adversary attacking the execution of \cref{prot:abort_to_rfair}, denoted as $\pi$, in the $\finout{f}{n}{t'}{n'}$-hybrid model, and let $\IS\subseteq[n]$ be a subset of size at most $t$, satisfying $\ssize{\IS\cap\committee_l}\leq t'$ for every $l\in[\ell]$.
We construct an ideal-model adversary $\Sim$ (simulator) computing $f$ with fairness and $\vCS$-abort.
On inputs $\set{x_i}_{i\in\IS}$ and auxiliary input $\aux$, the simulator $\Sim$ starts by emulating $\Adv$ on these inputs. $\Sim$ interacts with $\Adv$, playing the roles of the honest parties and the ideal functionality $\finout{f}{n}{t'}{n'}$. For simplicity, assume that all input and output values are elements in $\zo^\secParam$.

To simulate Step~\ref{step:abort_secret_share}, for every $l\in[\ell]$, the simulator \Sim generates shares of zero for every honest party $\Party_i$ as $(s_i^{1,l},\ldots,s_i^{n',l})\gets\Share(0^\secParam)$ and sends $s_i^{j,l}$ to \Adv for every corrupted $\TParty^l_j\in\committee_l$; in addition, $\Sim$ receives from $\Adv$ a value $s_i^{j,l}$ on behalf of every corrupted party $\TParty_i$ that is sent for an honest party $\Party^l_j\in\committee_l$.
To simulate Step~\ref{step:abort_ideal_functionality}, for every committee $\committee_l$ that contains corrupted parties, $\Sim$ simulates the ideal functionality computing $\finout{f}{n}{t'}{n'}$ by receiving from \Adv an input value $(\tilde{s}_1^{j,l},\ldots,\tilde{s}_n^{j,l})$ for every corrupted $\TParty^l_j\in\committee_l$. In case \Adv sends an early \abort\ message, the simulator \Sim responds with $\bot$ message to all corrupted parties; otherwise, \Sim computes $(y_{1,l},\ldots,y_{n',l})\gets\Share(0^\secParam)$ and responds with $y_{j,l}$ for every $\Party_j\in\committee_l\cap\IS$ . In the latter case, \Sim waits to receive \abort or \continue from \Adv.
In case \Adv sent $\abort$ for the computation corresponding to $\committee_l$ (either an early abort or a late abort), \Sim simulates broadcasting $(l,\bot)$ by every honest party in $\committee_l$; in addition, \Sim receives messages $(l,\bot)$ from \Adv.

In case \Adv sent an \abort\ message for every ideal computation of $\finout{f}{n}{t'}{n'}$ by a committee $\committee_l$, the simulator \Sim sends an \abort message to its trusted party, outputs whatever \Adv outputs and halts. Otherwise, let $\ls$ be the smallest $l\in[\ell]$ for which the adversary \Adv did not send an \abort message to \Sim (\ie to the ideal computation of $\finout{f}{n}{t'}{n'}$ for the committee $\committee_l$).

For every $i\in\IS$, the simulator \Sim computes $x'_i=\Recon(\tilde{s}_i^{1,\ls},\ldots,\tilde{s}_i^{n',\ls})$, where for $j\in\committee_{\ls}\setminus\IS$, the value $\tilde{s}_i^{j,\ls}=s_i^{j,\ls}$ is obtained during the simulation of Step~\ref{step:abort_secret_share} and for $j\notin\committee_{\ls}\cap\IS$, the value $\tilde{s}_i^{j,\ls}$ is obtained during the simulation of Step~\ref{step:abort_ideal_functionality} (in case some value was not sent by $\Adv$, set it to $\lambda$). If the reconstruction of $x'_i$ failed, set it to be the default value. \Sim sends $\sset{x'_i}_{i\in\IS}$ to the trusted party computing $f$ and receives back the output $y$. Finally, $\Sim$ computes $(y_1,\ldots,y_{n'})\gets\Share(y)$ and broadcasts $y_j$ on behalf of every $\Party^{\ls}_j\in\committee_{\ls}\setminus\IS$, receives from \Adv values for parties $\Party^{\ls}_j\in\committee_{\ls}\cap\IS$, outputs whatever \Adv outputs and halts.

It is straightforward to prove that, following the properties of the ECSS scheme, the output of the honest parties and \Adv in the real execution of protocol $\pi$ and the output of the honest parties and \Sim in the ideal computation of $f$ with fairness and $\vCS$-abort are statistically 
close for $t'<n'/2$ and identically distributed
for $t'<n'/3$.
\end{proof}

\fi
\newcommand{\idabortPar}{\MathAlgX{id \mhyphen abort\mhyphen par}}
\newcommand{\psiS}{\psi}
\newcommand{\psiSJ}{\psi}

\newcommand{\AdvPsi}{\Adv^\psi}
\newcommand{\hAdvPsi}{\widehat{\Adv}^\psi}

\section{Necessity of Super-Constant Sequential Fair Calls}\label{sec:impossibility}
In this section, we prove that a super-constant number of \emph{functionality rounds}, in which (possibly parallel) calls to the fair functionality are made, is necessary for a fully secure implementation of some functionalities. Specifically, we show the necessity of such number of functionality rounds when the functionality in consideration is coin flipping.

The model in which the lower bound is proved is defined in \cref{sec:imposability:model}, and the lower bound is stated and proved in \cref{sec:Impo:Proof}. The lower bound proof uses a useful corollary of Cleve's lower bound~\cite{Cleve86}, whose proof is given in \cref{sec:impo:Cleve}.

\subsection{The Model}\label{sec:imposability:model}
Consider an $n$-party coin-flipping protocol $\pi$ executed in the hybrid model where parties can compute $\fcf{n}$ ($n$-party coin-flipping functionality, see \cref{def:coinflip}) with fairness and $\committee\mhyphen\idabort$ for every $\committee\subseteq[n]$ (see \cref{def:ideal_ridabort_noinput}).
Namely, a protocol in this model has in addition to the standard communication rounds, also \emph{functionality rounds} in which any subset (committee) $\committee\subseteq[n]$ of the parties can ask the trusted party to flip a coin, and the trusted party outputs a uniform bit visible to all $n$ parties. If the subset $\committee$ contains a corrupted party, then the attacker can abort this call before learning the value of the output bit, but at the price of revealing the identity of a corrupted party in the committee to all parties (even those not in the committee). An all-corrupted committee can determine the output of the trusted party arbitrarily, without being identified.

We prove the lower bound in a stronger hybrid model than the above (where the life of the honest parties are easier) that allows \emph{parallel} calls to the trusted party by different committees at the same functionality round. In such a case, (only) an all-corrupted committee is assumed to be rushing, and can decide whether to abort or not based on the outcome of other parallel calls. The output of all other committees, is published at the same time to all parties. We denote the above hybrid model by \emph{$(\fcf{n}, [n]$-$\idabortPar)$-hybrid model}. An \emph{$\numcalls$-call} protocol in this model has at most $\numcalls$ functionality rounds.

Recall that in a coin-flipping protocol, the output of an all-honest execution is a common uniform bit. The protocol is \emph{$\consistent$-consistent} if in an honest execution, any two parties output the same value with probability at least $1/2 + \consistent$ (in the case that $\consistent=1/2-\negl(\secParam)$, we simply say the protocol is consistent).
Finally, a fail-stop attacker for $\pi$ might only deviate from the protocol by early aborting (in the $(\fcf{n}, [n]$-$\idabortPar)$-hybrid model, this can be done either by stop sending messages during an interaction round, or by aborting a call to the ideal functionality).

\begin{remark}[Is this the right model?] The above model is generous with the honest parties, as it assumes that \emph{only} an all-corrupted committee can be rushing: only a party in such a committee can decide to abort or not based on the result of the other parallel calls to the trusted party. Indeed, in this model, the reductions of \cref{sec:fairtofull_noinput_noHM} can be modified, in the spirit of the reductions given in \cref{sec:fairtofull_noinput_HM}, to yield protocols of only $\omega(1)$ sequential calls to the fair functionality.

We do not know whether the above model can be justified by the existence of a fully secure protocol that computes the functionality at hand (at least not for a randomized functionality in the dishonest-majority setting). Hence, the model we actually use to prove our positive results in \cref{sec:fairtofull_noinput_noHM} is much more pessimistic, and essentially does not allow parallel invocations of the trusted party. In this pessimistic model, it is not hard to prove that a super-logarithmic number of functionality rounds is needed, making the result of \cref{sec:fairtofull_noinput_noHM} optimal in this respect.
\end{remark}

\vspace{-1ex}
\subsection{The Lower Bound}\label{sec:Impo:Proof}
Given the above formulation, our lower bound is stated as follows.
\begin{theorem}\label{thm:lowerbound}
Let $\pi$ be an $n$-party, $\numcalls$-call, $m$-time, $\consistent$-consistent coin-flipping protocol in the $(\fcf{n}, [n]$-$\idabortPar)$-hybrid model. Then, there exists a fail-stop attacker, controlling a $\beta n$-size subset of the parties, for $1/2<\beta<1$, that can bias the output of $\pi$ by $\Omega(\consistent n^{-c}/m)$, for $c = O(\numcalls \cdot \log m / \log n)$.
\end{theorem}

The above yields that for a constant $\numcalls$ and non-negligible $\consistent$, protocol $\pi$ is not fully secure facing up to $\beta n$ corruptions. For simplicity, we will prove the theorem for consistent protocols (\ie $\consistent = 1/2$), but our technique readily captures any non-negligible $\consistent$.

The proof of \cref{thm:lowerbound} follows the high-level description given in the Introduction (\cref{sec:intro:technique}). In \cref{sec:small_committees}, we show how to bias any coin-flipping protocol, assuming the ideal functionality is \emph{not} invoked by ``large'' committees (all committees are of size at most $c \cdot \log n$, for a constant $c$ to be determined by the analysis). In \cref{sec:large_committees}, we adjust the proof of \cref{sec:small_committees} for proving \cref{thm:lowerbound}, by showing that large committees are not useful for protocols of constant number of functionality rounds.

\vspace{-1ex}
\subsubsection{Biasing Protocols with No Large Committees}\label{sec:small_committees}

In this section, we prove \cref{thm:lowerbound} assuming the ideal functionality is only invoked by ``small'' committees. This is captured by the following lemma.

\begin{lemma}\label{Lem:lowerbound_small_committees}
Let $\pi$ be an $n$-party, $\numcalls$-call, $m$-time, consistent coin-flipping protocol in the $(\fcf{n}, [n]$-$\idabortPar)$-hybrid model, in which all calls to the ideal functionality are made by committees of size at most $c\cdot \log n$ for a constant $c$. Then, there exists a fail-stop attacker, controlling a $\beta n$-size subset of the parties, for $1/2<\beta<1$, that can bias the output of $\pi$ by $\Omega(n^{-c}/m)$.
\end{lemma}


The proof is a reduction of the multiparty coin-flipping protocol $\pi$ in the $(\fcf{n}, [n]$-$\idabortPar)$-hybrid model to a two-party coin-flipping protocol in the plain model. Since the latter can be biased by a fail-stop attacker~\cite{Cleve86}, the proof is completed by showing how to translate any such fail-stop attacker into a (fail-stop) attacker on $\pi$. Interestingly, the proof of this part holds for any (no-large-committees) protocol, regardless of the number of its functionality rounds. The following discussion is \wrt a fixed $(1 - \beta)n$-size subset $\cs\subseteq [n]$ (chosen arbitrarily).

Loosely speaking, the two-party protocol $\psiS = (\Party_0,\Party_1)$ relates to the $n$-party protocol $\pi$ as follows: the $n$ parties of $\pi$ are partitioned into the two subsets
\remove{of sizes $(1 - \beta) n$ and $\beta n$}
$\cs$ and $[n]\setminus \cs$, that are controlled by $\Party_0$ and $\Party_1$, respectively, in a random emulation of $\pi$. Party $\Party_0$ emulates the ideal functionality: when $\numcomm$ committees, $\committee_1, \ldots, \committee_\numcomm$, call the ideal functionality in parallel in $\pi$, this translates in $\psiS$ to $\Party_0$ sending $\numcomm$ uniformly distributed bits, sequentially, one in each round. These bits are interpreted by the parties as the answers of the ideal functionality. If, as a concrete example, $\Party_0$ aborts in a round in which it is supposed to send a uniform bit, emulating the answer of the ideal functionality to a call by some committee {$\committee$, party $\Party_1$ interprets it as if the first party in $\committee$ aborted the call and continues the execution of $\pi$ as follows: it randomly chooses some $(1-\beta) n$-size subset of the parties under its control that does not intersect $\committee$, denoted $\cT$, and emulates the remainder of the execution as if all parties outside $\cT$ abort right after the call to the trusted party. For technical reasons, it would have been convenient to assume that $\committee$ has \emph{exactly} $ c \log n$ elements outside $\cs$. Since the latter might not always be the case, when choosing the subset $\cT$, party $\Party_1$ first appends arbitrary elements to $\committee$ so that the resulting subset has the requited property.

We now formally define the two-party protocol $\psi$ induced by the multiparty protocol $\pi$. For $\committee \subseteq [n]$, let $\append(\committee)$ be the committee $\committee'$ obtained from $\committee$ by padding it with arbitrary (but fixed) parties from $[n] \setminus \cs$, such that $\size{\committee' \setminus \cs} = c \log n$ (if $\size{\committee \setminus \cs} = c \log n$, then $\committee' = \committee$). The two-party coin-flipping protocol in the standard model $\psiS = (\Party_0, \Party_1)$ is defined as follows.
{\samepage{
\begin{protocol}[Protocol $\psiS= (\Party_0,\Party_1)$: two-party coin flipping from multiparty coin flipping]\label{prot:two_party}~
\begin{enumerate}
	\item
    Party $\Party_0$ controls the parties in $\I_0 = \cs$ and $\Party_1$ controls the parties in $\I_1 = [n] \setminus \cs$ in an emulation of $\pi$: in each communication round of~$\pi$, party $\Party_i$ (for $i \in \zo$) internally emulates all messages among parties in $\I_i$. Party $\Party_i$ also sends to $\Party_{1-i}$ all messages from parties in $\I_i$ to parties in $\I_{1-i}$ and over the broadcast channel.
	
	\item
    Party $\Party_0$ emulates the ideal functionality: when a parallel call to the functionality by $\numcomm$ committees, denoted by $\committee_1,\ldots, \committee_\numcomm$,\footnote{The order of the committees is set arbitrarily.} is due, $\Party_0$ sends $\numcomm$ separate messages, one for each committee, each containing a uniformly distributed bit.
\end{enumerate}

\textbf{Output:}
Party $\Party_0$ (resp., $\Party_1$) outputs the output value of an arbitrary non-aborted party in $\I_0$ (\resp $\I_1$) in $\pi$.

\vspace{1ex}
\textbf{Dealing with abort:}
\begin{enumerate}
	\item
    \emph{Party $\Party_i$ aborts while emulating a communication round:} party $\Party_{1-i}$ continues the emulation internally, as if all the parties in $\I_i$ stopped sending messages in the current round of $\pi$.
    If the aborting party is $\Party_0$, then $\Party_1$ continues the execution by emulating all remaining calls to the ideal functionality on its own.
	
	\item
    \emph{Party $\Party_0$ aborts while emulating a functionality round:} let $\committee_1,\ldots, \committee_\numcomm$ be the committees of the current functionality round, and let $\committee_j$ be the committee corresponding to the aborting round in $\psiS$. Party $\Party_1$ continues the emulation as if the party with the minimal index in $\committee_j$
    aborted the call to the ideal functionality (in $\pi$), and flips the remaining coins for $\committee_{j+1},\ldots,\committee_\numcomm$ internally.
    Next, $\Party_1$ decides on its output as follows: $\Party_1$ chooses uniformly at random a $(1-\beta)n$-size subset $\cT \subseteq [n]\setminus \left(\cs \cup \append(\committee_j) \right)$, and continues the emulation of $\pi$ internally as if all remaining parties outside $\cT$ stopped sending messages.
\end{enumerate}
\end{protocol}
}}
In the analysis of the attack on $\pi$ given below, it will be helpful to think of the subset of corrupted parties as a randomly chosen $\beta n$-size subset containing $\cs$, and then bound the probability that it strictly contains $\append(\committee)$ (where $\committee$ is the aborted committee). The set $\cT$ will than play the role of the honest parties, allowing us to easily relate the expected bias (over the choice of $\cT$) of our attacker for $\pi$ to that of an attacker controlling $\Party_0$ in $\psiS$.


We prove the lemma by translating the fail-stop attacker for $\psiS$ guaranteed by~\cite{Cleve86} (see \cref{lem:simple_advM}) into a fail-stop attacker for $\pi$. The fail-stop attacker for the two-party protocol naturally translates into an attack on $\pi$, if the aborting round in not a functionality round (simply aborts the parties controlled by the corrupted party). For the case that the aborting round is a functionality round, we exploit the fact that, by assumption, all committees are small. Assume the guaranteed attacker for $\psiS$ corrupts party $\Party_0$ and attacks by aborting in a round in which $\Party_0$ is due to emulate the answer of the ideal functionality to some committee $\comStar$.
In order to mimic such an attack in $\pi$, the attacker needs to control all parties in $\comStar$. This is indeed the case with noticeable probability when it chooses the set $\cT$ at random.

Analyzing the above is rather subtle, and crucially relies on the fact that $\cT$ is randomly chosen by $\Party_1$ {\em during} the execution of $\psiS$. This way, the aborting round in $\psiS$ (which is pre-determined and independent of a particular execution of the protocol) is fixed regardless of the choice of $\cT$. Hence, the identity of the parties in $\comStar$ is independent of the choice of $\cT$ in $\psiS$, and therefore, independent of the choice of the corrupted parties (outside $\cs$) in $\pi$. In case we would have fixed $\cT$ in advance, and let $\Party_0$ control all parties in $[n] \setminus \cT$, it might have been the case that for any choice of $\cT$, the aborting round in $\psiS$ would have corresponded to a committee $\comStar \subseteq \cT$. In this case, an attacker that controls $\Party_0$ and aborts during a functionality round cannot be
translated to an attacker on $\pi$.

When formalizing the above, we use a variant of the two-party attacker guaranteed by \citet{Cleve86}. Specifically, we assume the fail-stop attacker either aborts in a predetermined round, or does not abort at all. On the other hand, Cleve's attacker either aborts in a given predetermined round, or aborts in the next round. The existence of the required attacker is stated in the following lemma, whose proof (see \cref{sec:impo:Cleve}) easily follows the result of \citet{Cleve86}.

\begin{lemma}\label{lem:simple_advM}
Let $\psi = \left(\Party_0,\Party_1\right)$ be a two-party, $r$-round, $\consistent$-consistent coin-flipping protocol. Then, there exist a party $\Party \in \left\{\Party_0,\Party_1\right\}$, a round index $\is \in [r]$, an index $\js \in\{\is,\is-1\}$, and a bit $b \in \zo$ such that the following holds. Let $b_{\js}$ denote the output of $\Party$ in case the other party aborts in round $\js$, and consider the fail-stop attacker $\Adv$, that corrupts $\Party$ and instructs $\Party$ to abort in round $\is$ if $b_{\js} = b$. Then, $\Adv$ biases the output of the other party by $\consistent/(8r + 2)$.
\end{lemma}

Let $\Advpsi$ be the fail-stop attacker for $\psiS$ guaranteed by \cref{lem:simple_advM}. Since $\psiS$ consists of at most $m$ communication rounds (recall that $m$ is a bound on the running time of $\pi$),
by \cref{lem:simple_advM}, $\Advpsi$ biases the output of $\psiS$ by at least $1/(16m + 4)$. Observe that if $\Advpsi$ corrupts $\Party_1$, an attacker for $\pi$ controlling $\I_1 = [n]\setminus \cs$ can
bias the output of the honest parties by the same bias, by simply mimicking $\Advpsi$ (acting honestly until some predetermined round $i$, and then aborting with all of the parties in $\I_1$ only if $\Advpsi$ aborted in the corresponding round). Similarly, if $\Advpsi$ corrupts $\Party_0$ and attacks by aborting in a round that corresponds to an interaction round of $\pi$, then a similar attack is also possible. Hence, we assume \wlg that $\Advpsi$ attacks by corrupting $\Party_0$, and aborting, if at all, in a round of $\psiS$ that corresponds to a functionality round in $\pi$.

Let $\is$ be the round in $\psiS$ in which $\Advpsi$ might abort, let $\comStar$ be the committee in $\pi$ that corresponds to round $\is$ in $\psiS$, and let $\is_f$ be the index of the functionality round in $\pi$ in which $\comStar$ calls the ideal functionality.
For a $(1 - \beta) n$-size subset $\cT \subseteq [n]\setminus \cs$, the attacker $\AdvpiI$, corrupting the parties in $[n] \setminus \cT$, is defined as follows.

\begin{algorithm}[The fail-stop attacker $\AdvpiI$ for $\pi$]\label{adv:small_committeed}	~
\begin{enumerate}
	\item
    Start emulating $\Advpsi$, by forwarding to $\Advpsi$ the messages sent by the honest party in $\psi$ and the results of the calls to the ideal functionality.\footnote{This is done to ensure the coins that $\Party_0$ flips in $\psi$ are the same as the coins flipped by the ideal functionality in the execution of $\pi$. Note that this does not change the distribution over the flipped bits.}
	
	\item
    Until the $\is_f$'th functionality round, act honestly.

	\item
    In the $\is_f$'th functionality round: 
		\begin{enumerate}
		\item
        If $\append(\comStar) \subseteq [n] \setminus \cT$: wait until all other committees but $\comStar$ in the $\is_f$'th functionality round are done. {Let $\committee_1,\ldots,\committee_\numcomm$ be the committees that call the ideal functionality in the $\is_f$'th functionality round of $\pi$, and correspond to rounds earlier than $\is$ in $\psiS$.} Give the result bits of $\committee_1,\ldots,\committee_\numcomm$ to $\Advpsi$, sequentially. If $\Advpsi$ chooses to abort, instruct the minimal-index party in $\comStar$ to abort the call to the ideal functionality,\footnote{Note that this is indeed possible in the case in hand; since $\comStar$ is all-corrupted, and we assumed $\AdvpiI$ to be rushing. Thus, it can give $\Advpsi$ all results to calls made by $\committee_1,\ldots,\committee_\numcomm$, and only then choose to abort the call by $\comStar$.} and in the next interaction round of $\pi$, instruct all parties in $ [n] \setminus \cT$ to abort. If $\Advpsi$ does not abort, complete the honest execution of $\pi$.
		
		\item
		If $\append(\comStar) \not\subseteq [n] \setminus \cT$: continue in the honest execution of $\pi$.
	\end{enumerate}
\end{enumerate}	
\end{algorithm}

Namely, the fail-stop attacker $\AdvpiI$, corrupting the parties in $[n] \setminus \cT$, uses the real interaction of $\pi$ to emulate an interaction of $(\Advpsi,\Party_1)$. The attacker $\AdvpiI$ acts honestly until (the emulated) $\Advpsi$ chooses to abort, which, by assumption, happens in a functionality round. Recall that in the emulation of $\pi$ done in $\psiS$, such an abort is interpreted as the first member of the corresponding committee aborting the call to the ideal functionality, followed by all parties outside a random $(1 - \beta)n$-size subset of $ [n]\setminus (\cs \cup \append(\comStar))$, aborting after the call.
The above attack goes through, only if all parties in $\append(\comStar)$ are corrupted. The key observation is that for a randomly chosen $\cT$, with noticeable probability, the set $ [n] \setminus \cT$ contains this (small, by assumption) committee. If the latter happens, then $\AdvpiI$ perfectly replicates the attack by $\Advpsi$, and thus incurs the same bias in $\pi$ as $\Advpsi$ achieves in $\psiS$. Hence, there exists some $\cT$ such that $\AdvpiI$ succeeds in biasing the protocol.

The above intuition is made rigorous in the following lemma that concludes the proof of \cref{thm:lowerbound} for the no-large-committees case.

\begin{lemma}\label{lem:small_committees_attack}
There exists a $(1-\beta) n$-size subset $\cT \subseteq [n] {\setminus \cs}$ such that $\AdvpiI$ biases the output of the honest parties in $\pi$ by $\Omega(n^{-c}/m)$.
\end{lemma}

\begin{proof}
Recall our notation: let $\is$ be the round in $\psiS$ in which $\Advpsi$ might abort, let $\comStar$ be the committee in $\pi$ that corresponds to round $\is$ in $\psiS$, and let $\is_f$ be the index of the functionality round in $\pi$ in which $\comStar$ calls the ideal functionality. Denote $\widehat{\comStar} = \append(\comStar)$ and hence, by construction, $\size{\widehat{\comStar} \setminus \cs } =c \log n$. Finally, recall that $\ct \subseteq [n]\setminus \cs$ and that $\AdvpiI$ corrupts the parties in $[n]\setminus \ct$.

Let $B_\cT$ denote the bias of the output of the honest parties in $\pi$, in a random execution of $\pi$ in which $\AdvpiI$ controls the parties in {$[n] \setminus \cT$}. Let $V$ denote the view of the parties in $\cs$ up to the $\is_f$ functionality round, including the answers for the calls made in this round by the committees appearing before $\comStar$ in the (arbitrary) order of $\psiS$. It holds that the expected bias of the honest parties can be written as follows:

\begin{align}
\eex{\cT}{\eex{v}{\ebv}} & = \eex{v}{\eex{\cT}{\ebv}} \\
& = \eex{v}{\eex{\cT \mid \cT \cap \widehat{\comStar} = \emptyset}{\ebv}\cdot \pr{\cT \cap \widehat{\comStar} = \emptyset \mid V = v}} \nonumber\\
& + \eex{v}{\eex{\cT \mid \cT \cap \widehat{\comStar} \neq \emptyset}{\ebv}\cdot \pr{\cT \cap \widehat{\comStar} \neq \emptyset \mid V = v}} \nonumber\\
& = \eex{v}{\eex{\cT \mid \cT \cap \widehat{\comStar} = \emptyset}{\ebv}}\cdot \pr{\cT \cap \widehat{\comStar} = \emptyset} \nonumber\\
& + \eex{v}{\eex{\cT \mid \cT \cap \widehat{\comStar} \neq \emptyset}{\ebv}}\cdot \pr{\cT \cap \widehat{\comStar} \neq \emptyset}.\nonumber	
\end{align}
The above expectations are with respect to $v \la V$ and a {uniformly distributed} set $\cT \la [n]\setminus \cs$ of size $(1-\beta)n$. The first equality follows by the fact that $V$ is independent of the choice of $\cT$ (since $\cT$ is always disjoint of $\cs$). The last equality holds since, by construction, the size of ${\widehat{\comStar} \setminus \cs}$ is fixed, and thus the event $\cT \cap \widehat{\comStar} = \emptyset$ is independent of $V$.

By the definition of the two-party protocol $\psiS$ and the attacker $\AdvpiI$, {and due to the fact that $\widehat{\comStar}$ has a fixed number of elements in $[n] \setminus \cs$,} it holds that the expected bias of the honest parties in $\pi$ conditioned on $\widehat{\comStar}$ being all corrupted, $\eex{v}{\eex{\cT \mid \cT \cap \widehat{\comStar} = \emptyset}{\ebv}}$, is equal to the bias of the honest party in $\psiS$ in an execution with $\Advpsi$. Hence,
\begin{align}\label{eq:small_committees_attack:3}
	\eex{v}{\eex{\cT \mid \cT \cap \widehat{\comStar} = \emptyset}{\ebv}} \geq 1/(16m + 4).
\end{align}

Since when $\cT \cap \widehat{\comStar} \neq \emptyset$, the attacker $\AdvpiI$ simply acts as an honest $\Party_0$ would, {and for similar reasons as in the case of \cref{eq:small_committees_attack:3},} the expectation $\eex{v}{\eex{\cT \mid \cT \cap \widehat{\comStar} \neq \emptyset}{\ebv}}$ is the expected bias of the output of the honest parties in an {\em honest execution} of $\pi$. Hence,
\begin{align}\label{eq:small_committees:1}
{\eex{v}{\eex{\cT \mid \cT \cap \widehat{\comStar} \neq \emptyset}{\ebv}} = 0}.
\end{align}
	

%

\noindent
It follows that
\begin{align}
\eex{\cT}{\eex{v}{\ebv}} = \frac{1}{16m + 4} \cdot \pr{\cT \cap \widehat{\comStar} = \emptyset}.
\end{align}

\noindent
Finally, a simple counting argument yields that:
\begin{align}
\pr{\cT \cap \widehat{\comStar} = \emptyset} = \frac{\binom{\left(2\beta - 1\right)n}{c \log n}}{\binom{\beta n}{c \log n}} \geq \left(\frac{2\beta - 1}{\beta e}\right)^{c \log n}.
\end{align}

\noindent
Putting everything together, we conclude that
\begin{align}
	\eex{\cT}{\eex{v}{\ebv}} \geq \frac{1}{16m + 4} \left(\frac{2\beta - 1}{\beta e}\right)^{c \log n} \in \Omega ({n^{-c}}/{m}).
\end{align}

\noindent
In particular, there exists a set $\cT \in [n]{\setminus \cs}$ such that $\AdvpiI$ biases the output of the honest parties in $\pi$ by $\Omega ({n^{-c}}/{m})$.
\end{proof}

\subsubsection{Biasing Arbitrary Protocols}\label{sec:large_committees}

In this section, we extend the approach of \cref{sec:small_committees} to the case where there may be large committees (\ie larger than $c \log n$). Hence, proving \cref{thm:lowerbound}.

Loosely speaking, the two-party protocol and the attack of \cref{sec:small_committees} are adjusted in the following manner. The revised two-party protocol $\psiSJ$ also includes $\numcalls$ linear-size subsets of parties $\Jcoll$, each associated with a single functionality round. In the \ith functionality round of the emulated execution of $\pi$, parties $\Party_0$ and $\Party_1$ go on as if the parties in $\J_i$ abort all calls to the ideal functionality by committees that intersect with $\J_i$. The attacker $\AdvpiIprime$ we construct for $\pi$, corrupting all parties in $[n]\setminus \cT$, acts as the attacker from the no-large-committees case (\cref{adv:small_committeed}), with the following additional change: in the \ith functionality round, it instructs the parties in $\J_i$ to prematurely abort all calls made by committees it intersects with (we make sure $\Jcoll \subseteq [n] \setminus \cT$). Since the subsets $\Jcoll$ are all of linear size, with high probability the above strategy
will abort all calls made by committees of size larger than $c \log n$, essentially leaving us with a protocol with no large committees, and thus vulnerable to the strategy of \cref{adv:small_committeed}.

Moving to the formal proof, let $\beta' = (\beta - 1/2)/\numcalls$, and let {$c = \log \left( m(32m+10) \right)/\beta' \log n$}.
The following claim shows that there exists a collection $\Jcoll$ such that in a random execution of $\pi$, for every $i \in [\numcalls]$, subset $\J_i$ intersects all committees of size at least $c \log n$ of the \ith functionality round with high probability.

\begin{claim}\label{clm:good_subsets}
For $\numcalls$ distinct subsets $\overline{\J} = \Jcoll$, let $\pi_{\overline{\J}}$ denote the variant of $\pi$ in which in the \ith functionality round, the parties in $\J_i$ abort all functionality calls made by committees they take part in this functionality round. Let $E_{\overline{\J}}^i$ denote the event, defined \wrt a random execution of $\pi_{\overline{\J}}$, that a call to the trusted party is made by a committee $\committee$ of size larger than $c \log n$ in the $\ith$ functionality round, and $\committee \cap \cj_i = \emptyset$. Then, there exists a collection $\overline{\J} = \Jcoll\subseteq [n]$ of distinct $\beta' n$-size subsets, such that $\pr{E_{\overline{\J}}^i} \leq 1/(32m + 10)$ for every $i \in [\numcalls]$.
\end{claim}
\begin{proof}
The claim is proved using a simple probabilistic argument. Let $\overline{\J} = \Jcoll$ be a collection of $\numcalls$ disjoint $\beta'$-subsets of $[n]$ chosen uniformly at random. We show that $\pr{E_{\overline{\J}}^i} \leq 1/(32m + 10)$ for all $i\in [\numcalls]$, where the probability is taken over the selection of $\Jcoll$ and the randomness of the parties and the ideal functionality. The claim thus immediately follows.
	
Let $i\in [\numcalls]$ and let $\committee$ be a committee of size at least $c \log n$ in the \ith functionality round. The probability that $\J_i$ is disjoint of $\committee$ is bounded by:
\begin{align*}
	\pr{\committee \cap \J_i = \emptyset} & = \prod_{k=0}^{\size{\committee} - 1} \left(1 - \frac{\size{\J_i}}{n-k}\right) \\
	& \leq \left(1 - \frac{\size{\J_i}}{n}\right)^{\size{\committee}} \\
	& \leq e^{-\frac{\size{\J_i} \size{\committee}}{n}} \\
	& \leq e^{- \beta' c \log n} \\
	& = n^{-\beta' c}.
\end{align*}
Substituting $c$ with {$\log \left( m(32m + 10) \right) /\beta' \log n$}, we get that {$\pr{\committee \cap \J_i = \emptyset} \leq 1/(32m + 10) m$}. As the number of committees in the \ith round of calls is at most $m$ (recall that $m$ is the total running-time of $\pi$), by union bound over these committees we conclude that
\[
\pr{E_{\overline{\J}}^i} \leq \frac{1}{32m + 10}.
\]
\end{proof}


The two-party coin-flipping protocol is a variant of the two-party protocol considered in \cref{sec:small_committees}. Fix subsets $\Jcoll$ whose existence is guaranteed by \cref{clm:good_subsets}, denote $\J = \bigcup_{i\in [\numcalls]} \J_i$, and fix an arbitrary $(1 - \beta)n$-size subset $\cs\subseteq [n] \setminus \J$.
Let $\append(\committee)$ be defined as follows. If $\size{\committee} \leq c \log n$, then we define $\append(\committee)$ to be the committee $\committee'$ obtained from $\committee$ by padding it with arbitrary (but fixed) parties from $[n] \setminus \left( \cs \cup \J \right)$, such that $\size{\committee' \setminus \left( \cs \cup \J \right)} = c \log n$. Otherwise ($\size{\committee} > c \log n$), set $\append(\committee)$ to an arbitrary $(c \log n)$-size subset of $[n]\setminus (\cs \cup \J)$.
The two-party protocol $\psiSJ$ is defined as follows:

\begin{protocol}[Protocol $\psiSJ = (\Party_0,\Party_1)$]\label{prot:two_party_large}~
\begin{enumerate}
	\item
    Party $\Party_0$ controls the parties in $\I_0 = \cs \cup \J$ and $\Party_1$ controls the parties in $\I_1 = [n] \setminus \left( \cs \cup \J\right)$, in an emulation of $\pi$: in each communication round, $\Party_i$ (for $i \in \left\{0,1\right\}$) emulates internally all messages among parties in $\I_i$. Party $\Party_i$ also sends to $\Party_{1-i}$ all messages from parties in $\I_i$ to parties in $\I_{1-i}$ and over the broadcast channel.
	
	\item
    For $i = 1,\ldots, \numcalls$: let $\committee^{(i)}$ be the set of committees in the \ith functionality round of the emulation, and let
    \[
    \committee^{(i)}_\cap = \left\{\committee \in \committee^{(i)} \colon \committee \cap \J_i \neq \emptyset\right\}
    \quad \text{ and } \quad
    \committee^{(i)}_{\not\cap} = \left\{\committee \in \committee^{(i)} \colon \committee \cap \J_i = \emptyset\right\}.
    \]
    For each committee $\committee \in \committee^{(i)}_\cap$, the parties assume the call by $\committee$ is aborted by the parties in $\committee \cap \J_i$. Denote $\committee^{(i)}_{\not\cap} = \committee_1,\ldots,\committee_\numcomm$; then $\Party_0$ sends $\numcomm$ separate messages, one for each committee, each containing a uniform bit as the result of the call by that committee.
\end{enumerate}

\textbf{Output:} Party $\Party_0$ (resp., $\Party_1$) outputs the output value of an arbitrary non-aborted party in $\I_0$ (resp., $\I_1$) in $\pi$. 

\vspace{2ex}
\textbf{Dealing with abort:}
\begin{enumerate}
	\item
    \emph{Party $\Party_i$ aborts during a communication round:} party $\Party_{1-i}$ continues the emulation internally, as if all the parties in $\I_i$ aborted in the current round of $\pi$. If $\Party_1$ continues the emulation internally, it emulates all remaining calls to the ideal functionality by itself.
	
	\item
    \emph{Party $\Party_0$ aborts during a functionality round:} let $\committee_1,\ldots, \committee_\numcomm$ be the committees of the current functionality round, that are not aborted by the parties in $\J$, and let $\committee_j$ be the committee corresponding to the aborting round. Party $\Party_1$ continues the emulation by itself, as if the minimal-index party in $\committee_j$ aborted the call to the ideal functionality, and flips the remaining coins for $\committee_{j+1},\ldots,\committee_\numcomm$ internally. Then, $\Party_1$ decides on its output as follows:
    it chooses uniformly at random a $(1-\beta)n$-size subset $\cT \subseteq [n]\setminus \left(\cs \cup \J \cup \append(\committee_j) \right)$, and continues the emulation of $\pi$ internally as if all remaining parties outside $\cT$ aborted right after the call to the ideal functionality.
\end{enumerate}	
\end{protocol}

Note that in the \ith functionality round, party $\Party_0$ sends uniform bits only for committees that do not intersect with $\J_i$. For the other committees, the parties assume their calls were aborted by the parties in $\J_i$.

The proof of the theorem proceeds as follows. Similarly to \cref{sec:small_committees}, we turn to translate the guaranteed attacker against the two-party protocol $\psiSJ$ into an attacker on $\pi$. The attacker $\AdvpiIprime$ on $\pi$, follows closely the attacker from the no-large-committees case, with the following adjustment. To simulate the addition of the aborting subsets to the two-party protocol $\AdvpiIprime$, which corrupts the pre-designated aborting subsets $\Jcoll$, simply instructs each of them to abort in the functionality rounds dictated by $\psiSJ$. In the case where all calls by large committees are indeed aborted, we are essentially left with the no-large-committees case, and the proof then proceeds as in \cref{sec:small_committees}.

%

Let $\Advpsiprime$ be the attacker for $\psiSJ$ guaranteed by \cref{lem:simple_advM}. As in \cref{sec:small_committees}, we assume without loss of generality that $\Advpsiprime$ corrupts $\Party_0$ and might abort in round $\is$ of $\psiSJ$, that corresponds to a functionality round in $\pi$. Let $\comStar$ be the committee in $\pi$ that corresponds to round $\is$ in $\psi$, and let $\is_f$ be the index of the functionality round in $\pi$ in which the committee members of $\comStar$ invoke the ideal functionality.

{The main difference between the above attacker and the one we used in \cref{sec:small_committees}, is that in the above, if the size of the ``attacked committee'' is too large, \ie if $\size{\comStar}> c \log n$, then $\comStar$ is no longer a subset of $\append(\comStar)$. This seems to be problematic, since the attack on protocol $\pi$ given in \cref{sec:small_committees}, crucially relies on the fact that $\comStar \subseteq \append(\comStar)$. Fortunately, by the choice of the sets $\Jcoll$, the committee $\comStar$ will be large only with very small probability, and thus the resulting attack on $\pi$ will go through. We formalize this intuition by considering an attacker $\hAdvPsi$ on protocol $\psi$ that never aborts if the attacked committee is larger than $c \log n$, and still noticeably biases the protocol $\psi$.} Formally, the attacker $\hAdvPsi$ acts as ${\Adv}^\psi$, but with the following change: let $\comStar$ be committee corresponding to round $\is$ (in which ${\Adv}^\psi$ is about to abort). If $\size{\comStar} > c \log n$,
then $\hAdvPsi$ does not abort (even if $\Advpsiprime$ does). Recall that $\Advpsiprime$ might wish to abort the call of $\comStar$ only if $\comStar$ does not intersect $\J_{\is_f}$ (otherwise, $\Party_0$ does not send an answer bit to $\comStar$ in $\psi$). By \cref{clm:good_subsets}, if $\size{\comStar} > c \log n$, the latter happens with probability at most $1/(32m+10)$.
Combining this together with \cref{lem:simple_advM}, it holds that the bias of $\Party_1$ in an execution of $\psi$ with $\hAdvPsi$ is at least $1/(32m + 8)$.


For a $(1 - \beta) n$-size subset $\cT \subseteq [n]\setminus (\cs \cup \J)$, the attacker $\AdvpiIprime$ corrupting $[n] \setminus \cT$ is defined as follows:
	
{ \samepage
\begin{algorithm}[The fail-stop attacker $\AdvpiIprime$ for $\pi$]\label{adv:large_committees}~
\begin{enumerate}
	\item
    Start emulating $\hAdvPsi$, by forwarding to $\hAdvPsi$ the messages sent by the honest party in $\psi$ and the results of the calls to the ideal functionality.
	\item
    For $i = 1,\ldots, \is_f$: for each committee $\committee \in \committee^{(i)}$, instruct the parties in $\committee \cap \J_i$ to abort the call of $\committee$ to the ideal functionality.
	\item
	In the $\is_f$ invocation:
	\begin{enumerate}
		\item
        If $\append(\comStar) \subseteq [n] \setminus \cT$: wait until all other committees but $\comStar$ in the $\is_f$ round of calls are done. {Let $\committee_1,\ldots,\committee_\numcomm$ be the committees that call the ideal functionality in the $\is_f$'th functionality round of $\pi$, and correspond to rounds earlier than $\is$ in $\psiS$.} Give the result bits of $\committee_1, \ldots, \committee_\numcomm$ to $\hAdvPsi$ sequentially. If $\hAdvPsi$ chooses to abort, instruct the minimal-index party in $\comStar$ to abort the call to the ideal functionality, and in the next ordinary communication round of $\pi$, instruct all parties in $[n] \setminus \cT$ to abort. If $\hAdvPsi$ does not abort, complete the honest execution of $\pi$.
		\item
		If $\append(\comStar) \not\subseteq [n] \setminus \cT$: continue in the honest execution of $\pi$.
	\end{enumerate}
	\item
    If not aborted after invocation $\is_f$, then for $i = \is_f+1,\ldots, \numcalls$: for each committee $\committee \in \committee^{(i)}$, instruct the parties in $\committee \cap \J_i$ to abort the call of $\committee$ to the ideal functionality.
\end{enumerate}
\end{algorithm}
}

The success bias obtained by $\AdvpiIprime$ is analyzed
in following lemma, which immediately yields the proof of \cref{thm:lowerbound}.

\begin{lemma}\label{lem:large_committees_attack}
There exists a $(1 - \beta) n$-subset $\cT \subseteq [n] \setminus \left( \cs \cup \J \right)$ such that $\AdvpiIprime$ biases the output of the honest parties in $\pi$ by $\Omega(n^{-c}/m)$.
\end{lemma}

\begin{proof}
The following proof follows similar lines to the proof of \cref{lem:small_committees_attack}.

Recall our notation: let $\is$ be the round in $\psiS$ in which $\hAdvPsi$ might abort, let ${\comStar}$ be the committee in $\pi$ that corresponds to round $\is$ in $\psiSJ$, and let $\is_f$ be the index of the functionality round in $\pi$ in which the committee members of ${\comStar}$ call the ideal functionality. Denote $\widehat{\comStar} = \append(\comStar)$ and hence, by construction, $\ssize{\widehat{\comStar} \setminus \left( \cs \cup \J \right)} = c \log n$. Finally, recall that $\ct \subseteq [n]\setminus (\cs \cup \cj)$ and that $\AdvpiI$ corrupts the parties in $[n]\setminus \ct$.

Let $B_\cT$ denote the bias of the output of the honest parties in $\pi$, in a random execution of $\pi$ in which $\AdvpiI$ controls the parties in {$[n] \setminus \cT$}. Let $V$ denote the view of the parties in $\cs \cup \J$ up to the $\is_f$ functionality round, including the answer of the calls made in this round by the committees appearing before $\comStar$ in the (arbitrary) order of $\psiSJ$. As in the proof of \cref{lem:small_committees_attack}, it holds that
\begin{align}
\eex{\cT}{\eex{v}{\ebv}}
& = \eex{v}{\eex{\cT \mid \cT \cap \widehat{\comStar} = \emptyset}{\ebv}}\cdot \pr{\cT \cap \widehat{\comStar} = \emptyset} \\
& + \eex{v}{\eex{\cT \mid \cT \cap \widehat{\comStar} \neq \emptyset}{\ebv}}\cdot \pr{\cT \cap \widehat{\comStar} \neq \emptyset}.	\nonumber
\end{align}
The above expectations are with respect to $v \la V$ and $\cT \la [n]\setminus \left( \cs \cup \J \right)$ of size $(1-\beta)n$.

As in the proof of \cref{lem:small_committees_attack}, it holds that $\eex{v}{\eex{\cT \mid \cT \cap \widehat{\comStar} = \emptyset}{\ebv}}$ is the bias of the honest party in $\psiSJ$ in an execution with {$\hAdvPsi$}. {Note it might be the case that $\cT \cap \widehat{\comStar} = \emptyset$, but $\cT \cap \comStar \neq \emptyset$ (if $\size{\comStar} > c \log n$, and hence $\widehat{\comStar}$ is an arbitrary $(c \log n)$-size subset of $[n]\setminus (\cs \cup \J)$), so $\AdvpiIprime$ cannot abort the call by committee $\comStar$. But this is benign, since in case $\size{\comStar} > c \log n$, $\hAdvPsi$ does not abort.} Hence,

\begin{align}
\eex{v}{\eex{\cT \mid \cT \cap \widehat{\comStar} = \emptyset}{\ebv}} \geq 1/(32m + 8).
\end{align}

Similarly, $\eex{v}{\eex{\cT \mid \cT \cap \widehat{\comStar} \neq \emptyset}{\ebv}}$ is the bias achieved by an attacker that always instructs the parties in $\Jcoll$ to abort (subset $\cj_i$ in the \ith functionality round), but nothing else. We assume \wlg that
\begin{align}\label{eq:large_committees:1}
\eex{v}{\eex{\cT \mid \cT \cap \widehat{\comStar} \neq \emptyset}{\ebv}} < c' \cdot n^{-c}/{m},
\end{align}
for $c'>0$ to be determined later by the analysis, as otherwise, an attacker corrupting $\Jcoll$ biases the output of the honest parties in $\pi$ by $c' \cdot n^{-c}/{m}$, concluding the lemma proof.

A simple counting argument yields that
\begin{align}
\pr{\cT \cap \widehat{\comStar} = \emptyset} \geq \left(\frac{2\beta - 1}{\beta e}\right)^{c \log n}.
\end{align}

Putting everything together, we conclude that
\begin{align*}
\eex{\cT}{\eex{v}{\ebv}} \geq \frac{1}{32m + 8}\left(\frac{2\beta - 1}{e}\right)^{c \log n} - c' \cdot n^{-c}/{m},
\end{align*}
which is in $\Omega({n^{-c}}/{m})$ for small enough $c'$. It follows that there exists a set $\cT \in [n] \setminus (\cs \cup \cj)$ such that $\AdvpiIprime$ biases the output of the honest parties in $\pi$ by $\Omega({n^{-c}}/{m})$.
\end{proof}

\subsection{A Single-Aborting-Round Variant of Cleve's Attacker}\label{sec:impo:Cleve}
\citet{Cleve86} showed that for every (plain-model) $r$-round, $\consistent$-consistent, two-party coin-flipping protocol, there exists an efficient fail-stop attacker that by corrupting one of the parties, can bias the output of the other party by at least $\consistent/(4r+1)$. In more detail, let $\pi = (\Party_0, \Party_1)$ be a two-party, $r$-round, $\consistent$-consistent coin-flipping protocol. The attacker $\Adv$ guaranteed by~\cite{Cleve86} is defined as follows: Assume for concreteness that $\Adv$ is corrupting $\Party_0$ and would like to bias the output of $\Party_1$ towards one. The attacker $\Adv$ follows the protocol honestly up to some pre-defined round $\is \in [r]$, in which $\Party_0$ is about to send a message. Then, before sending the message, it examines the value of its output in the case that $\Party_1$ aborts right after receiving the message; denote this value by $b_{\is}$. If $b_{\is} = 0$, the attacker $\Adv$ instructs $\Party_0$ to abort in round $\is$; otherwise, $\Adv$ instructs $\Party_0$ to honestly send its $\is$'th message and abort right after doing that.

The following lemma establishes the existence of an even simpler kind of a fail-stop attacker, one that might abort only in a single pre-defined round.

\begin{lemma}[Restatement of \cref{lem:simple_advM}]\label{lem:simple_adv}
Let $\pi = \left(\Party_0,\Party_1\right)$ be a two-party, $r$-round, $\consistent$-consistent coin-flipping protocol. Then, there exist a party $\Party \in \left\{\Party_0,\Party_1\right\}$, a round index $\is \in [r]$, an index $\js \in\sset{\is,\is-1}$, and a bit $b \in \zo$ such that the following holds. Consider the fail-stop attacker $\Adv'$, that corrupting $\Party$, instructs $\Party$ to abort in round $\is$ if $b_{\js} = b$. Then, $\Adv'$ biases the output of the other party by $\consistent/(8r + 2)$.
\end{lemma}

\begin{proof}
Let $\Adv$ be the fail-stop attacker for $\pi$ guaranteed by~\cite{Cleve86}: $\Adv$ controls party $\Party \in \set{\Party_0, \Party_1}$, and aborts in round $\js$ if $b_{\js} = b$ for some $b \in \zo$; otherwise $\Adv$ aborts in round $\js+1$. {(The attacker guaranteed by~\cite{Cleve86} might also instruct $\Party_0$ to abort in the very first round; it is easy to see, however, that this attacker also yields an attacker of the type stated in the lemma.)} Assume for concreteness that $\Party = \Party_0$ and $b=0$. Consider the following two adversaries, both corrupting $\Party_0$:
\begin{itemize}
	\item
    Attacker $\Adv_0$ follows the protocol honestly up to round $\js$. Then, if $b_{\js} = 0$, it aborts in round $\js$; otherwise, it continues with the execution of the protocol honestly.
	
	\item
    Attacker $\Adv_1$ follows the protocol honestly up to round $\js+1$. Then, if $b_{\js} = 1$, it aborts in round $\js+1$; otherwise, it continues with the execution of the protocol honestly.
\end{itemize}

We prove that the average of the bias the two adversaries achieve is half the bias of $\Adv$. Recalling that the bias of $\Adv$ is at least $\consistent/(4r+1)$, this will complete the proof. Let $O$ be a random variable denoting the output of $\Party_1$ in $\left(\Adv, \Party_1\right)$. Also, denote by $Q_j$ the event in which $\Party_0$ aborted in round~$j$, and by $\neg Q$ the event in which $\Party_0$ did not abort.
The bias achieved by $\Adv$ can be written as:
\begin{align}
B = \pr{b_{\js} = 0} \cdot \ex{O \mid b_{\js} = 0 \wedge Q_{\js}} + \pr{b_{\js} = 1} \cdot \ex{O \mid b_{\js} = 1 \wedge Q_{\js+1}} - \half.
\end{align}

\noindent
Similarly, we can write the bias of $\Adv_0$ and of $\Adv_1$ as:
\begin{align*}
B_0 & = \pr{b_{\js} = 0} \cdot \ex{O \mid b_{\js} = 0 \wedge Q_{\js}} + \pr{b_{\js} = 1} \cdot \ex{O \mid \neg Q} -\half. \\
B_1 & = \pr{b_{\js} = 1} \cdot \ex{O \mid b_{\js} = 0 \wedge Q_{\js+1}} + \pr{b_{\js} = 0} \cdot \ex{O \mid \neg Q} - \half.
\end{align*}
Taking the average of $B_0$ and $B_1$ yields that:
\begin{align*}
\half\left(B_0 + B_1\right) & = \half\left(\pr{b_{\js} = 0} \cdot \ex{O \mid b_{\js} = 0 \wedge Q_{\js}} + \pr{b_{\js} = 1} \cdot \ex{O \mid b_1 = 0 \wedge Q_{\js+1}} - \half\right) \\
& + \half\left(\pr{b_{\js} = 1} \cdot \ex{O \mid \neg Q} + \pr{b_{\js} = 0} \cdot \ex{O \mid \neg Q} - \half \right) \\
& = \half B + \half \left(\ex{O \mid \neg Q} - \half \right) \\
& = \half B.
\end{align*}
The last equality follows by the fact that $\ex{O \mid \neg Q} = \half$ (the expected output of the honest party in an honest execution of $\pi$).
\end{proof}

{\small{
\bibliographystyle{abbrvnat}
\bibliography{crypto}
}}

\appendix

\section{Fairness with Restricted Identifiable Abort (No Inputs)}\label{sec:restricted_idabort_noinput}

We start by introducing a variant of fairness with identifiable abort that will be used as an intermediate step in our constructions. This definition captures the delegation of the computation to smaller committees that independently carry out the (same) fair computation, such that the adversary can only abort the computation of committees with corrupted parties. Looking ahead, in the honest-majority setting we will consider a vector of committees $\vCS=(\committee_1,\ldots,\committee_\ell)$ that run in parallel, whereas in the no-honest-majority case we will consider a single committee $\committee$ (\ie $\ell=1$).

\paragraph{Ideal model with fairness and restricted identifiable abort (no inputs).}
An ideal computation, with $\vCS$-identifiable-abort, of a no-input, $n$-party functionality $f$ for parties $(\Party_1,\ldots,\Party_n)$ \wrt $\vCS=(\committee_1,\ldots,\committee_\ell)$, where $\committee_1,\ldots,\committee_\ell\subseteq [n]$, in the presence of an ideal-model adversary $\Adv$ controlling the parties indexed by $\IS\subseteq[n]$, proceeds via the following steps.
\begin{itemize}
	\item[\emph{Early abort}:]
    $\Adv$ can abort the computation by choosing an index of a corrupted party $\is_j\in\IS \cap \committee_j$ for every $j\in [\ell]$ and sending the abort message $(\abort,\sset{\is_1,\ldots,\is_\ell})$ to the trusted party. In case of such abort, the trusted party sends the message $(\bot,\sset{\is_1,\ldots,\is_\ell})$ to all parties and halts.

    \item[\emph{Trusted party answers the parties}:]
    If $\committee_j\subseteq\IS$ for some $j\in[\ell]$, the trusted party receives from the adversary output values $y'_1,\ldots,y'_n$ and sends $y'_i$ to $\Party_i$.
    Otherwise, the trusted party uniformly samples random coins $r$, computes $(y_1, \ldots, y_n)=f(\lambda,\ldots,\lambda;r)$, and sends $y_i$ to party $\Party_i$ for every $i\in[n]$.

    \item[\emph{Outputs}:]
    Honest parties always output the message received from the trusted party and the corrupted parties output nothing.
    The adversary $\Adv$ outputs an arbitrary function of the messages received by the corrupted parties from the trusted party and its auxiliary input.
\end{itemize}

\begin{definition}[ideal-model computation with fairness and restricted identifiable abort (no inputs)]\label{def:ideal_ridabort_noinput}
Let $f\colon(\zs)^n \mapsto (\zs)^n$ be a no-input, $n$-party functionality, let $\IS\subseteq [n]$, and let $\vCS=(\committee_1,\ldots,\committee_\ell)$, where $\committee_1,\ldots,\committee_\ell\subseteq [n]$.
The {\sf joint execution of $f$ with $\vCS$ under $(\Adv, I)$ in the ideal model}, on auxiliary input $\aux$ to $\Adv$ and security parameter $\secParam$, denoted $\IDEAL^{\vCS\mhyphen\idfair}_{f,\IS,\Adv(\aux)}(\secParam)$, is defined as the output vector of $\Party_1, \ldots, \Party_n$ and $\Adv(\aux)$ resulting from the above described ideal process.
\end{definition}

\begin{definition}\label{def:SecureProtocol_ridfair_noinput}
Let $f\colon(\zs)^n \mapsto (\zs)^n$ be a no-input, $n$-party functionality and let $\pi$ be a probabilistic polynomial-time protocol computing $f$.
The {\sf protocol $\pi$ $(\delta,t)$-securely computes $f$ with fairness and $(\ell, n',t')$-identifiable-abort (and computational security)}, if for every probabilistic polynomial-time real-model adversary \Adv, there exists a probabilistic polynomial-time adversary $\Sim$ for the ideal model, such that for every $\IS\subseteq [n]$ of size at most $t$ and subsets $\vCS=(\committee_1,\ldots,\committee_\ell)$ satisfying $\committee_j\subseteq [n]$, $\ssize{\committee_j}=n'$, and $\ssize{\IS\cap\committee_j}\leq t'$, for every $j\in[\ell]$, it holds that
\[
\set{\bigbrack \REAL_{\pi, \IS, \Adv(\aux)}(\secParam)}_{\aux\in\zs, \secParam\in\N}
\deltaci
\set{\bigbrack \IDEAL^{\vCS\mhyphen\idfair}_{f, \IS, \Sim(\aux)}(\secParam)}_{\aux\in\zs, \secParam\in\N}.
\]
If $\delta$ is negligible, we say that $\pi$ is a protocol that $t$-securely computes $f$ with fairness and $(\ell, n',t')$-identifiable-abort and computational security. We denote fairness with $(1,n',t')$-identifiable-abort by fairness with $(n',t')$-identifiable-abort.

The {\sf protocol $\pi$ $(\delta,t)$-securely computes $f$ with fairness and $(\ell, n',t')$-identifiable-abort (and information-theoretic (statistical) security)}, if for every real-model adversary \Adv, there exists an adversary $\Sim$ for the ideal model, whose running time is polynomial in the running time of $\Adv$, such that for every $\IS\subseteq [n]$ of size at most $t$, and subsets $\vCS=(\committee_1,\ldots,\committee_\ell)$ as above,
it holds that
\[
\set{\bigbrack \REAL_{\pi, \IS, \Adv(\aux)}(\secParam)}_{\aux\in\zs, \secParam\in\N}
\deltaclose
\set{\bigbrack \IDEAL^{\vCS\mhyphen\idfair}_{f, \IS, \Sim(\aux)}(\secParam)}_{\aux\in\zs, \secParam\in\N}.
\]
If $\delta$ is negligible, we say that $\pi$ is a protocol that $t$-securely computes $f$ with fairness and $(\ell, n',t')$-identifiable-abort and statistical security.
\end{definition}
\enote{This definition should be united with the with inputs one, and be moved to the preliminaries.}

\end{document}